%% file: maj3.tex
\documentclass[10pt]{article}

\usepackage{fullpage}
\usepackage{amsfonts, amssymb, amsmath, amsthm}
\usepackage{latexsym}
\usepackage{url}
\usepackage[colorlinks=false,breaklinks]{hyperref}
\usepackage{graphicx}
\usepackage{color}
\usepackage{subfigure}
\usepackage{cancel}
\usepackage{multirow}

\usepackage{import}
\usepackage{ltxtable}

\usepackage{calc}	

\newcommand{\bra}[1]{\langle#1|}
\newcommand{\ket}[1]{|#1\rangle}
\newcommand{\braket}[2]{\langle#1|#2\rangle}
\newcommand{\ketbra}[2]{\ket{#1}\!\bra{#2}}
\newcommand{\abs}[1]{\lvert #1\rvert}
\newcommand{\norm}[1]{\| #1 \|}
\newcommand{\eps}{{\epsilon}}

\newcommand{\ADV} {\mathrm{Adv}}
\newcommand{\gateset}{\mathcal{S}}

\DeclareMathOperator{\MAJ}{{\operatorname{MAJ}_3}}
\DeclareMathOperator{\EQUAL}{{\operatorname{EQUAL}}}
\DeclareMathOperator{\EXACT}{{\operatorname{EXACT}}}

\def\algorithm {{\mathcal{ALG}}}
\def\adjoint{\dagger}
\def\H {A}
\def\o {o}
\def\Deltam {{\overline{\Delta}}}
\def\Pim {{\overline{\Pi}}}
\def\VV {\widehat{V}}

\newcommand{\ri} {r} 
\newcommand{\rj} {\tilde{r}}
\newcommand{\si} {s}
\newcommand{\sj} {\tilde{s}}
\newcommand{\Ui} {z} 
\newcommand{\Uj} {\tilde{z}}
\newcommand{\abi} {\sigma} 
\newcommand{\abj} {\tilde{\sigma}}

\def\ACJ {{C}} 
\def\cp {{{y}}} 
\def\cm {{\bar{y}}} 
\def\cpp {{{Y}}} 
\def\cmm {{\bar{Y}}} 
\def\dcm {{(\Deltam \cm)}} 
\def\dcmm {{\cmm_\Deltam}} 
\def\DeltaXDelta  {{X_\Delta}} 
\def\DeltamXDeltam  {{X_\Deltam}} 
\def\XDeltaSchur {{( X / \DeltaXDelta )}}

\def\resinv {{-1}} 

\DeclareMathOperator{\Span}{\operatorname{Span}}
\DeclareMathOperator{\Range}{\operatorname{Range}}
\DeclareMathOperator{\poly}{\operatorname{poly}}

\DeclareMathOperator{\wsizeop}{{\operatorname{wsize}}}		
\newcommand{\wsize}[1]{{\wsizeop({#1})}}

\newcommand{\wsizex}[2]{{\wsizeop({#1},{#2})}}
\newcommand{\wsizexS}[3]{{\wsizeop_{#3}({#1},{#2})}}

\newcommand{\hugelpar}[1]{\left(\vbox to #1{}\right.}
\newcommand{\hugerpar}[1]{\left.\vbox to #1{}\right)}
\newcounter{sprows}
\newcounter{spcols}
\newlength{\spheight}
\newlength{\spraise}
\newcommand{\spleft}[2][0pt]{\multirow{\value{sprows}}{*}{%
	\vbox to \spraise{\vss\hbox{$#2 \hugelpar{\spheight}\hskip -#1$}\vss}}}
\newcommand{\spright}[2][0pt]{\multirow{\value{sprows}}{*}{%
	\vbox to \spraise{\vss\hbox{\hskip -#1 $\hugerpar{\spheight} #2$}\vss}}}

\newcommand{\comment}[1]{\emph{\color{blue}Comment:\color{black} #1}} 
\newlength{\commentslength}
\newcommand{\comments}[1]{
\hspace{-2\parindent}
\addtolength{\commentslength}{-\commentslength}
\addtolength{\commentslength}{\linewidth}
\addtolength{\commentslength}{-\parindent}
\fcolorbox{blue}{white}{\smallskip\begin{minipage}[c]{\commentslength}
\emph{Comments:}\begin{itemize}#1\end{itemize}\end{minipage}}\bigskip
}
\renewcommand{\comment}[1]{}\renewcommand{\comments}[1]{}
\newcommand{\rem}[1]{}

\newtheorem{theorem}{Theorem}[section]
\newtheorem{lemma}[theorem]{Lemma}

\newtheorem{claim}[theorem]{Claim}

\newtheorem{definition}[theorem]{Definition}
\newtheorem{remark}[theorem]{Remark}
\newtheorem{example}[theorem]{Example}

\newfont{\subsubsecfnt}{ptmri8t at 10pt}
\renewcommand{\subparagraph}[1]{\smallskip{\subsubsecfnt #1.}}

\numberwithin{equation}{section} 

\newcommand{\eqnref}[1]{\hyperref[#1]{{(\ref*{#1})}}}
\newcommand{\thmref}[1]{\hyperref[#1]{{Theorem~\ref*{#1}}}}
\newcommand{\lemref}[1]{\hyperref[#1]{{Lemma~\ref*{#1}}}}
\newcommand{\corref}[1]{\hyperref[#1]{{Corollary~\ref*{#1}}}}
\newcommand{\defref}[1]{\hyperref[#1]{{Definition~\ref*{#1}}}}
\newcommand{\secref}[1]{\hyperref[#1]{{Section~\ref*{#1}}}}
\newcommand{\figref}[1]{\hyperref[#1]{{Figure~\ref*{#1}}}}
\newcommand{\tabref}[1]{\hyperref[#1]{{Table~\ref*{#1}}}}
\newcommand{\remref}[1]{\hyperref[#1]{{Remark~\ref*{#1}}}}
\newcommand{\appref}[1]{\hyperref[#1]{{Appendix~\ref*{#1}}}}
\newcommand{\claimref}[1]{\hyperref[#1]{{Claim~\ref*{#1}}}}
\newcommand{\exampleref}[1]{\hyperref[#1]{{Example~\ref*{#1}}}}

\allowdisplaybreaks[1]

\begin{document}

\title{Span-program-based quantum algorithm for evaluating formulas}
\author{%
Ben W.~Reichardt%
  \thanks{School of Computer Science and Institute for Quantum Computing, University of Waterloo.  Part of the work conducted while at the Institute for Quantum Information, California Institute of Technology, supported by NSF Grants CCF-0524828 and PHY-0456720, and by ARO Grant W911NF-05-1-0294.  Email: {\tt breic@uwaterloo.ca}}
\and
Robert \v Spalek%
  \thanks{Google, Inc.
  Part of the work conducted while at the University of California, Berkeley, supported by NSF Grant CCF-0524837 and ARO Grant DAAD 19-03-1-0082.  Email: {\tt spalek@google.com}}
}
\date{}

\maketitle

\begin{abstract}
We give a quantum algorithm for evaluating formulas over an extended gate set, including all two- and three-bit binary gates (e.g., NAND, 3-majority).  The algorithm is optimal on read-once formulas for which each gate's inputs are balanced in a certain sense.

The main new tool is a correspondence between a classical linear-algebraic model of computation, ``span programs," and weighted bipartite graphs.  A span program's evaluation corresponds to an eigenvalue-zero eigenvector of the associated graph.  A quantum computer can therefore evaluate the span program by applying spectral estimation to the graph.  

For example, the classical complexity of evaluating the balanced ternary majority formula is unknown, and the natural generalization of randomized alpha-beta pruning is known to be suboptimal.  In contrast, our algorithm generalizes the optimal quantum AND-OR formula evaluation algorithm and is optimal for evaluating the balanced ternary majority formula.  
\end{abstract}

%
%

\section{Introduction}
\nobreak

A formula $\varphi$ on gate set $\gateset$ and of size $N$ is a tree with $N$ leaves, such that each internal node is a gate from $\gateset$ on its children.  
The read-once formula evaluation problem is to evaluate $\varphi(x)$ given oracle access to the input string $x = x_1 x_2 \ldots x_N$.  
An optimal, $O(\sqrt N)$-query quantum algorithm is known to evaluate ``approximately balanced" formulas over the gates $\gateset = \{ \text{AND, OR, NOT} \}$~\cite{AmbainisChildsReichardtSpalekZhang07andor}.  
We extend the gate set $\gateset$.  
We develop an optimal quantum algorithm for evaluating balanced, read-once formulas over a gate set $\gateset$ that includes arbitrary three-bit gates, as well as bounded fan-in EQUAL gates and bounded-size $\{ \text{AND, OR, NOT, PARITY} \}$ formulas considered as single gates.   The correct notion of ``balanced" for a formula including different kinds of gates turns out to be ``adversary-balanced," meaning that the inputs to a gate must have exactly equal adversary lower bounds.  
The definition of ``adversary-balanced" formulas also includes as a special case layered formulas in which all gates at a given depth from the root are of the same type.  

The idea of our algorithm is to consider a weighted graph $G(\varphi)$ obtained by replacing each gate of the formula $\varphi$ with a small gadget subgraph, and possibly also duplicating subformulas.  \figref{f:example} has several examples.
We relate the evaluation of $\varphi$ to the presence or absence of small-eigenvalue eigenvectors of the weighted adjacency matrix $A_{G(\varphi)}$ that are supported on the root vertex of $G(\varphi)$.  The quantum algorithm runs spectral estimation to either detect these eigenvectors or not, and therefore to evaluate $\varphi$.  

As a special case, for example, our algorithm implies: 
\begin{theorem} \label{t:balancedmajority}
A balanced ternary majority $(\MAJ)$ formula of depth $d$, on $N = 3^d$ inputs, can be evaluated by a quantum algorithm with bounded error using $O(2^d)$ oracle queries, which is optimal.  
\end{theorem}
\noindent The classical complexity of evaluating this formula is known only to lie between $\Omega((7/3)^d)$ and $o((8/3)^d)$, and the previous best quantum algorithm, from~\cite{AmbainisChildsReichardtSpalekZhang07andor}, used $O(\sqrt{5}^d)$ queries.  

\begin{figure}
\centering
\subfigure{\label{f:notsubstitution}\includegraphics[scale=.52]{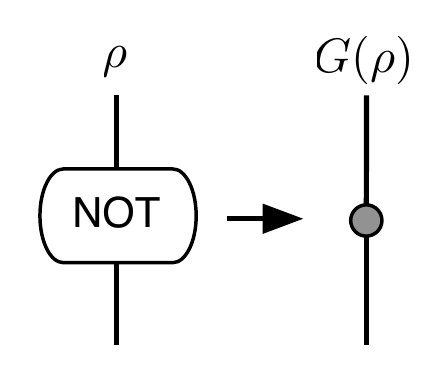}}
\quad\;
\subfigure{\includegraphics[scale=.52]{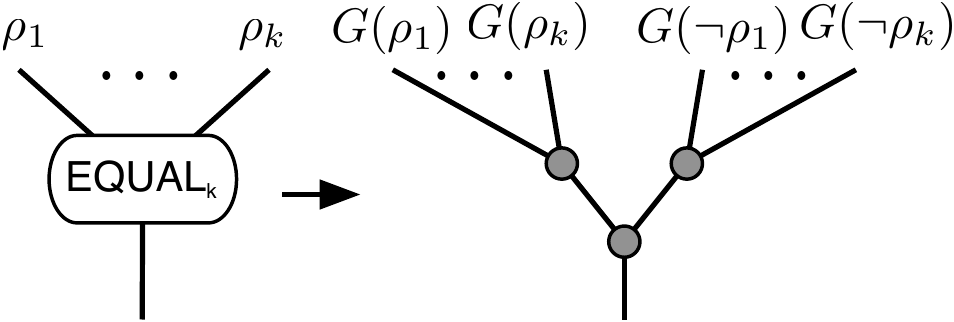}}
\quad\;
\subfigure{\includegraphics[scale=.52]{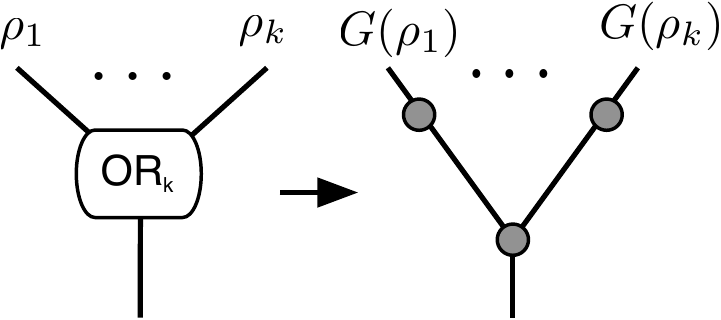}}
\quad\;
\subfigure{\includegraphics[scale=.52]{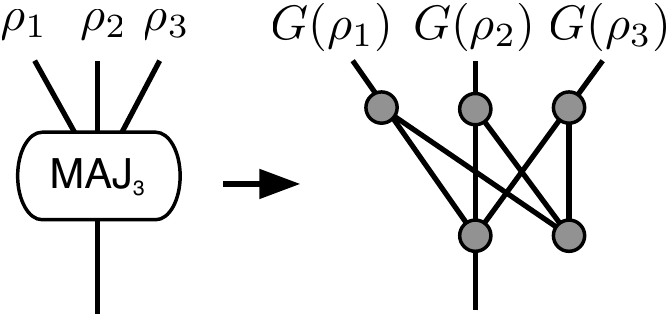}}
\caption{To convert a formula $\varphi$ to the corresponding graph $G(\varphi)$, we recursively apply substitution rules starting at the root to convert each gate into a gadget subgraph.  Some of the rules are shown here, except with the edge weights not indicated.  The dangling edges at the top and bottom of each gadget are the input edges and output edge, respectively.  To compose two gates, the output edge of one is identified with an input edge of the next (see \figref{f:parityandmaj3}).} \label{f:example}
\end{figure}

The graph gadgets themselves are derived from ``span programs"~\cite{KarchmerWigderson93span}.  Span programs have been used in classical complexity theory to prove lower bounds on formula size~\cite{KarchmerWigderson93span, BabaiGalWigderson99superpolyspanprogram} and monotone span programs are related to linear secret-sharing schemes~\cite{BeimelGalPaterson96spanprogram}.  (Most, though not all~\cite{AllenderBealsOgihara99span}, applications are over finite fields, whereas we use the definition over $\bf{C}$.)  We will only use compositions of constant-size span programs, but it is interesting to speculate that larger span programs could directly give useful new quantum algorithms.

\paragraph{Classical and quantum background}

The formula evaluation problem has been well-studied in the classical computer model.  
Classically, the case $\gateset = \{ \text{NAND} \}$ is best understood.  A formula with only NAND gates is equivalent to one with alternating levels of AND and OR gates, a so-called ``AND-OR formula," also known as a two-player game tree.  One can compute the value of a balanced binary AND-OR formula with zero error in expected time $O(N^{\log_2[(1 + \sqrt{33})/4]}) = O(N^{0.754})$~\cite{snir:dec, sw:and-or}, and this is optimal even for bounded-error algorithms~\cite{santha:and-or}.  However, the complexity of evaluating balanced AND-OR formulas grows with the degree of the gates.  For example, in the extreme case of a single OR gate of degree $N$, the complexity is $\Theta(N)$.  The complexity of evaluating AND-OR formulas that are not ``well-balanced" is unknown.  

If we allow the use of a quantum computer with coherent oracle access to the input, however, then the situation is much simpler; between $\Omega(\sqrt N)$ and $N^{\frac12+o(1)}$ queries are necessary and sufficient to evaluate {\emph{any}} $\{ \text{AND, OR, NOT} \}$ formula with bounded error.  
In one extreme case, Grover search~\cite{grover:search, grover:search-time} evaluates an OR gate of degree $N$ using $O(\sqrt N)$ oracle queries and $O(\sqrt{N} \log \log N)$ time.  
In the other extreme case, Farhi, Goldstone and Gutmann recently devised a breakthrough algorithm for evaluating the depth-$\log_2 N$ balanced binary AND-OR formula in $O(\sqrt N)$ time in the unconventional Hamiltonian oracle model~\cite{fgg:and-or}.  
Ambainis~\cite{ambainis07nand} improved this to $O(\sqrt{N})$-queries in the standard query model.
Childs, Reichardt, {\v S}palek and Zhang~\cite{ChildsReichardtSpalekZhang07andor} gave an $O(\sqrt N)$-query algorithm for evaluating balanced or ``approximately balanced" formulas, and extended the algorithm to arbitrary $\{\text{AND, OR, NOT}\}$ formulas with $N^{\frac12 + o(1)}$ queries, and also $N^{\frac12 + o(1)}$ time after a preprocessing step.  (Ref.~\cite{AmbainisChildsReichardtSpalekZhang07andor} contains the merged results of~\cite{ambainis07nand, ChildsReichardtSpalekZhang07andor}.)

This paper shows other nice features of the formula evaluation problem in the quantum computer model.  
Classically, with the exception of $\{\text{NAND}\}$, $\{\text{NOR}\}$ and a few trivial cases like $\{\text{PARITY}\}$, most gate sets are poorly understood.  
In 1986, Boppana asked the complexity of evaluating the balanced, depth-$d$ ternary majority ($\MAJ$) function~\cite{sw:and-or}, and today the complexity is only known to lie between $\Omega((7/3)^d)$ and 
$O((2.6537\ldots)^d)$~\cite{JayramKumarSivakumar03majority}.
In particular, the na{\" i}ve generalization of randomized alpha-beta pruning---recursively evaluate two random immediate subformulas and then the third if they disagree---runs in expected time $O((8/3)^d)$ and is suboptimal.  
This suggests that the balanced ternary majority function is significantly different from the balanced $k$-ary NAND function, for which randomized alpha-beta pruning is known to be optimal.  
In contrast, we show that the optimal {quantum} algorithm of 
\cite{ChildsReichardtSpalekZhang07andor}
does extend to give an optimal $O(2^d)$-query algorithm for evaluating the balanced ternary majority formula.  
Moreover, the algorithm also generalizes to a significantly larger gate set $\gateset$.

\paragraph{Organization}

We introduce span programs and explain their correspondence to weighted bipartite graphs in \secref{s:spanprograms}.  
The correspondence involves considering parts of a span program $P$ as the weighted adjacency matrix for a corresponding graph $G_P$.  
We prove that the eigenvalue-zero eigenvectors of this adjacency matrix
evaluate $P$ (\thmref{t:weakzeroenergy}).
This theorem provides useful intuition.

We develop a quantitative version of \thmref{t:weakzeroenergy} in \secref{sec:spc}.  We lower-bound the overlap of the eigenvalue-zero eigenvector with a known starting state.  This lower-bound will imply completeness of our quantum algorithm.  
To show soundness of the algorithm, we also analyze small-eigenvalue eigenvectors in order to prove a spectral gap around zero.  Essentially, we solve the eigenvalue equations in terms of the eigenvalue $\lambda$, and expand a series around $\lambda = 0$.  
The results for small-eigenvalue and eigenvalue-zero eigenvectors are closely related, and we unify them using a measure we term ``span program witness size."
The details of the proofs from this section are in \appref{s:wsizeproof}.

\secref{sec:algorithm} applies the span program framework to the formula evaluation problem.  
\thmref{t:result} is our general result, an optimal quantum algorithm for evaluating formulas that are over the gate set $\gateset$ of \defref{t:gatesetdef}, and that are adversary-balanced (\defref{t:adversarybalanceddef}).  
The proof of \thmref{t:result} has three parts.  First, in \secref{s:gatebygate}, we display an optimal span program for each of the gates in $\gateset$.  Second, we compose the span programs for the individual gates to obtain a span program for the full formula $\varphi$.  This is equivalent to joining together the gadget graphs described in \figref{f:example} to obtain a graph $G(\varphi)$.  We combine the spectral analyses of the individual span programs to analyze the spectrum of $G(\varphi)$ (\thmref{t:formulaspectrum}).  Finally, this analysis straightforwardly leads to a quantum algorithm based on phase estimation of a quantum walk on $G(\varphi)$, in \secref{s:qalg}.  

\secref{s:extensions} concludes with a discussion of some extensions to the algorithm.  

\section{Span programs and eigenvalue-zero graph eigenvectors} \label{s:spanprograms}

A span program $P$ is a certain linear-algebraic way of specifying a function $f_P$.  
For details on span programs applied in classical complexity theory, we can still recommend the original reference~\cite{KarchmerWigderson93span} as well as, e.g., the more recent~\cite{GalPudlak03spanprogram}.

\begin{definition}[Span program] \label{t:spanprogramdef}
A span program $P$ consists of a nonzero ``target" vector $t$ in a vector space over $\bf{C}$, together with ``grouped input" vectors $\{v_j : j \in J\}$.  Each $v_j$ is labeled with a subset $X_j$ of the literals $\{ x_1, \overline{x_1}, \ldots, x_n, \overline{x_n} \}$.  
To $P$ corresponds a boolean function $f_P : \{0,1\}^n \rightarrow \{0,1\}$; defined by $f_P(x) = 1$ (i.e., true) if and only if there exists a linear combination $\sum_j a_j v_j = t$ such that $a_j = 0$ if any of the literals in $X_j$ evaluates to zero (i.e., false).
\end{definition}

\begin{example} \label{t:maj3example}
For example, the span program 
\[
\setcounter{sprows}{2}
\setlength{\spheight}{16pt}
\setlength{\spraise}{24pt}
\begin{array}{r@{}c@{}l r@{}*{2}{c@{\ }}c@{}l}
& & & X_J = ( & \{x_1\} & \{x_2\} & \{x_3\} & ) \\
\noalign{\smallskip}
\spleft{t=} & 1 & \spright{,} & \spleft[2pt]{v_J=} & \frac1{\sqrt3}& \frac1{\sqrt3}& \frac1{\sqrt3}& \spright[2pt]{} \\
& 0 & & & 1 & e^{\strut 2\pi i/2} & e^{-2\pi i/3} &
\end{array}
\]
computes the $\MAJ$ function.  Indeed, at least two of the $v_j$ must have nonzero coefficient in any linear combination equaling the target $t$.  
Of course, the second row of $(\begin{matrix} v_1 & v_2 & v_3 \end{matrix})$ could be any $(\begin{matrix} \alpha & \beta & \gamma \end{matrix})$ with $\alpha, \beta, \gamma$ distinct and nonzero, and the span program would still compute $\MAJ$.  This specific setting is used to optimize the running time of the quantum algorithm (\claimref{t:majwsize}).
\end{example}

In this section, we will show that by viewing a span program $P$ as the weighted adjacency matrix $A_{G_P}$ of a certain graph $G_P$, the true/false evaluation of $P$ on input $x$ corresponds to the existence or nonexistence of an eigenvalue-zero eigenvector of $A_{G_P}(x)$ supported on a distinguished output node (\thmref{t:weakzeroenergy}).  

In turn, this will imply that writing a span program $P$ for a function $f$ immediately gives a quantum algorithm for evaluating $f$, or for evaluating formulas including $f$ as a gate (\secref{sec:algorithm}).  
The algorithm works by spectral estimation on $A_{G_P}(x)$.
Its running time depends on the span program's ``witness size" (\secref{sec:spc}).
For example, if $f_P(x)$ is true, then the witness size is essentially the shortest squared length of any witness vector $(a_j)_{j \in J}$ in \defref{t:spanprogramdef}.  

\begin{remark} \label{t:spanprogramremark}
Let us clarify a few points in \defref{t:spanprogramdef}.
\begin{enumerate}
\item
It is convenient, but nonstandard, to allow grouped inputs, i.e., literal subsets $X_j$ possibly with $\abs{X_j} > 1$, instead of just single literals, to label the columns.  A grouped input $j$ can be thought of as evaluating the AND of all literals in $X_j$.  A span program $P$ with some $\abs{X_j}> 1$ can be expanded out so that all $\abs{X_j} \leq 1$, without increasing $\sum_j \abs{X_j}$, known as the {\em size} of $P$.
\item
It is sometimes convenient to allow $X_j = \emptyset$.  In this case, vector $v_j$ is always available to use in the linear combination; grouped input $j$ evaluates to true always.  However, such vectors can be eliminated from $P$ without increasing the size~\cite[Theorem~7]{KarchmerWigderson93span}.  
\item
By a basis change, one can always adjust the target vector $t$ to $(1,0,0,\ldots,0)$.  
\end{enumerate}
\end{remark}

\subsection{Span program as an adjacency matrix}

A span program $P$ with target vector $t = (1,0,\ldots,0)$ corresponds to a certain weighted bipartite graph.  

Notation: For an index sequence $H = (h_1, \ldots, h_{\abs{H}})$ and a set of variables $\{a_h\}$, let $a_H = (a_{h_1}, \ldots, a_{h_{\abs{H}}})$.  For example, $v_J$ denotes the sequence of grouped input vectors.  
It will be convenient to define several more index sequences: $O$ (``output"), $C$ (``constraints") and $I$ (``inputs").  Let $O$ and $C$ together index the coordinates of the vector space, with $O = \{1\}$ being the first coordinate, and $C$ the remainder.  Let $I_j$ index $X_j$ for each $j \in J$, and let $I = \bigcup_{j \in J} I_j$ a disjoint union so $\abs{I} = \mathrm{size}(P)$.

\begin{figure}
\begin{equation*}
\hskip-10pt
\raisebox{-35pt}{\includegraphics[scale=.6]{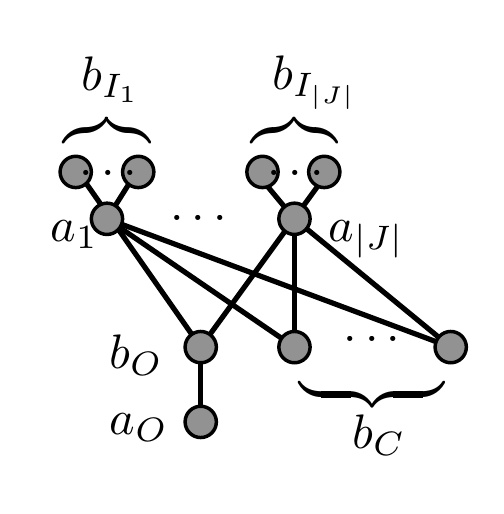}}
\quad
\setcounter{sprows}{3}
\setlength{\spheight}{28pt}
\setlength{\spraise}{-12pt}
A_{G_P} = 
\begin{array}{*{2}{@{\ }c} | *{2}{@{\ }c} @{} r}
& \multispan1{$a_O$} & a_J & & \\
& \multispan1{} & \overbrace{\quad} & & \\
\noalign{\vskip-2pt}
\spleft{} & 1 & A_{OJ} & \spright{} & b_O \\
& \begin{smallmatrix}0\\ :\\\vphantom{(} 0\end{smallmatrix} & A_{CJ} & & \Big\} b_C \\
\cline{2-3}
& \begin{smallmatrix}\strut 0\\ :\\0\end{smallmatrix} & A_{IJ} & & \Big\} b_{I\;}
\end{array}
\end{equation*}
\caption{The bipartite graph $G_P$ corresponding to span program $P$ (the output edge is $(a_O, b_O)$, while the grouped inputs are $a_1, \ldots, a_{\abs{J}}$).
}
\label{f:gadget}
\end{figure}

We will construct a graph $G_P$ on $\abs{I} + \abs{J} + \abs{C} + 2 \abs{O}$ vertices.
Writing the grouped input vectors out as the columns of a matrix, let $\left(\begin{smallmatrix} A_{OJ} \\ A_{CJ} \end{smallmatrix}\right) = \sum_{j \in J} \ketbra{v_j}{j}$; $A_{OJ}$ is a $1 \times \abs{J}$ matrix row, and $A_{CJ}$ is a $\abs{C} \times \abs{J}$ matrix.  
Let $A_{IJ} = \sum_{{j \in J , i \in I_j}} \ketbra{i}{j}$; $A_{IJ}$ encodes $P$'s grouped inputs.  
Now consider the bipartite graph $G_P$ of \figref{f:gadget}, the upper right block of whose weighted Hermitian adjacency matrix is $A_{G_P}$.  (The adjacency matrix is block off-diagonal because the graph is bipartite.)
The edges $(a_j, b_i)$ for $j \in J$ and $i \in I_j$ are ``input edges," while $(a_O, b_O)$ is the ``output edge."  The input and output edges all have weight one.  The weights of edges $(b_O, a_j)$ for $j \in J$ are given by $A_{OJ}$ (the first coordinates of the grouped input vectors $v_J$), while the weights of edges $(b_c, a_j)$ for $c \in C, j \in J$ are given by $A_{CJ}$ (the remaining coordinates of $v_J$).

\begin{example}
For the $\MAJ$ span program of \exampleref{t:maj3example}, $|C|=1$, $|I|=|J|=3$, the graph $G_P$ is shown in \figref{f:example}, and the matrix $A_{G_P}$ is
\[
\left( \begin{array}{c | c c c}
1 & \frac 1 {\sqrt 3} & \frac 1 {\sqrt 3} & \frac 1 {\sqrt 3} \\
0 & 1 & e^{\strut 2 \pi i / 3}	 & e^{-2 \pi i / 3} \\
\hline
0 & 1 & 0 & 0 \\
0 & 0 & 1 & 0 \\
0 & 0 & 0 & 1 \\
\end{array} \right) \enspace .
\]
\end{example}

\subsection{Eigenvalue-zero eigenvectors of the span program adjacency matrix}

\begin{theorem} \label{t:weakzeroenergy}
For an input $x \in \{0,1\}^n$, define a weighted graph $G_P(x)$ by deleting from $G_P$ the edges $(a_j, b_i)$ if the $i$th literal in $X_j$ is true.  
Consider all the eigenvalue-zero eigenvector equations of the weighted adjacency matrix $A_{G_P(x)}$, except for the constraint at $a_O$.  
These equations have a solution with support on vertex $a_O$ if and only if $f_P(x) = 1$, and have a solution with support on $b_O$ if and only if $f_P(x) = 0$.
\end{theorem}

\begin{proof}
Notation: Use $a_j, b_i, b_c, a_O, b_O$ to denote coefficients of a vector on the vertices of $G_P$.  Let $A_{IJ}(x)$ include only edges to false inputs, i.e., $A_{IJ}(x) = \sum_{j \in J, \text{false $i$} \in I_j} \ketbra{i}{j}$.  

The eigenvalue-$\lambda$ eigenvector equations of $A_{G_P(x)}$ are 
\begin{subequations} \label{e:zeroenergyequations}
\begin{align}
\lambda b_O &= a_O + A_{OJ} a_J \label{e:zeroenergyequationsbO} \\
\lambda b_C &= A_{CJ} a_J \label{e:zeroenergyequationsbC} \\ 
\lambda b_I &= A_{IJ}(x) a_J \label{e:zeroenergyequationsbI} \\
\lambda a_O &= b_O \label{e:zeroenergyequationsaO} \\
\lambda a_J &= {A_{OJ}}^\adjoint b_O + {A_{CJ}}^\adjoint b_C + {A_{IJ}}(x)^\adjoint b_I \label{e:zeroenergyequationsaJ}
\end{align}
\end{subequations}
At $\lambda=0$, these equations say that for each vertex, the weighted sum of the adjacent vertices' coefficients must be zero.  
We are looking for solutions satisfying all these equations except possibly Eq.~\eqnref{e:zeroenergyequationsaO}.  
Since the graph is bipartite, at $\lambda = 0$ the $a$ coefficients do not interact with the $b$ coefficients.  In particular, Eqs.~(\ref{e:zeroenergyequations}d,e) (resp. \ref{e:zeroenergyequations}a-c) can always be satisfied by setting the $b$ (resp. $a$) coefficients to zero.

By scaling, there is a solution with nonzero $a_O$ iff there is a solution with $a_O = -1$.  Then Eqs.~(\ref{e:zeroenergyequations}a,b) are equivalent to ${t} = \left(\begin{smallmatrix}A_{OJ} \\ A_{CJ} \end{smallmatrix}\right) a_J = \sum_j a_j{v_j}$.  Moreover, Eq.~\eqnref{e:zeroenergyequationsbI} implies that $a_j$ can be nonzero only if grouped input $j$ is true.  (If $X_j$ includes any false inputs, then $A_{IJ}(x) \ket{j} \neq 0$, so $a_j = 0$.)  These conditions are the same as those in \defref{t:spanprogramdef}.  

Next, we argue that there is a solution of Eq.~\eqnref{e:zeroenergyequationsaJ} with $\lambda = 0$ and $b_O = 1$ if and only if $f_P(x) = 0$.  Indeed, 
\[
f_P(x) = 0 \Leftrightarrow t \notin \Span\{v_j : \text{$j$ true}\} \Leftrightarrow t \notin \Range\!\left[\left(\begin{smallmatrix}A_{OJ} \\ A_{CJ} \end{smallmatrix}\right) \Pi\right]
\]
where $\Pi = \sum_{\text{true $j$}} \ketbra{j}{j}$.  
In turn, this holds iff there is a vector $w$ orthogonal to the range and having inner product one with $t$---precisely constraint \eqnref{e:zeroenergyequationsaJ} with ${w} = \left(\begin{smallmatrix}b_O \\ b_C \end{smallmatrix}\right)$.  
\end{proof}

\begin{remark} \label{t:dualrailremark}
By \thmref{t:weakzeroenergy}, we can think of the graph $G_P$ as giving a ``dual-rail" encoding of the function $f_P$: there is a $\lambda = 0$ eigenvector of $G_P(x)$ supported on $a_O$ if and only if $f_P(x) = 1$, and there is one supported on $b_O$ iff $f_P(x) = 0$.  This justifies calling edge $(a_O, b_O)$ the output edge of $G_P$.
\end{remark}

\subsection{Dual span program} \label{s:dualspanprogram}

A span program $P$ immediately gives a dual span program, denoted $P^\adjoint$, such that $f_{P^\adjoint}(x) = \neg f_P(x)$ for all $x \in \{0,1\}^n$.  For our purposes, though, it suffices to define a NOT gate graph gadget to allow negation of subformulas.

\begin{definition}[NOT gate gadget] \label{t:notgatedef}
Implement a NOT gate $x \mapsto \overline{x}$ as two weight-one edges connected (\figref{f:example}).  The edge $(a_i, b_i)$ is the input edge, while $(a_O, b_O)$ is the output edge.  The middle vertex $a_i = b_O$ is shared.  
\end{definition}

At $\lambda = 0$, the coefficient on $a_O$ is minus that on $b_i$, and $a_i = b_O$ by definition.  Therefore, this gadget complements the dual rail encoding of \thmref{t:weakzeroenergy}.

The NOT gate gadget of \defref{t:notgatedef} can be used to define a dual span program $P^\adjoint$ by complementing the output and all inputs with NOT gates, and also complementing all input literals in the sets $X_J$.  Since it is not essential here, we leave the formal definition as an exercise.  Alternative constructions of dual programs are given in~\cite{CramerFehr02secretshare, NikovNikovaPreneel05spanprogram}.

\begin{example} \label{t:dualspanprogramexample}
For distinct, nonzero $\alpha, \beta, \gamma$, the span program\!\!\!
\[
\setcounter{sprows}{4}
\setlength{\spheight}{24pt}
\setlength{\spraise}{6pt}
\begin{array}{r@{}c@{}l r@{\ }*{4}{c@{\ }}c@{}l}
& & & X_J = ( & \emptyset & \emptyset & \{\overline{x_1}\} & \{\overline{x_2}\} & \{\overline{x_3}\} & ) \\
\noalign{\smallskip}
\spleft{t=} & 1 & \spright{,} & \spleft[2pt]{v_J=} & 1 & 0 & 0 & 0 & 0 & \spright[2pt]{} \\
            & 0 &             &                    & 1 & \alpha & 1 & 0 & 0 & \\
            & 0 &             &                    & 1 & \beta & 0 & 1 & 0 & \\
            & 0 &             &                    & 1 & \gamma & 0 & 0 & 1 & \\
\end{array}
\]
computes $\neg \MAJ(x_1, x_2, x_3)$. 
It was constructed, by adding NOT gate gadgets, as the dual to the span program in \exampleref{t:maj3example}, up to choice of weights.
\end{example}

\subsection{Span program composition}

\begin{definition}[Composed graph and span program] \label{t:composespanprogramdef}
Consider span program $Q$ on $\{0,1\}^n$ and programs $Q_i$, $i \in [n] \equiv \{1, \ldots, n\}$, with corresponding graphs $G_Q$ and $G_{Q_i}$.  The composed graph is defined by identifying the input edges of $G_Q$ with the output edges of copies of the other graphs.  If an edge corresponds to input literal $x_i$, then identify that edge with the output edge of a copy of $G_{Q_i}$; and if an edge corresponds to $\overline{x_i}$, then insert a NOT gate gadget (i.e., an extra vertex, as in \defref{t:notgatedef}) before a copy of $G_{Q_i}$.  
The composed span program, denoted $Q \circ Q_{[n]}$, is the program corresponding to the composed graph (i.e., $G_{Q \circ Q_{[n]}}$ is the composed graph).  Thus $f_{Q \circ Q_{[n]}} = f_Q \circ f_{Q_{[n]}}$.  
\end{definition}

\begin{definition}[Formula graph and span program]
Given span programs for each gate in a formula $\varphi$, span program $P(\varphi)$ is defined as their composition according to the formula.  Let $G(\varphi)$ be the composed graph, $G(\varphi) = G_{P(\varphi)}$.  
\end{definition}

\begin{example}
For example, the span program 
\[
\setcounter{sprows}{4}
\setlength{\spheight}{24pt}
\setlength{\spraise}{4pt}
\begin{array}{r@{}c@{}l r@{}*{5}{c@{\ }}c@{}l}
& & & X_J = ( & \{x_1\} & \{x_2\} & \emptyset & \{x_3\} & \{x_4\} & \{x_5\} & ) \\
\noalign{\smallskip}
\spleft{t=} & 1 & \spright{,} & \spleft[2pt]{v_J=} & 1&1&1&0&0&0& \spright[2pt]{} \\
& 0 & & & \alpha&\beta&\gamma&0 & 0& 0&\\
& 0 & & & 0 & 0& 1& 1& 1& 1& \\
& 0 & & & 0 & 0& 0& \alpha&\beta&\gamma&\\
\end{array}
\]
is a composed span program that 
computes the function $\MAJ(x_1, x_2, \MAJ(x_3, x_4, x_5))$, provided $\alpha, \beta, \gamma$ are distinct and nonzero.  (See \exampleref{t:maj3example}.)  \figref{f:maj3ofmaj3} shows the associated composed graph.
\end{example}

\begin{figure}
\begin{center}
\includegraphics[scale=.62]{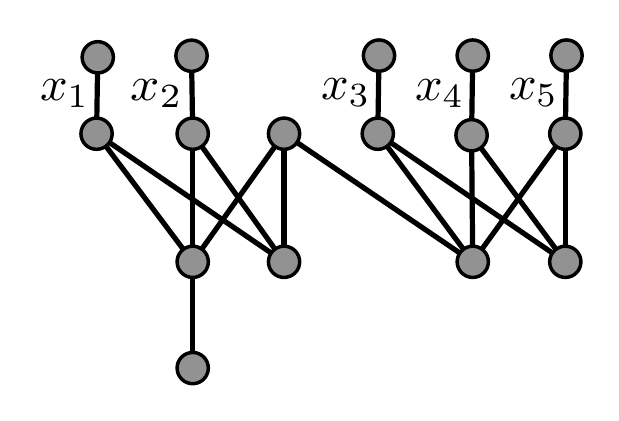}
\end{center}
\vskip -10pt
\caption{Graph for $\MAJ(x_1, x_2, \MAJ(x_3, x_4, x_5))$, with input edges labeled by the associated literals.}
\label{f:maj3ofmaj3}
\end{figure}

\begin{figure}
\begin{center}
\subfigure{
\includegraphics[scale=.52]{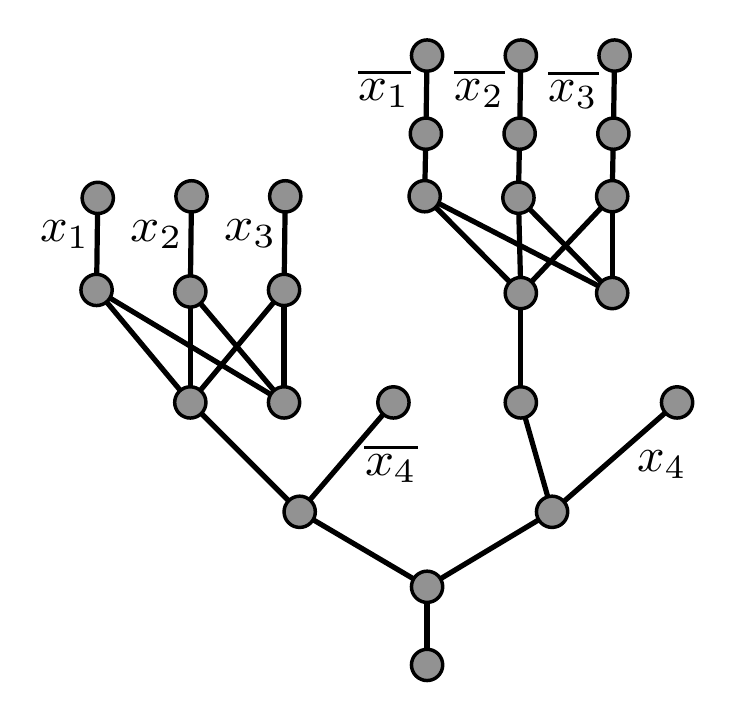}
}
\subfigure{
\includegraphics[scale=.52]{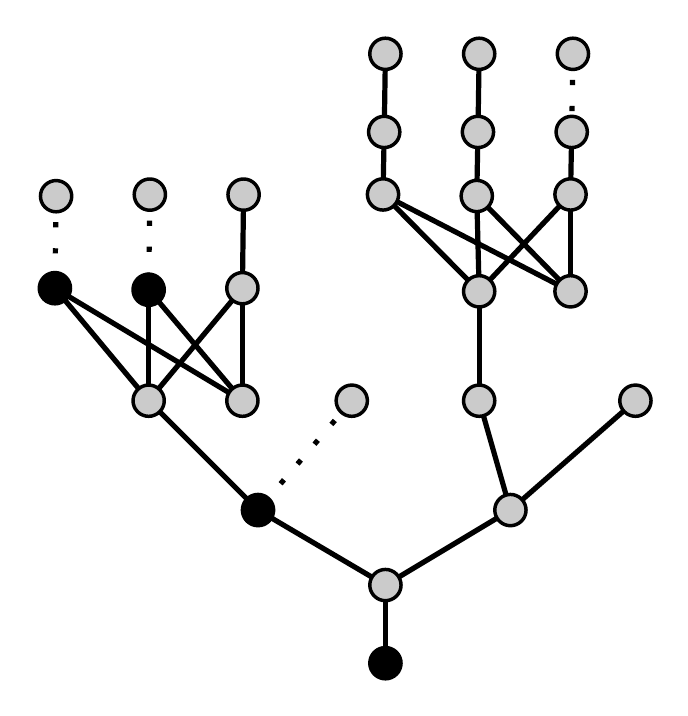}
}
\end{center}
\vskip -10pt
\caption{Graph for $\MAJ(x_1, x_2, x_3) \oplus x_4$, and its evaluation on input $x=1100$.}
\label{f:parityandmaj3}
\end{figure}

\begin{example}[Duplicating and negating inputs]
To the left in \figref{f:parityandmaj3} 
is the composed graph for the formula $\MAJ(x_1, x_2, x_3) \oplus x_4 = \EQUAL_2(\MAJ(x_{[3]}), \overline{x_4})$, obtained using the substitution rules of \figref{f:example}.  
(A span program for PARITY will be given in \lemref{t:parityorwsize}.)
Note that we are effectively negating some inputs twice, by putting NOT gate gadgets below the negated literals $\overline{x_1}, \overline{x_2}, \overline{x_3}$.  This is of course redundant, and is only done to maintain the strict correspondence of graphs to span programs, as in \exampleref{t:dualspanprogramexample}, by keeping the input vertices $b_I$ at odd distances from $a_O$.

To the right is the same graph evaluated on input $x = 1100$, i.e., with edges to true literals deleted.  Since the formula evaluates to true, \thmref{t:weakzeroenergy} promises that there is a $\lambda = 0$ eigenvector supported on $a_O$.  In this case, that eigenvector is unique.  It has support on the black vertices.
\end{example}

\section{Span program witness size}
\label{sec:spc}

In \secref{s:spanprograms}, we established that after converting a formula $\varphi$ into a weighted graph $G(\varphi)$, by replacing each gate with a gadget subgraph coming from a span program, the eigenvalue-zero eigenvectors of the graph effectively evaluate $\varphi$.  
The dual-rail encoding of $\varphi(x) = f_{P(\varphi)}(x)$ 
promised by \thmref{t:weakzeroenergy} will suffice to give a phase-estimation-based quantum algorithm for evaluating $\varphi$.  
The goal of this section is to make \thmref{t:weakzeroenergy} more quantitative, which will enable us to analyze the algorithm's running time.  

In particular, we will {lower-bound} the achievable squared support on either $a_O$ or $b_O$ of a unit-normalized $\lambda = 0$ eigenvector.  This will enable the algorithm to detect if $\varphi(x) = 1$ by starting a quantum walk at $a_O$; if $\varphi(x) = 1$, then
$\ket{a_O}$ will have large overlap with the $\lambda = 0$ eigenvector. 

We also study eigenvalue-$\lambda$ eigenvectors of $G_{P(\varphi)}(x)$, for $\abs{\lambda} \neq 0$ sufficiently small.  At small enough eigenvalues, the dual-rail encoding property of \thmref{t:weakzeroenergy} still holds, in a different fashion.  Note that since the graph is bipartite, we may take $\lambda > 0$ without loss of generality.  For small enough $\lambda > 0$, it will turn out that the function evaluation corresponds to the output {\em ratio} $\ri_{O} \equiv a_O / b_O$.  If $f_P(x) = 1$, then
$\ri_{O}$ is large and negative, roughly of order $-1/\lambda$.  If $f_P(x) = 0$, then $\ri_{O}$ is small and positive, roughly of order $\lambda$.  
Ultimately, the point of this analysis is to show that if the formula evaluates to false, then there do not exist any eigenvalue-$\lambda$ eigenvectors supported on $a_O$ for small enough $\abs{\lambda} \neq 0$.  This spectral gap will prevent the algorithm from outputting false positives.

\smallskip
Consider a span program $P$.  Let us generalize the setting of \thmref{t:weakzeroenergy} to allow $P$'s inputs to be themselves span programs, as in \defref{t:composespanprogramdef}.  Assume that for some $x$, every $\lambda \in [0, \Lambda)$ and each input $i \in I$, we have constructed unit-normalized states $\ket{\psi_i(\lambda)}$ satisfying the eigenvalue-$\lambda$ constraints for all the $i$th subgraph's vertices except $a_i$.  

\begin{definition}[Subformula complexity] \label{def:subformulacomplexity}
At $\lambda = 0$, for each input $i \in I_j$, let $\abi_i$ denote $\ket{\psi_i}$'s squared support on either $a_j$ or $b_i$, depending on whether the input evaluates to true or false, respectively.  

For $\lambda>0$, 
assume that the coefficients of $\ket{\psi_i}$ along the $i$th input edge are nonzero, and 
let $\ri_{i} = a_j / b_i$ be their ratio.  If the literal associated to input $i$ evaluates to false, then let $\si_{i} = \ri_{i} / \lambda$; if it is true, then let $\si_{i} = -1 / (\ri_{i} \lambda)$.

For an input $i \in I$, its \emph{subformula complexity} is 
\begin{equation}
\label{e:Ui}
\Ui_{i} = \max_x \max \left\{ 1/\abi_i, \sup_{\lambda \in (0, \Lambda)} \si_{i} \right\} \enspace .
\end{equation}
\end{definition}

\noindent For example, if $\Ui_{i}$ is small, then $\ket{\psi_i(0)}$ has large support on either $a_i$ or $b_i$.  In general, $\Ui_{i} \ge 1$.  If input $i$ corresponds to a literal and not the output edge of another span program, then $\Ui_{i} = 1$.  

We construct a normalized state $\ket{\psi_O(\lambda)}$ that satisfies {all} the eigenvalue-$\lambda$ eigenvector constraints of the composed graph, except Eq.~\eqnref{e:zeroenergyequationsaO} at $a_O$.
We construct $\ket{\psi_O}$ by putting together the scaled $\ket{\psi_i}$s and also assigning coefficients to the vertices in $G_P$.  Similarly to Eq.~\eqnref{e:Ui}, define
\begin{equation}
\Ui_{O}(x) = \max \left\{ 1/\abi_O, \sup_{\lambda \in (0, \Lambda)} \si_{O} \right\} \enspace ,
\end{equation}
where $\abi_O$ is the squared support of $\ket{\psi_O(0)}$ on $a_O$ or $b_O$, and, for $\lambda > 0$, $\si_{O}$ is $-1/(\ri_{O} \lambda)$ or $\ri_{O}/\lambda$ if $f(x)$ is true or false, respectively.  We will relate $\Ui_{O} = \max_x \Ui_{O}(x)$ to the input complexities $\Ui_{I}$ (\thmref{t:unifiedsolution}).  

First of all, notice that if $\abs{I_j} > 1$, then several of the input subgraphs share the vertex $a_j$.  They must be scaled so that their coefficients at $a_j$ all match, motivating the following definition.  

\begin{definition}
\label{def:groupedinputcomplexity}
The \emph{grouped input complexity} of $j \in J$ on input $x$ is 
\begin{equation} \label{e:Ujdef}
\Uj_{j}(x) = \left\{ \begin{array}{cl} \max \Big\{\sum_{i \in I_j} \Ui_{i}, 1 \Big\} & \text{if $j$ evaluates to true}
\\ \Big(\sum_\text{false $i \in I_j$} \Ui_{i}^{-1} \Big)^{-1} & \text{otherwise}
\end{array} \right.
\end{equation}
Recall that grouped input $j$ evaluates to true iff all inputs in $I_j$ are true.
(In the first case, we take the maximum with 1 to handle the case $I_j = \emptyset$.)  
\end{definition}

\noindent When $j$ is false, some input $i \in I_j$ is false, so the coefficient at $a_j$ must be set to zero at $\lambda = 0$.  However, for each false $i \in I_j$, $\ket{\psi_i}$ can be scaled arbitrarily.  The definition in Eq.~\eqnref{e:Ujdef} comes from choosing scale factors $f_i$ in order to maximize the sum of the scaled coefficients on the vertices $b_i$, under the constraint that the total norm be one, $\sum_{i \in I_j} \abs{f_i}^2 = 1$.

A few more definitions are needed to state \thmref{t:unifiedsolution}.  

\begin{definition}[Asymptotic notation]
Let $a \lesssim b$ mean that there exist constants $c_1, c_2$ such that $a \leq c_1 + b (1 + c_2 \abs{\lambda} \max_i \Ui_{i})$.  
\end{definition}

\begin{definition}[Matrix notations] \label{def:matrixnotations}
For a given input $x$, let $\Pi(x) = \sum_{\text{true $j$}} \ketbra{j}{j}$ a projection onto the true grouped inputs, $\Pim(x) = 1-\Pi(x)$, and $\Uj(x) = \sum_j \Uj_{j}(x) \ketbra{j}{j}$ a diagonal matrix of the grouped input complexities.  To simplify equations, we will generally leave implicit the dependence on $x$, writing $\Pi$, $\Pim$ and $\Uj$.  
Let $A = \left(\begin{smallmatrix} A_{OJ} \\ A_{CJ} \end{smallmatrix}\right) = \sum_j \ketbra{v_j}{j}$ with columns the vectors $\ket{v_j}$.  
\end{definition}

\begin{definition}[Moore-Penrose pseudoinverse]
For a matrix $M$, let $M^+$ denote its Moore-Penrose pseudoinverse.  If the singular-value decomposition of $M$ is $M = \sum_k m_k \ketbra{k}{k'}$ with all $m_k > 0$ and for sets of orthonormal vectors $\{\ket{k}\}$ and $\{\ket{k'}\}$, then $M^+ = \sum_k m_k^{-1} \ketbra{k'}{k}$.  Note that $M M^+ = \sum_k \ketbra{k}{k}$ is the projection onto $M$'s range.
\end{definition}

\begin{definition}[Span program witness size] \label{def:spwsize}
For span program $P$ and input subformula complexities $\Ui_{I}$, the witness size of $P$ is $\wsize P = \max_x \wsizex P x$, where for an input $x$, $\wsizex P x$ is defined as follows: 
\begin{itemize}
\item If $f_P(x) = 1$, then $\ket{t} \in \Range(A \Pi)$, 
so there is a witness $\ket{w} \in {\bf C}^{\abs{J}}$ satisfying $A \Pi \ket{w} = \ket{t}$.  Then $\wsizex P x$ is the minimum squared length, weighted by $\Uj(x)^{1/2}$, of any such witness: 
\begin{align} \label{e:spctrue}
\wsizex{P}{x} 
&= \min_{\ket w : A \Pi \ket{w} = \ket{t}} \norm{\Uj^{1/2} \ket{w}}^2 \\
&= \norm{ (A \Pi \Uj^{-1/2})^+ \ket{t} }^2 
 \enspace . \nonumber
\end{align}
\item If $f_P(x) = 0$, then $\ket{t} \notin \Range(A \Pi)$.
Therefore there is a witness $\ket{w'} \in {\bf C}^{\abs{C}+1}$ satisfying $\braket{t}{w'} = 1$ and $\Pi A^\adjoint \ket{w'} = 0$.  Then 
\begin{align} \label{e:spcfalse}
\wsizex{P}{x} 
&= \min_{\substack{\ket{w'} : \braket{t}{w'} = 1 \\ \Pi A^\adjoint \ket{w'} = 0}} \norm{\Uj^{1/2} A^\adjoint \ket{w'}}^2 \\
&= \norm{ \big(1 + (\Pim (A \Uj^{1/2})^+ A \Uj^{1/2} - 1)^+ \Pi \big) (A \Uj^{1/2})^+ \ket{t} }^{-2}
 \enspace , \nonumber
\end{align}
the inverse squared length of the projection of $(A \Uj^{1/2})^+ \ket{t}$ onto the intersection of $\Pim$ and $\Range(\Uj^{1/2} A^\adjoint)$.
\end{itemize}
\noindent
By $\ket{w_x}$, resp.~$\ket{w'_x}$, we denote a witness for input $x$ achieving the minimum in Eq.~\eqnref{e:spctrue}, resp.~\eqnref{e:spcfalse}.
\end{definition}

\noindent 
The span program witness size is easily computed on any given input $x$.  
\lemref{t:spcinterpretation} below will give two alternative expressions for $\wsizex{P}{x}$.  
Now our main result is: 

\begin{theorem} \label{t:unifiedsolution}
Consider a constant span program $P$.  Assume that $\Lambda \Ui_{i} \leq \eps$ for a small enough constant $\eps > 0$ to be determined and for all $i \in I$.  
Then 
\begin{equation} \label{e:unifiedsolution}
\Ui_{O}(x) \lesssim \wsizex{P}{x} \enspace .
\end{equation}
\end{theorem}

For $\lambda = 0$, Eq.~\eqnref{e:unifiedsolution} says that the achievable squared magnitude on $a_O$ or $b_O$ of a normalized eigenvalue-zero eigenvector is at least $1/\wsizex{P}{x}$, up to small controlled terms.  For $\lambda > 0$, Eq.~\eqnref{e:unifiedsolution} says that the ratio $\ri_{O} = a_O / b_O$ is either in $(0, \wsizex{P}{x} \lambda]$ or $(-\infty, -1/ (\wsizex{P}{x} \lambda)]$, up to small controlled terms, depending on whether $f_P(x)$ is false or true.  

\begin{proof}[Proof sketch of \thmref{t:unifiedsolution}]
At $\lambda = 0$, the proof of \thmref{t:unifiedsolution} is the same as that of \thmref{t:weakzeroenergy}, except scaling the inputs so as to maximize the squared magnitude on $a_O$ or $b_O$.  This maximization problem is essentially the same as the problems stated in \defref{def:spwsize} (up to additive constants).  The explicit expressions for the solutions follow by geometry.

For $\lambda > 0$, we solve the eigenvalue equations (\ref{e:zeroenergyequations}a,b,e) by inverting a matrix and applying the Woodbury formula.  We argue that all inverses exist in the given range of $\lambda$.  We obtain 
\[
\ri_{O} = a_O / b_O = \lambda + \bra{\o} \rj \ket{\o} + \lambda \bra{\o} \rj {A_{CJ}}^\adjoint X^{-1} A_{CJ} \rj \ket{\o} \enspace ,
\]
where $\ket{\o} = {A_{OJ}}^\adjoint$, $\rj = -\frac1{\lambda} \sj^{-1} \Pi  + \lambda \sj \Pim$ (with $\sj$ defined from $\si$ similarly to how $\Uj$ is defined from $\Ui$), and 
$X = A_{CJ} \sj^{-1} \Pi {A_{CJ}}^\adjoint  - \lambda^2 A_{CJ} \sj \Pim {A_{CJ}}^\adjoint - \lambda^2$.  The largest term in $X$, $A_{CJ} \sj^{-1} \Pi {A_{CJ}}^\adjoint$, is only invertible restricted to its range, $\Delta = A_{CJ} \Pi (A_{CJ} \Pi)^+$.  Therefore, we compute the Taylor series in $\lambda$ of the pseudoinverse of $\Delta X \Delta$ and of its Schur complement, $\big(X / (\Delta X \Delta)\big)$, separately, and then recombine them.  The lowest-order term in the solution again corresponds to \defref{def:spwsize} (if $f_P(x)$ is false, the $1/\lambda$ term is zero), and we bound the higher-order terms.  
\end{proof}

The full proof of \thmref{t:unifiedsolution} is given in \appref{s:wsizeproof}.  

\begin{remark}
In case $f_P(x) = 0$, $A^\adjoint \ket{w'}$ appears also in the ``canonical form" of the span program $P$~\cite{KarchmerWigderson93span}.
\end{remark}

The above analysis of span programs does not apply to the NOT gate, because the ability to complement inputs was assumed in \defref{t:spanprogramdef}.  Implementing the NOT gate $x \mapsto \overline{x}$ with a span program on the literal $\overline{x}$ would be circular.  Therefore we provide a separate analysis.

\begin{lemma}[NOT gate] \label{t:not}
Consider a NOT gate, implemented as two weight-one edges connected as in \defref{t:notgatedef}.
Assume $\abs{\lambda} \le 1 / (\sqrt{2 \Ui_{i}})$.  Then $\Ui_{O} \lesssim \Ui_{i}$.
\end{lemma}

\begin{proof}
{\subsubsecfnt Analysis at $\lambda = 0$.}
If the input is true, then $\abi_i$ measures the squared support on $a_i$ of a normalized $\lambda = 0$ eigenvector.  Then $\abi_O = \abi_i$, since $a_i = b_O$ the output vertex.
If the input is false, so $b_i = \sqrt{\abi_i}$, then $b_i + a_O = 0$.  Therefore, we simply need to renormalize: $\abi_O = \abi_i / (1 + \abi_i)$, or equivalently
$\frac1{\abi_O} = \frac1{\abi_i} + 1$.

\subparagraph{Analysis for small $\lambda > 0$}
We are given $\ri_{i} = a_i / b_i$.  The eigenvector equation is $b_i + a_O = \lambda a_i = \lambda b_O$.  Therefore, $\ri_{O} = a_O / b_O = \lambda - 1/\ri_{i}$.  
If the input is false, so $\si_{i} = \ri_{i} / \lambda$, then $\si_{O} = -1 / (\lambda \ri_{O}) = \si_{i} / (1 - \lambda^2 \si_{i})$.  Therefore, $\si_{i} < \si_{O} \leq \si_{i} (1 + 2 \lambda^2 \si_{i})$ since $\lambda^2 \si_{i} \leq 1/2$.
If the input is true, so $\si_{i} = -1 / (\lambda \ri_{i})$, then $\si_{O} = \ri_{O} / \lambda = \si_{i} + 1$ .

Therefore $\Ui_{O} \lesssim \Ui_{i}$ as claimed.
Note that w.l.o.g.\ we may assume there are never two NOT gates in a row in the formula $\varphi$, so the additive constants lost do not accumulate.  
\end{proof}

\section{Formula evaluation algorithm\texorpdfstring{\!\!\!}{}}
\label{sec:algorithm}

In \secref{s:result}, we specify the gate set $\gateset$ (\defref{t:gatesetdef}) and define the correct notion of ``balance" for a formula that includes different kinds of gates (\defref{t:adversarybalanceddef}).  These two definitions allow us to formulate the general statement of our results, \thmref{t:result}, of which \thmref{t:balancedmajority} is a corollary.

In \secref{s:gatebygate}, we present span programs of optimal witness size for each of the gates in $\gateset$.  \thmref{t:formulaspectrum} in \secref{s:spectralanalysis} plugs together the spectral analyses of the individual span programs to give a spectral analysis of $G(\varphi)$.  Finally, we sketch in \secref{s:qalg} how this implies a quantum algorithm, therefore proving \thmref{t:result}.

\subsection{General formula evaluation result} \label{s:result}

\begin{definition}[Extended gate set $\gateset$] \label{t:gatesetdef}
Let 
\begin{equation}\begin{split}
\gateset' &= \left\{ \begin{array}{c}
\text{arbitrary 1-, 2-, or 3-bit gates,} \\
\text{$O(1)$-fan-in EQUAL gates}
\end{array}
\right\} \\
\gateset &= \left\{ 
\begin{array}{c}
\text{$O(1)$-size $\{\text{AND, OR, NOT, PARITY}\}$} \\ 
\text{read-once formulas composed onto} \\
\text{the gates from $\gateset'$} 
\end{array}
\right\}
\end{split}\end{equation}
\end{definition}

\begin{example}
The gate set $\gateset$ includes simple gates like AND, as well as substantially more complicated gates like $\MAJ(x_1,x_2,x_3) \wedge (x_4 \oplus x_5 \oplus \cdots \oplus x_{k-1} \oplus (x_k \vee x_{k+1}))$, provided $k = O(1)$.
It does not include gates from $\gateset'$ composed onto gates from $\gateset$: for example $\MAJ(x_1, x_2 \oplus x_3, x_4 \wedge x_5) \notin \gateset$.
\end{example}

To define ``adversary-balanced" formulas, we need to define the nonnegative-weight quantum adversary bound.

\begin{definition}[Nonnegative adversary bound] \label{t:adversarydef}
Let $f: \{0,1\}^k \rightarrow \{0,1\}$.  Define 
\begin{equation} \label{e:adversarydef}
\ADV(f) = \max_{\substack{\Gamma \geq 0\\ \Gamma \neq 0}} \frac { \norm{\Gamma} }
  { \max_i \norm{ \Gamma \circ D_i } }
 \enspace ,
\end{equation}
where $\Gamma \circ D_i$ denotes the entrywise matrix product between $\Gamma$ and $D_i$ a zero-one-valued matrix defined by $\bra{x} D_i \ket{y} = 1$ if and only if bitstrings $x$ and $y$ differ in the $i$th coordinate, for $i \in \{1, \ldots, k\}$.
The maximum is over all $2^k \times 2^k$ symmetric matrices $\Gamma$ with nonnegative entries satisfying $\bra{x} \Gamma \ket{y} = 0$ if $f(x) = f(y)$.  
\end{definition}

The motivation for this definition is that $\ADV(f)$ gives a lower bound on the number of queries to the \emph{phase-flip input oracle}
\begin{equation} \label{e:ox}
O_x : \ket{b, i} \mapsto (-1)^{b \cdot x_i} \ket{i}
\end{equation}
required to evaluate $f$ on input $x$.

\begin{theorem}[Adversary lower bound {\cite{Ambainis06polynomial, BarnumSaksSzegedy03adv}}] \label{t:adversaryquerybound}
The two-sided $\epsilon$-bounded error quantum query complexity of function $f$, $Q_{\epsilon}(f)$, is at least $\frac{1 - 2 \sqrt{\epsilon (1-\epsilon)}}{2} \ADV(f)$.
\end{theorem}

\noindent To match the lower bound of \thmref{t:adversaryquerybound}, our goal will be to use $O(\ADV(\varphi))$ queries to evaluate $\varphi$.

\begin{figure}
\begin{center}
\begin{tabular}{|c c| c c|}
\hline \hline
Gate $f$& $\ADV(f)$& Gate $f$& $\ADV(f)$ \\
\hline
0 & 0 
& $\MAJ(x_{[3]}) = (x_1 \wedge x_2) \vee ((x_1 \vee x_2) \wedge x_3)$ & 2 \\
$x_1$ & 1 
& $\EQUAL_3(x_{[3]}) = (x_1 \wedge x_2 \wedge x_3) \vee (\overline{x_1} \wedge \overline{x_2} \wedge \overline{x_3})$ & $3/\sqrt 2$ \\
$x_1 \wedge x_2$ & $\sqrt 2$ 
& $(x_1 \wedge x_2 \wedge x_3) \vee (\overline{x_1} \wedge \overline{x_2})$ & $\sqrt{3+\sqrt 3}$ \\
$x_1 \oplus x_2$ & 2 
& $x_1 \vee (x_2 \oplus x_3)$ & $\sqrt{5}$ \\
$x_1 \wedge x_2 \wedge x_3$ & $\sqrt 3$ 
& $x_1 \oplus (x_2 \wedge x_3)$ & $1+\sqrt 2$ \\
$x_1 \vee (x_2 \wedge x_3)$ & $\sqrt 3$ 
& $\mathrm{EXACT}_{\text{2 of 3}}(x_{[3]}) = \MAJ(x_{[3]}) \wedge (\overline{x_1} \vee \overline{x_2} \vee \overline{x_3})$ & $\sqrt 7$ \\
$(x_3 \wedge x_2) \vee (\overline{x_3} \wedge x_1)$ & 2 
& $x_1 \oplus x_2 \oplus x_3$ & 3 \\
\hline \hline
\end{tabular}
\end{center}
\newsavebox{\tmpbox}
\caption{Binary gates on up to three bits.  Up to equivalences---permutation of inputs, complementation of some or all inputs or output---there are fourteen binary gates on three inputs $x_1, x_2, x_3$.  Adversary bounds $\ADV(f)$ for all functions $f$ on up to four bits have been computed by~{\protect\cite{HoyerLeeSpalek07negativeadvurl}}, and see~{\protect\cite{ReichardtSpalek07formulaurl}}.}  
\label{f:threebitgates}
\end{figure}

\begin{definition}[Adversary-balanced formula] \label{t:adversarybalanceddef}
For a gate $g$ in formula $\varphi$, let $\varphi_g$ denote the subformula of $\varphi$ rooted at $g$.
Define $\varphi$ to be adversary-balanced if for every gate $g$, the adversary lower bounds for its input subformulas are the same; if $g$ has children $h_1, \ldots, h_k$, then $\ADV(\varphi_{h_1}) = \cdots = \ADV(\varphi_{h_k})$.  
\end{definition}
\noindent
\defref{t:adversarybalanceddef} is motivated by a version of an adversary composition result~\cite{Ambainis06polynomial, HoyerLeeSpalek07negativeadv}:

\begin{theorem}[Adversary composition {\cite{HoyerLeeSpalek07negativeadv}}] \label{t:weakadversarycomposition}
Let $f = g \circ (h_1, \ldots, h_k)$, where $\ADV(h_1) = \cdots = \ADV(h_k)$ and $\circ$ denotes function composition.  Then $\ADV(f) = \ADV(g) \ADV(h_1)$.  
\end{theorem}

If $\varphi$ is adversary-balanced, then by \thmref{t:weakadversarycomposition} $\ADV(\varphi_g)$ is the product of the gate adversary bounds along any non-self-intersecting path $\chi$ from $g$ up to an input, $\ADV(\varphi_g) = \prod_{h \in \chi} \ADV(h)$.
Note that $\ADV(\neg f) = \ADV(f)$, so NOT gates can be inserted anywhere in an adversary-balanced formula.

The main result of this paper is

\begin{theorem}[Main result] \label{t:result}
There exists a quantum algorithm that evaluates an adversary-balanced formula $\varphi(x)$ over $\gateset$ using $O(\ADV(\varphi))$ queries to the phase-flip input oracle $O_x$.  After efficient classical preprocessing independent of the input $x$, and assuming $O(1)$-time coherent access to the preprocessed classical string, the running time of the algorithm is $\ADV(\varphi) (\log \ADV(\varphi))^{O(1)}$.
\end{theorem}

From \figref{f:threebitgates}, the adversary bound $\ADV(\MAJ) = 2$.  By \thmref{t:weakadversarycomposition} the adversary bound for the balanced $\MAJ$ formula of depth $d$ is $2^d$.  \thmref{t:balancedmajority} is therefore essentially a corollary of \thmref{t:result} (for the balanced $\MAJ$ formula, coherent access to a preprocessed classical string is not needed).  

\subsection{Optimal span programs for gates in \texorpdfstring{$\gateset$}{S}} \label{s:gatebygate}

In this section, we will substitute specific span programs into \defref{def:spwsize}, in order to prove:

\begin{theorem} \label{t:wsizegates}
Let $\gateset$ be the gate set of \defref{t:gatesetdef}. For every gate $f \in \gateset$,
there exists a span program $P$ computing $f_P = f$, such that the witness size of $P$ (\defref{def:spwsize}) on equal input complexities $\Ui_{i} = 1$ is  
\begin{equation}
\wsize{P} = \ADV(f)
 \enspace .
\end{equation}
$\ADV(f)$ is the adversary bound for $f$ (\defref{t:adversarydef}).
\end{theorem}

\begin{proof}
We analyze five of the fourteen inequivalent binary functions on at most three bits, listed in \figref{f:threebitgates}: $0$ and $x_1$ (both trivial), the $\MAJ$ gate (\claimref{t:majwsize}), the $k$-bit EQUAL$_k$ gate (\claimref{t:equalwsize}), and a certain three-bit function, $g(x) = (x_1 \wedge x_2 \wedge x_3) \vee (\overline{x_1} \wedge \overline{x_2})$ (\claimref{t:gwsize}).

\begin{claim} \label{t:majwsize}
Let $P_{\MAJ}$ be the span program from \exampleref{t:maj3example}.  Then $\wsize{P_{\MAJ}} = 2 = \ADV(\MAJ)$.
\end{claim}

\begin{proof}
Substitute $P_{\MAJ}$ into \defref{def:spwsize}.  Some of the witness vectors are $\ket{w'_{000}} = (1, 0)$, $\ket{w'_{100}} = (1, -1/\sqrt3)$, and $\ket{w_{110}} = (e^{-i \pi/6}, e^{i \pi/6}, 0)$, $\ket{w_{111}} = (1,1,1)/\sqrt3$.
\end{proof}

\begin{claim} \label{t:equalwsize}
Letting $\alpha = \sqrt[4]{k-1}$, the span program
\[
\begin{array}{r@{}c@{}l r@{}*{1}{c@{\ }}c@{}l}
& & & X_J = ( & \{x_1, x_2, \dots, x_k\} & \{\overline{x_1}, \overline{x_2}, \dots, \overline{x_k}\} & ) \\
\noalign{\smallskip}
t=( & 1 & )\,, & v_J=( & \alpha & \alpha & ) \\
\end{array}
\]
computes $\EQUAL_k$ with witness size $\frac{k}{\sqrt{k-1}} = \ADV(\EQUAL_k)$.\footnote{The optimal adversary matrix $\Gamma$ comes from the $2 \times 2 k$ matrix 
$\big(\begin{smallmatrix}
1 & 1 & \cdots & a & a & \cdots \\
a & a & \cdots & 1 & 1 & \cdots
\end{smallmatrix}\big)$, 
where the rows correspond to inputs $0^k$ and $1^k$, and the columns correspond to inputs of Hamming weight 1 then $k-1$, and $a = 1 / (k-1)$.}
\end{claim}

\begin{proof}
Substitute into \defref{def:spwsize}.  
The witnesses are $\ket{w_{0^k}} = \big(0, \frac{\sqrt{k}}{\alpha}\big)$, $\ket{w_{1^k}} = \big(\frac{\sqrt{k}}{\alpha}, 0\big)$ and $\ket{w'_x} = (1)$ for $x \notin \{0^k, 1^k\}$.
\end{proof}

\begin{claim} \label{t:gwsize}
Let $g(x) = (x_1 \wedge x_2 \wedge x_3) \vee (\overline{x_1} \wedge \overline{x_2})$.
Letting $\alpha_1 = \sqrt[4]{1 + \frac 1 {\sqrt 3}}$ and $\alpha_2 = \sqrt{\sqrt3 - 1}$, the span program 
\[
\setcounter{sprows}{2}
\setlength{\spheight}{14pt}
\setlength{\spraise}{6pt}
\begin{array}{r@{}c@{}l r@{}*{2}{c@{\ }}c@{}l}
& & & X_J = ( & \{x_1, x_2\} & \{\overline{x_1}, \overline{x_2} \} & \{x_3\} & ) \\
\noalign{\smallskip}
\spleft{t=} & 1 & \spright{,} & \spleft[2pt]{v_J=}
  & \alpha_1 & \alpha_1 \alpha_2 & 0 & \spright{} \\
& 0 & & & \alpha_2 & 0 & 1 &
\end{array}
\]
computes $g$ with witness size $\sqrt{3 + \sqrt 3} = \ADV(g)$.\footnote{The optimal adversary matrix $\Gamma$ comes from the matrix 
$\Bigg(\begin{smallmatrix}
1 & a & \sqrt{3} \\
1 & a & \sqrt{3} \\
b & 0 & \sqrt{3} 
\end{smallmatrix}\Bigg)$, 
where $a = \Big(\tfrac12(5-\sqrt{13-6\sqrt{3}})\Big)^{1/2}$, $b = \tfrac12 \Big(-1-\sqrt{3}+\sqrt{2(8+\sqrt{3})}\Big)$, and the rows correspond to inputs $011$, $101$, $110$, the columns to inputs $000$, $001$, $111$.} 
\end{claim}

\begin{proof}
By substitution into \defref{def:spwsize}.
\end{proof}

For all the remaining gates in $\gateset$, it suffices to analyze the NOT gate (\lemref{t:not}), and OR and PARITY gates on \emph{unbalanced} inputs (\lemref{t:parityorwsize}).  That is, we allow $\Ui_{1}$ and $\Ui_{2}$ to be different, with $\Ui_{1} / \Ui_{2}, \Ui_{2}/\Ui_{1} = O(1)$.  For functions $b$ and $b'$ on disjoint inputs, $\ADV(b \oplus b') = \ADV(b) + \ADV(b')$, and $\ADV(b \vee b') = \sqrt{\ADV(b)^2 + \ADV(b')^2}$~\cite{bs:q-read-once, HoyerLeeSpalek07negativeadv}; we obtain matching upper bounds for span program witness size.

\begin{lemma} \label{t:parityorwsize}
Consider $f(x,x') = f'(b(x), b'(x'))$, with $f' \in \{ \mathrm{PARITY}, \mathrm{OR} \}$, and $b$ and $b'$ functions on $O(1)$ bits. 
Assume that there exist span programs $P_{b}$ and $P_{b'}$ for $b$ and $b'$ with respective witness sizes $B = \wsize{P_b}$ and $B' = \wsize{P_{b'}}$.  Then there exists a span program $P$ for $f$ with witness size $\wsize{P} = B+B'$ if $f' = \mathrm{PARITY}$, or $\sqrt{B^2+B'^2}$ if $f' = \mathrm{OR}$.  
\end{lemma}

\begin{proof}
Substitute the following span programs with zero constraints into \defref{def:spwsize}:
\[
\begin{array}{r@{}c@{}l r@{}*{1}{c@{\ }}c@{}l}
& & & X_J = ( & \{x_1, \overline{x_2}\} & \{\overline{x_1}, x_2\} & ) \\
\noalign{\smallskip}
P_{\text{PARITY}}: t=( & 1 & )\,, & v_J=( & 1 & 1 & ) \enspace , \\
\noalign{\bigskip}
& & & X_J = ( & \{x_1\} & \{x_2\} & ) \\
\noalign{\smallskip}
P_{\text{OR}}: t=( & 1 & )\,, & v_J=\Big(
  & \frac{\sqrt B}{\sqrt[4]{B^2 + B'^2}}
  & \frac{\sqrt {B'}}{\sqrt[4]{B^2 + B'^2}} & \Big) \enspace .
\end{array}
\]
The witness vectors for PARITY are $\ket{w'_{00}} = (1)$ and $\ket{w_{10}} = (\sqrt{B^2 + B'^2}, 0)$, and the witness vectors for OR are $\ket{w'_{00}} = (1)$, $\ket{w_{10}} = (\sqrt[4]{B^2 + B'^2}, 0)$, and $\ket{w_{11}} = (1,1) \cdot \frac12 \sqrt[4]{B^2 + B'^2}$.
\end{proof}

Then, e.g., the function EXACT$_{\text{2 of 3}}(x_1, x_2, x_3) = \MAJ(x_1, x_2, x_3) \wedge (\overline{x_1} \vee \overline{x_2} \vee \overline{x_3})$, so \lemref{t:parityorwsize} implies a span program for EXACT$_{\text{2 of 3}}$ with witness size $\sqrt{7} = \sqrt{\wsize\MAJ^2 + \wsize{\mathrm{OR}_3}^2}$.
\end{proof}

\begin{remark}
Our procedure for analyzing a function $f$ has been as follows: 
\begin{enumerate}
\item 
First determine a span program $P$ computing $f_P = f$.  The simplest span program is derived from the minimum-size \{AND, OR, NOT\} formula for $f$.
\item 
Next, compute $\wsizex{P}{x}$ for each input $x$, as a function of the variable weights of $P$.  
\item 
Finally, optimize the free weights of $P$
to minimize $\wsize{P} = \max_x \wsizex{P}{x}$.  For example, note that scaling $A_{OJ}$ up helps the true cases in \defref{def:spwsize}, and hurts the false cases; therefore choose a scale to balance the worst true case against the worst false case.  

We respect the symmetries of $f$ during optimization.  On the other hand, if two literals are not treated symmetrically by $f$, then we do not group them together in any grouped input $X_j$.  For example, in \claimref{t:gwsize} we do not group $x_3$ together with $x_1$ and $x_2$ in $X_1$.  
\end{enumerate}
\end{remark}

\begin{remark}
The proof of \thmref{t:wsizegates} uses separate analyses for $\EQUAL_k$, $\MAJ$ and $g$ because the upper bounds from \lemref{t:parityorwsize} for these functions do not match the adversary lower bounds.  For example, from \figref{f:threebitgates} the smallest \{AND, OR, NOT\} formula for $\MAJ$ has five inputs, $\MAJ(x_{[3]}) = (x_1 \wedge x_2) \vee ((x_1 \vee x_2) \wedge x_3)$.  \lemref{t:parityorwsize} therefore gives a span program $P$ for $\MAJ$ with witness size $\wsize{P} = \sqrt 5$.  In fact, optimizing the weights of this $P$ gives a span program with witness size $\sqrt{3+\sqrt{2}} < \sqrt 5$; the worst-case inputs of the read-once formula $(x_1 \wedge x_2) \vee ((x_4 \vee x_5) \wedge x_3)$ do not arise under the promise that $x_4 = x_1$ and $x_5 = x_2$.  
However, this is still worse than the span program $P_{\MAJ}$ of \exampleref{t:maj3example}, with $\wsize{P_{\MAJ}} = 2$.
\end{remark}

\begin{remark}
\lemref{t:parityorwsize} implies that any $\{$AND, OR, NOT$\}$ formula of bounded size has a span program with witness size the square root of the sum of the squares of the input complexities.  We conjecture that this holds even for formulas with size $\omega(1)$; see~\cite{Ambainis06variabletimesearch, AmbainisChildsReichardtSpalekZhang07andor} for special cases.
\end{remark}

\subsection{Span program spectral analysis of \texorpdfstring{$\varphi$}{varphi}} \label{s:spectralanalysis}

\begin{theorem} \label{t:formulaspectrum}
Consider an adversary-balanced formula $\varphi$ on the gate set $\gateset$, with adversary bound $\ADV(\varphi)$.  Let $P$ be the composed span program computing $f_P = \varphi$.
For an input $x \in \{0,1\}^N$, recall the definition of the weighted graph $G_P(x)$ from \thmref{t:weakzeroenergy}; if the literal on an input edge evaluates to true, then delete that edge from $G_P$. 
Let $\tilde G_P(x)$ be the same as $G_P(x)$ except with the weight on the output edge $(a_O, b_O)$ set to $w = \eps_w/\sqrt{\ADV(\varphi)}$ (instead of weight one), where $\eps_w > 0$ is a sufficiently small constant.
Then, 
\begin{itemize}
\item
If $\varphi(x) = 1$, there exists a normalized eigenvalue-zero eigenvector of the adjacency matrix $A_{\tilde G_P(x)}$ with $\Omega(1)$ support on the output vertex $a_O$.  
\item
If $\varphi(x) = 0$, then for some small enough constant $\eps > 0$, 
$A_{\tilde G_P(x)}$ does not have any eigenvalue-$\lambda$ eigenvectors supported on $a_O$ or $b_O$ for $\abs{\lambda} \leq \eps /\ADV(\varphi)$.
\end{itemize}
\end{theorem}

\begin{proof}
The proof of \thmref{t:formulaspectrum} has two parts.  First, we will prove by induction that $\Ui_{g} = O(\ADV(\varphi_g))$.  Then, by considering the last eigenvector constraint, $\lambda a_O = w b_O$, we either construct the desired eigenvector or derive a contradiction, depending on whether $\varphi(x)$ is true or false.  

\subparagraph{Base case}
Consider an input $x_i$ to the formula $\varphi$.  
If $x_i = 1$, then the corresponding input edge $(a_j, b_i)$ is not in $G_P(x)$.  In particular, the input $i$ does not contribute to the expression for $\Uj_{j}(x)$ in Eq.~\eqnref{e:Ujdef}, so $\Ui_{i}$ may be left undefined.  
If $x_i = 0$, then the input edge $(a_j, b_i)$ is in $G_P(x)$.  The eigenvalue-$\lambda$ equation at $b_i$ is $\lambda b_i = a_j$.  For $\lambda = 0$, this is just $a_j = 0$, so let $\abi_i = 1$.  For $\lambda > 0$, this is $\ri_{i} = \lambda \si_{i} = a_j / b_i = \lambda$, so $\si_{i} = \Ui_{i} = 1$.  

\subparagraph{Induction}
Assume that $\abs{\lambda} \leq \eps /\ADV(\varphi)$, for some small enough constant $\eps > 0$.  

Consider a gate $g$.  Let $h_1, \ldots, h_k$ be the inputs to $g$.  Let $\varphi_g$ denote the subformula of $\varphi$ based at $g$.  By \thmref{t:unifiedsolution} and \thmref{t:wsizegates}, the output bound $\Ui_{g}$ satisfies
\begin{equation}
\Ui_{g} \lesssim \ADV(g) \max_i \Ui_{h_i}
 \enspace ,
\end{equation}
or equivalently 
\begin{equation} \label{e:Uigrecurrence}
\Ui_{g} \leq c_1 + \ADV(g) (\max_i \Ui_{h_i})(1 + c_2 \cdot \abs{\lambda} \ADV(\varphi_g))
\end{equation}
for certain constants $c_1, c_2$.  Different kinds of gates give different constants in Eq.~\eqnref{e:Uigrecurrence}, but since the gate set is finite, all constants are uniformly $O(1)$.

Since $\abs{\lambda} \leq \eps / \ADV(\varphi)$, the recurrence Eq.~\eqnref{e:Uigrecurrence} has solution 
\[
\Ui_{g} \leq O\Bigg( \max_\chi \prod_{h \in \chi} \ADV(h) \Big( 1 + \eps' \frac{\ADV(\varphi_h)}{\ADV(\varphi)} \Big) \Bigg)
\enspace , 
\]
where 
the maximum is taken over the choice of $\chi$ a non-self-intersecting path from $g$ up to an input.  Because $\varphi$ is by assumption adversary balanced (\defref{t:adversarybalanceddef}), $\prod_{h \in \chi} \ADV(h) = \ADV(\varphi_{g})$ (\thmref{t:weakadversarycomposition}).  Also, $\prod_{h \in \chi} (1 + \eps' \frac{\ADV(\varphi_h)}{\ADV(\varphi)} ) = O(1)$.  Therefore, the solution satisfies
\begin{equation} \label{e:Uiggoal}
\Ui_{g} = O(\ADV(\varphi_g))
 \enspace .
\end{equation}

\subparagraph{Final amplification step}
Assume $\varphi(x) = 1$.  Then by Eq.~\eqnref{e:Uiggoal}, there exists a normalized eigenvalue-zero eigenvector of the graph $G_P(x)$ with squared amplitude $\abs{a_O}^2 \geq \abi_O = 1/O(\ADV(\varphi))$.  Recall that $w = \eps_w/\sqrt{\ADV(\varphi)}$ is the weight of the output edge $(a_O, b_O)$ of $P$ in $\tilde G_P(x)$, and let $\hat a_O = w a_O$.  The $\lambda = 0$ eigenvector equations for $\tilde G_P(x)$ are the same as those for $G_P(x)$, except with $\hat a_O$ in place of $a_O$.  Therefore, we may take $\abs{\hat a_O}^2 = 1/O(\ADV(\varphi))$, so for a normalized eigenvalue-zero eigenvector of $\tilde G_P(x)$, $\abs{a_O}^2 = \Omega(1)$.  
By reducing the weight of the output edge from $1$ to $w$, we have amplified the support on $a_O$ up to a constant.  

Now assume that $\varphi(x) = 0$.  By \thmref{t:weakzeroenergy}, there does not exist any eigenvalue-zero eigenvector supported on $a_O$.  Also $b_O = 0$ at $\lambda = 0$ by the constraint $\lambda a_O = w b_O$. 
For $\lambda \neq 0$, $\abs{\lambda} \leq \eps / \ADV(\varphi)$, Eq.~\eqnref{e:Uiggoal} implies that in any eigenvalue-$\lambda$ eigenvector for $\tilde G_P(x)$, either $\hat a_O = b_O = 0$ or the ratio $\abs{\hat a_O / b_O} \leq \abs{\lambda} \cdot O(\ADV(\varphi))$, so 
\begin{equation} \label{e:tailratioequation}
\abs{a_O / b_O} \leq c_3 \cdot \frac{\abs{\lambda}}{w} \ADV(\varphi)
\end{equation}
for some constant $c_3$ that does not depend on $w$.
We have not yet used the eigenvector equation at $a_O$, $\lambda a_O = w b_O$.  Combining this equation with Eq.~\eqnref{e:tailratioequation}, we get $w^2 \leq c_3 \lambda^2 \ADV(\varphi) \leq c_3 \eps^2 / \ADV(\varphi)$.  Substituting $w = \eps_w/\sqrt{\ADV(\varphi)}$, this is a contradiction provided we set $\eps_w$ so $\eps_w^2 > c_3 \eps^2$.  Therefore, the adjacency matrix of $\tilde G_P(x)$ cannot have an eigenvalue-$\lambda$ eigenvector supported on $a_O$ or $b_O$.
\end{proof}

\subsection{Quantum algorithm}
\label{s:qalg}

We apply \thmref{t:formulaspectrum} and the Szegedy correspondence between discrete- and continuous-time quantum walks~\cite{Szegedy04walkfocs} to design the optimal quantum algorithm needed to prove \thmref{t:result}.  
The approach is similar to that used for the NAND formula evaluation algorithm of~\cite{ChildsReichardtSpalekZhang07andor}, with only technical differences.  Full details are given in \appref{s:algorithm}.

The main idea is to construct a discrete-time quantum walk $U_x = \tilde{O}_x U_{0^N}$ on the directed edges of $G_P$ whose spectrum and eigenvectors correspond exactly to those of $A_{\tilde G_P(x)}$.  Here $U_{0^N}$ is a fixed unitary operator only depending on the formula graph $A_{\tilde G_P(0^N)}$, which can be implemented efficiently without access to the input $x$, and $\tilde{O}_x$ is defined by 
\begin{equation} \label{e:phasefliporacledef}
\tilde{O}_x \ket{v,w} = \begin{cases}
(-1)^{x_{i(v)}} \ket{v,w} & \text{if $v$ is a leaf} \\
\ket{v,w} & \text{otherwise}
\end{cases}
\end{equation}
where $i(v)$ is the index of the input variable corresponding to the leaf $v$.  One call to $\tilde{O}_x$ can be implemented using one call to the standard phase-flip oracle $O_x$ of Eq.~\eqnref{e:ox}.

Now starting at the output edge $\ket{a_O, b_O}$, run phase estimation~\cite{cemm:qalg} on $U_x$ with precision $\delta_p = O(1/\ADV(\varphi))$ and error $\delta_e$ a small enough constant.  Output ``$\varphi(x) = 1$" iff the output phase is zero.  The query complexity of this algorithm is $O(1/\delta_p) = O(\ADV(\varphi))$.  The first part of \thmref{t:formulaspectrum} implies completeness, because the initial state has constant overlap with an eigenstate of $U_x$ with phase zero.  The second part of \thmref{t:formulaspectrum} implies soundness, because the spectral gap away from zero is greater than the precision $\delta_p$.  

\section{Extensions and open problems} \label{s:extensions}

\thmref{t:result} can be extended in several directions, and there are many open problems including those from~\cite{ChildsReichardtSpalekZhang07andor}.  For example, is the eigenvalue-zero eigenstate useful for extracting witness information?  
We would like to raise several other questions.

\subsection{Four-bit gates} \label{s:fourbitgates}
The gate set $\gateset$ includes all three-bit binary gates.  What about four-bit gates?  
Up to symmetries, there are 208 inequivalent binary functions that depend on exactly four input bits $x_1, \ldots, x_4$.
The functions we have considered so far are listed at the webpage~\cite{ReichardtSpalek07formulaurl}.  To summarize, 
\begin{itemize}
\item
Thirty of the functions can be written as a PARITY or OR of two subformulas on disjoint inputs.  These functions are already included in the gate set $\gateset$ (\defref{t:gatesetdef}).  
\item
For 25 additional functions, we have found a span program with witness size matching the adversary lower bound.  These functions can be added to $\gateset$ 
without breaking \thmref{t:result}.  
\item
For 20 of the remaining functions, we have found a span program with complexity beating the square-root of the minimum \{AND, OR, NOT\} formula size, but not matching the adversary lower bound.  
\end{itemize}

\begin{example}[Threshold 2 of 4] \label{t:threshold2of4example}
In analogy to \exampleref{t:maj3example}, one might consider the span program
\[
\setcounter{sprows}{2}
\setlength{\spheight}{14pt}
\setlength{\spraise}{6pt}
\begin{array}{r@{}c@{}l r@{}*{3}{c@{\ }}c@{}l}
& & & X_J = ( & \{x_1\} & \{x_2\} & \{x_3\} & \{x_4\} & ) \\
\noalign{\smallskip}
\spleft{t=} & 1 & \spright{,} & \spleft[2pt]{v_J=} & 1& 1& 1& 1& \spright[2pt]{\enspace .} \\
& 0 & & & 1 & i & -1 & -i &
\end{array}
\]
This span program computes Threshold$_{\text{2 of 4}}(x_{[4]})$---$\MAJ$ is Threshold$_{\text{2 of 3}}$---but it is not optimal.  Intuitively, the problem is that the different pairs of inputs are not symmetrical.  An optimal span program, with witness size $\sqrt{6}$, is
\[
\setcounter{sprows}{3}
\setlength{\spheight}{18pt}
\setlength{\spraise}{6pt}
\begin{array}{r@{}c@{}l r@{}*{3}{c@{\ }c}c@{\ }c@{}l}
& & & X_J =
( & \{x_1\} & \{x_1\} & \{x_2\} & \{x_2\} &\{x_3\} & \{x_3\} & \{x_4\} & \{x_4\} &) \\
\noalign{\smallskip}
\spleft{t=} & 1 & \spright{,\!\!} & \spleft[2pt]{v_J=} & 1& 1& 1& 1& 1& 1& 1& 1& \spright[2pt]{\!\!\!} \\
& 0 & & & 1&1&1&-1 & i&-i&i&i&\\
& 0 & & & i &-i&i&i & 1&1&1&-1&
\end{array}
\]
It was derived by embedding a four-simplex symmetrically in the $2 \times 2$ unitary matrices, in correct analogy to \exampleref{t:maj3example}.  This embedding gives a span program over an extension ring of $\bf C$ that, following~\cite[Theorem~12]{KarchmerWigderson93span} and~\cite[Prop.~2.8]{BabaiGalWigderson99superpolyspanprogram}, can be simulated by a span program over the base ring.
\end{example}

The Hamming-weight threshold functions $\operatorname{Threshold}_{\text{$h$ of $k$}} : \{0,1\}^k \rightarrow \{0,1\}$ defined by 
\[
\operatorname{Threshold}_{\text{$h$ of $k$}}(x) = \left\{ \begin{array}{cl} 1 & \text{if $\abs{x} \geq h$} \\ 0 & \text{if $\abs{x} < h$} \end{array} \right .
\]
are functions that we currently have an understanding of only for $h \in \{0, 1, k\}$ and a partial understanding of for $h \in \{2, k-1\}$.   Another function of particular interest is the six-bit Kushilevitz function~\cite{HoyerLeeSpalek07negativeadv, Ambainis06polynomial}.  It seems that $k$-bit gates are inevitably going to require more involved techniques to evaluate optimally, for $k$ large enough.  It may well be that four-bit gates are already interesting in this sense.

\subsection{Unbalanced formulas}

Can the restriction that the gates have adversary-balanced inputs be significantly weakened?  So far, we have only analyzed the PARITY and OR gates for unbalanced inputs, in \lemref{t:parityorwsize}.  
For the $\MAJ$ gate, we have found an optimal span program for the case in which only two of the inputs are balanced:
\begin{lemma} \label{t:maj3unbalanced}
Let $f(x,x',x'') = \MAJ(b(x), b'(x'), b''(x''))$ with $b, b', b''$ functions on $O(1)$ bits computed by span programs $P_b, P_{b'}, P_{b''}$ with witness sizes $B = \wsize{P_b} = \ADV(b) = \wsize{P_{b'}} = \ADV(b')$ and $B'' = \wsize{P_{b''}} = \ADV(b'')$.  Let $\beta = B''/B$ and $\alpha = \frac1{2\sqrt{2}}({\sqrt{8 + \beta^2} - \beta})^{1/2}$.  Then there exists a span program $P$ for $f$ with $\wsize{P} = \tfrac12 \big(\sqrt{8 + \beta^2} + \beta \big) B = \ADV(f)$:
\[
\setcounter{sprows}{2}
\setlength{\spheight}{16pt}
\setlength{\spraise}{6pt}
\begin{array}{r@{}c@{}l r@{}*{2}{c@{\ }}c@{}l}
& & & X_J = ( & \{x_1\} & \{x_2\} & \{x_3\} & ) \\
\noalign{\smallskip}
\spleft{t=} & 1 & \spright{,} & \spleft[2pt]{v_J=}
  & \alpha & \alpha & \sqrt{\frac{1}{2} +\beta \alpha ^2} & \spright{\enspace .} \\
& 0 & & & i & -i & 2 \alpha &
\end{array}
\]
\end{lemma}
Therefore, for example, the four-bit gates $\MAJ(x_1, x_2, x_3 \wedge x_4)$ and $\MAJ(x_1, x_2, x_3 \oplus x_4)$ can be added into $\gateset$ without affecting the correctness of \thmref{t:result} (see \secref{s:fourbitgates}).  
However, we do not have an understanding of $\MAJ$ when all three input complexities differ.  In this case, the formula for the adversary lower bound is substantially more complicated, and we do not have a matching span program. 

For other gates, with the exception of PARITY and OR, we know similarly little.  For a highly unbalanced formula with large depth, there is the further problem of whether the formula can be rebalanced without increasing its adversary lower bound too much~\cite{ChildsReichardtSpalekZhang07andor}.  

\subsection{Witness vectors and the adversary bound}

The witnesses in \defref{def:spwsize} have an interesting property
related to a dual version of the adversary bound~\cite{LaplanteMagniez04kolmogorov, SpalekSzegedy04advequivalent}: 
Assume that all $\abs{X_j}=1$ and $\Ui = 1$.  For $x, y \in \{0,1\}^n$ with $f_P(x) = 1$, $f_P(y) = 0$, consider the witnesses $\ket{w_x}$, $\ket{w'_y}$ achieving the minima in Eqs.~\eqnref{e:spctrue}, \eqnref{e:spcfalse}, and let $\ket{w_y} = A^\adjoint \ket{w'_y}$.  Then $\ket{w_x} = \Pi(x) \ket{w_x}$ and $\Pim(y) \ket{w_y} = \ket{w_y}$, so 
\[
\bra{w_x} \Pi(x) \Pim(y) \ket{w_y} = \bra{w_x} A^\adjoint \ket{w'_y} = \braket{t}{w'_y} = 1 \enspace .
\]
Therefore, if we define $p_x(i) = \frac 1 {\norm{\ket{w_x}}^2} \sum_{\substack{j:\ X_j = \{x_i\}\\ \vee\, X_j = \{\overline{x_i}\}}} |\braket j {w_x}|^2$ for each $x$ (for both true and false $f_P(x)$) and for $i \in [n]$, then we get a feasible set of probability distributions for the minimax formulation of the adversary bound~\cite{SpalekSzegedy04advequivalent}.  If $\wsize P = \ADV(f_P)$, then this set of probability distributions is optimal.

In this paper, we only use the adversary bound with nonnegative weights $\ADV(f)$.  
In fact, H\o yer, Lee and \v Spalek showed that Eq.~\eqnref{e:adversarydef} still provides a lower bound on the quantum query complexity even when one removes the restriction that the entries of $\Gamma$ be nonnegative~\cite{HoyerLeeSpalek07negativeadv}.  This more general adversary bound, $\ADV^\pm(f)$ is clearly at least $\ADV(f)$.  
\thmref{t:weakadversarycomposition} is not known to hold for $\ADV^\pm$ composition; however, under the conditions of the theorem, it is known that $\ADV^\pm(f) \geq \ADV^\pm(g) \ADV^\pm(h_1)$.  
For every three-bit function $f$, no advantage is gained by allowing negative weights: $\ADV^\pm(f) = \ADV(f)$.  
For most functions $f$ on four bits, though, $\ADV^\pm(f) > \ADV(f)$~\cite{HoyerLeeSpalek07negativeadvurl}.  Therefore, one gets an asymptotically higher lower bound for formulas with such functions as gates than using $\ADV$.  
However, for no function $f$ with $\ADV(f) < \ADV^{\pm}(f)$ do we have a span program that matches $\ADV^{\pm}(f)$.  The dual formulation of $\ADV^\pm$ cannot be expressed using probability distributions and one therefore cannot hope for a simple correspondence with the witnesses like described above.

Both variants of the adversary bound, $\ADV$ and $\ADV^\pm$, can be expressed as optimal solutions of certain semidefinite programs.  Can one find a semidefinite formulation of span program witness size?

\subsection{Eliminating the preprocessing}

In many cases for $\varphi$, the preprocessing step of algorithm $\algorithm$ can be eliminated.  Because $\varphi$ is an adversary-balanced formula on a known gateset, a decomposition through \thmref{t:szegedization} can be computed separately for each gate of $\gateset$ and then put together at runtime.  This decomposition is \emph{not} the decomposition of \claimref{t:Hamiltoniancoindecomposition}, which involves global properties of $\varphi$ like $\norm{\H'}$.  For an example, see the exactly balanced NAND tree algorithm in~\cite{ChildsReichardtSpalekZhang07andor}.

The decompositions can be combined because all the weights of gate input/output edges are one.  This is quite different from the case of unbalanced NAND trees considered by~\cite{ChildsReichardtSpalekZhang07andor}, in which the weight of an input edge depends on the subformula entering it.

\subsection{Arbitrary \{AND, OR, NOT, PARITY\} formulas}

Some of the conditions on the gates in $\gateset$ (\defref{t:gatesetdef}) can be loosened.  For example, $\gateset$ includes as single gates $O(1)$-size $\{\text{AND, OR, NOT, PARITY}\}$ formulas on inputs that are themselves possibly elements of $\gateset'$.  Let $f$ be such a gate, $f = g \circ (h_1, \ldots, h_k)$ with $g$ an \{AND, OR, NOT, PARITY\} formula of size $O(1)$, and each $h_i$ either the identity or a gate from $\gateset'$.  We have assumed that all the inputs to $f$ have equal adversary bounds.  However, the stated proof works equally well if only each $h_i$ has inputs with equal adversary bounds, provided the inputs to $h_i$ and to $h_{i'}$ have adversary bounds that differ by at most a constant factor.  

We believe that the assumption that $g$ be of size $O(1)$ can also be significantly weakened.  A stronger analysis like that of~\cite{ChildsReichardtSpalekZhang07andor} for ``approximately balanced" \{AND, OR, NOT\} formulas can presumably also be applied with PARITY gates.  We have avoided this analysis to simplify the proofs, and to focus on the main novelty of this paper, the extended gate sets.

For \{AND, OR, NOT, PARITY\} formulas that are not ``approximately balanced," rebalancing will typically be required.  We have not investigated how the formula rebalancing procedures of~\cite{bce:size-depth, bb:size-depth} affect the formula's adversary bound.  In~\cite{ChildsReichardtSpalekZhang07andor}, it sufficed to consider the effect on the formula size, because the adversary bound for any \{AND, OR, NOT\} formula on $N$ inputs is always $\sqrt{N}$.

\subsection{New algorithms based on span programs}

We have begun the development of a new framework for quantum algorithms based on span programs.  In this paper, we have only composed bounded-size span programs evaluating functions each on $O(1)$ bits.  An intriguing question is, do there exist interesting quantum algorithms based directly on asymptotically large span programs?  Some candidate problems may be found in~\cite{BabaiGalWigderson99superpolyspanprogram, BeimelGalPaterson96spanprogram}, although note that the quantum algorithm works for span programs over $\bf{C}$ that need not be monotone.  

\section*{Acknowledgements}

We thank Troy Lee for pointing out span programs to us.  
B.R.\ would like to thank Andrew Childs, Sean Hallgren, Cris Moore, David Yonge-Mallo and Shengyu Zhang for helpful conversations.  

\bibliographystyle{alpha}
\bibliography{andor}

\appendix

\section{Proof of \texorpdfstring{\thmref{t:unifiedsolution}}{Theorem \ref*{t:unifiedsolution}}} \label{s:wsizeproof}

In this section we will prove \thmref{t:unifiedsolution}.  Before beginning, though, let us show that the alternative expressions for the span program witness size $\wsizex P x$ of \defref{def:spwsize} are equivalent, so $\wsize P$ is well defined.  It will also be useful to derive several alternative expressions for $\wsizex P x$.  

\begin{definition}[Additional matrix notations] \label{def:additionalmatrixnotations}
Recall \defref{def:matrixnotations}.  
Let $\bra \o = A_{OJ} = \bra t A$ and let $\ACJ = (1 - \ketbra t t) A$, so $A = \ACJ + \ketbra{t}{\o}$ (the matrix $A_{CJ}$ is $\ACJ$ with range restricted to $1 - \ketbra t t$).  Let $\Delta = \ACJ \Pi (\ACJ \Pi)^+$ be the projection onto the range of $\ACJ \Pi$, 
and let $\Deltam = 1-\Delta$.  For a matrix $M$ and a projection $\Theta$, let $M_\Theta$ denote the restriction $\Theta M \Theta$ of $M$ to the range of $\Theta$.
For a vector $\tilde{m}_J$, we will commonly write $\tilde{m} = \sum_{j \in J} \tilde{m}_j \ketbra j j$ for the diagonal matrix with diagonal entries $\tilde{m}_J$.  
\figref{fig:notation} summarizes the matrices used in this section.  
\end{definition}

We will use several times the following estimates for pseudoinverse norms: 

\begin{claim} \label{t:hyperplanes}
For matrices $A$ and $B$ with $\Range(B^\adjoint) \subseteq \Range(A)$ (i.e., $B = B A A^+$), $\norm{A (B A)^+} \leq \norm{B^+}$ and $\norm{(B A)^+} \leq \norm{A^+} \norm{B^+}$.  
\end{claim}
\begin{proof}
Since $A (B A)^+ = A (B A)^+ B B^+ = \big[A (BA)^+ (BA) A^+\big] B^+$ and the bracketed term is a projection, $\norm{A (B A)^+} \leq \norm{B^+}$.  Then also $\norm{(B A)^+} = \norm{ A^+ A (B A)^+ } \leq \norm{A^+} \norm{B^+}$.  
\end{proof}

\begin{lemma} \label{t:spcinterpretation}
For $S$ any positive-definite, diagonal matrix, let 
\begin{align} \label{e:wsizexSdef}
\wsizexS P x S &= \left\{ \begin{array}{cc}
\smallskip
\min_{\ket w : A \Pi \ket{w} = \ket{t}} \norm{S \ket{w}}^2 & \text{if $f_P(x) = 1$} \\
\min_{\substack{\ket{w'} : \braket{t}{w'} = 1 \\ \Pi A^\adjoint \ket{w'} = 0}} \norm{S A^\adjoint \ket{w'}}^2 & \text{if $f_P(x) = 0$}
\end{array}\right. 
\intertext{Then if $f_P(x) = 1$,
}
\label{e:spcinterpretationtrue}
\wsizexS P x S
&= { \min_{\substack{\ket w : \bra{w} \Pi \ket{w} = 1 \\ \ACJ \Pi S^{-1} \ket{w} = 0}} \abs{ \bra{\o} \Pi S^{-1} \ket{w} }^{-2} } \\
&= \norm{ (A \Pi S^{-1})^+ \ket{t} }^2 \nonumber \\
&= \norm{ (\Pi - (\ACJ \Pi S^{-1})^+ \ACJ \Pi S^{-1}) S^{-1} \ket{\o} }^{-2} \nonumber \\
\intertext{and, if $f_P(x) = 0$, }
\label{e:spcinterpretationfalse}
\wsizexS P x S
&= { \min_{\substack{\ket{w'} : \norm{S A^\adjoint \ket{w'}} = 1 \\ \Pi A^\adjoint \ket{w'} = 0}} \abs{ \braket{t}{w'} }^{-2} } \\
&= \norm{ \big(1 + (\Pim (A S)^+ A S - 1)^+ \Pi \big) (A S)^+ \ket{t} }^{-2} \nonumber \\
&= \norm{(1-(\Deltam \ACJ S)^+ \Deltam \ACJ S) S (1 - \ACJ^\adjoint (\Pi \ACJ^\adjoint)^+) \ket{\o}}^2 \nonumber
 \enspace .
\end{align}
Moreover, $\ket{w_S^*} = {\arg \min}_{\ket w : A \Pi \ket{w} = \ket{t}} \norm{S \ket{w}}^2 = S^{-1} (A \Pi S^{-1})^+ \ket t$ has norm $\norm{\ket{w_S^*}} = O(1)$ and 
\begin{equation*}
\ket{w_S'^*} = 
{\arg \min}_{\substack{\ket{w'} : \braket{t}{w'} = 1 \\ \Pi A^\adjoint \ket{w'} = 0}} \norm{S A^\adjoint \ket{w'}}^2
= \ket t - (\Pi \ACJ^\adjoint)^+ \Pi \ket{\o} - (S \ACJ^\adjoint \Deltam)^+ S (1 - \ACJ^\adjoint (\Pi \ACJ^\adjoint)^+) \ket{\o}
\end{equation*}
has norm $\norm{\ket{w_S'^*}} = O(1)$.

In particular, the two different expressions for $\wsizex P x = \wsizexS P x {\sqrt{\Uj}}$ in \defref{def:spwsize} are equivalent, so $\wsizex P x$ and $\wsize P$ are well defined.
\end{lemma}

\begin{proof}
Assume $f_P(x) = 1$.  
That $\min_{\ket w : A \Pi S^{-1} \ket{w} = \ket{t}} \norm{\ket{w}}^2
= \min_{\substack{\ket w : \bra{w} \Pi \ket{w} = 1 \\ \ACJ \Pi S^{-1} \ket{w} = 0}} \abs{ \bra{\o} \Pi S^{-1} \ket{w} }^{-2}$ is immediate.  
In general, ${\arg \min}_{\ket{x} : M \ket{x} = \ket{b}} \norm{ \ket{x} } = M^+ \ket{b}$; therefore, $\min_{\ket w : A \Pi S^{-1} \ket{w} = \ket{t}} \norm{\ket{w}}^2 = \norm{ (A \Pi S^{-1})^+ \ket{t} }^2$.  
By basic geometry, $\arg \max_{\substack{\ket{w} : \norm{ \Pi \ket{w} } = 1 \\ \ACJ \Pi S^{-1} \ket{w} = 0}} \abs{ \bra{\o} \Pi S^{-1} \ket{w} } \propto \big(1 - (\ACJ \Pi S^{-1})^+ (\ACJ \Pi S^{-1}) \big) \Pi S^{-1} \ket{\o}$, i.e., is proportional to the projection of $\Pi S^{-1} \ket{\o}$ onto the space orthogonal to the range of $S^{-1} \Pi \ACJ^\adjoint$.  Eq.~\eqnref{e:spcinterpretationtrue} follows.  

Next assume $f_P(x) = 0$.  That $\min_{\substack{\ket{w'} : \braket{t}{w'} = 1 \\ \Pi A^\adjoint \ket{w'} = 0}} \norm{S A^\adjoint \ket{w'}}^2 = \min_{\substack{\ket{w'} : \norm{S A^\adjoint \ket{w'}} = 1 \\ \Pi A^\adjoint \ket{w'} = 0}} \abs{ \braket{t}{w'} }^{-2}$ is immediate.  
Now, without loss of generality, $\ket{t} \in \Range(A) = \Range(A S)$, since otherwise $f_P$ is false on every input.  Therefore, $\braket{t}{w'} = \bra{t} (S A^\adjoint)^+ (S A^\adjoint) \ket{w'} = \bra{t} (S A^\adjoint)^+ \ket w$ if $\ket w= S A^\adjoint \ket{w'}$.  We want to find the length-one vector $\ket w$ that is in the range of $S A^\adjoint$ and also of $\Pim$, and that maximizes $\abs{ \bra{t} (S A^\adjoint)^+ \ket w }^2$.  The answer is clearly the normalized projection of $(A S)^+ \ket{t}$ onto the intersection $\Range(S A^\adjoint) \cap \Range(\Pim)$.  In general, given two projections $\Pi_1$ and $\Pi_2$, the projection onto the intersection of their ranges can be written $1 - (\Pi_1 \Pi_2 - 1)^+ (\Pi_1 \Pi_2 - 1)$.  Substituting $\Pi_1 = \Pim$ and $\Pi_2 = (A S)^+ A S$ gives the second claimed expression. 

Finally, we show that $\min_{\substack{\ket{w'} : \braket{t}{w'} = 1 \\ \Pi A^\adjoint \ket{w'} = 0}} \norm{S A^\adjoint \ket{w'}}^2 = \norm{(1-(\Deltam \ACJ S)^+ \Deltam \ACJ S) S (1 - \ACJ^\adjoint (\Pi \ACJ^\adjoint)^+) \ket{\o}}^2$.
Since $f_P(x)$ is false, $\ket{t}$ does not lie in the span of the true grouped input vectors, $\ket t \notin \Range(A \Pi)$, or equivalently $\Pi \ket\o \in \Range(\Pi \ACJ^\adjoint)$.  Therefore, there exists a vector $\ket{w'} = \ket{t} + \ket{b_C}$ that is orthogonal to the span of the true columns of $A$ and has inner product one with $\ket{t}$.  Any such $\ket{b_C}$ has the form 
\[
\ket{b_C} = - (\Pi \ACJ^\adjoint)^+ \Pi \ket{\o} + \Deltam \ket{v} \enspace ,
\]
where $\ket{v}$ is an arbitrary vector with $\braket{t}{v} = 0$.  We want to choose $\ket{v}$ to minimize the squared length of 
\begin{align}
\Pim S A^\adjoint \ket{w'} &= 
S \Pim (\ket{\o} + \ACJ^\adjoint \ket{b_C}) \nonumber \\
&= S (1 - \ACJ^\adjoint (\Pi \ACJ^\adjoint)^+) \ket{\o} + S \ACJ^\adjoint \Deltam \ket{v} 
 \enspace .
\end{align}
The answer is clearly the squared length of $S (1 - \ACJ^\adjoint (\Pi \ACJ^\adjoint)^+) \ket{\o}$ projected orthogonal to the range of $S \ACJ^\adjoint \Deltam$, as claimed.  This corresponds to setting $\ket{v} = -(S \ACJ^\adjoint \Deltam)^+ S (1 - \ACJ^\adjoint (\Pi \ACJ^\adjoint)^+) \ket{\o}$.

The norms of $\ket{w_S^*}$ and $\ket{w_S'^*}$ are bounded using \claimref{t:hyperplanes}.
\end{proof}

\begin{remark}
The expressions for witness size in Eqs.~\eqnref{e:spcinterpretationtrue} and~\eqnref{e:spcinterpretationfalse} look quite different depending on whether $f_P(x) = 1$ or $f_P(x) = 0$, with the latter case being more complicated.  It can be seen, though, that $\abs{\wsizex P x - \wsizex {P^\adjoint} x} = O(1)$ for any fixed span program $P$, where $P^\adjoint$ is the dual span program described in \secref{s:dualspanprogram} with $f_{P^\adjoint}(x) = \neg f_P(x)$.  
\end{remark}

Let us now show that if $\Uj_j' \lesssim \Uj_j$ for all $j \in J$, then $\wsizexS P x {\sqrt{\Uj'}} \lesssim \wsizex P x$ (\lemref{t:wsizefudgefactors}).  This will be useful in showing that $\wsizex P x$ is a rough upper bound on the exact expressions that we will derive in the sections below.  

\begin{remark} \label{t:wsizexmonotone}
From \defref{def:spwsize}, it is immediate that $\wsizex P x$ is monotone increasing in each input complexity $\Ui_{i}$.  
\end{remark}

\begin{lemma} \label{t:wsizefudgefactors}
Let $S$ and $T$ be any positive-definite diagonal matrices.  Then 
\begin{align}
\label{e:wsizefudgefactorstimes}
\wsizexS P x {S \sqrt{1+T}} &\leq \wsizexS P x S \cdot (1 + \norm{T}) \;, \\
\label{e:wsizefudgefactorsplus}
\wsizexS P x {\sqrt{S^2 + T^2}} &\leq \wsizexS P x S + O(\norm{T}^2)
 \enspace .
\end{align}
In particular, if $\Uj_J'$ is such that $\Uj_j' \lesssim \Uj_j$ for all $j \in J$, then $\wsizexS P x {\sqrt{\Uj'}} \lesssim \wsizexS P x {\sqrt{\Uj}} = \wsizex P x$.  
\end{lemma}
\begin{proof}
Eq.~\eqnref{e:wsizefudgefactorstimes} is immediate from the definition in Eq.~\eqnref{e:wsizexSdef}.  To derive Eq.~\eqnref{e:wsizefudgefactorsplus}, first note that $\norm{\sqrt{S^2 + T^2} \ket v}^2 = \norm{S \ket v}^2 + \norm{T \ket v}^2$.  Then when $f_P(x) = 1$, 
\begin{align*}
\min_{\ket w : A \Pi \ket w = \ket t} (\norm{S \ket w}^2 + \norm{T \ket w}^2)
&\leq
\norm{S \ket{w_S^*}}^2
+ \norm{T \ket{w_S^*}}^2 \\
&= 
\wsizexS P x S
+ O(\norm{T}^2) \enspace ,
\end{align*}
where we have used that $\ket{w_S^*} = {\arg \min}_{\ket w  : A \Pi \ket w = \ket t} \norm{S \ket w}^2$ has norm $\norm{\ket{w_S^*}} = O(1)$ by \lemref{t:spcinterpretation}.  The argument when $f_P(x) = 0$ is similar:
\begin{equation*}
\min_{\substack{\ket{w'} : \braket{t}{w'} = 1 \\ \Pi A^\adjoint \ket{w'} = 0}} (\norm{S A^\adjoint \ket{w'}}^2 + \norm{T A^\adjoint \ket{w'}}^2)
\leq 
\wsizexS P x S + \norm{T A^\adjoint \ket{w_S'^*}}^2 \enspace ,
\end{equation*}
where 
$
\ket{w_S'^*} = {\arg \min}_{\substack{\ket{w'} : \braket{t}{w'} = 1 \\ \Pi A^\adjoint \ket{w'} = 0}} \norm{S A^\adjoint \ket{w'}}^2
$
has $O(1)$ norm by \lemref{t:spcinterpretation}.  
\end{proof}

\begin{figure}
\begin{center}
\begin{tabular}{r@{ }l l}
$A$ & $= \sum_j \ketbra{v_j}{j} = \ACJ + \ketbra t \o$
	& span program matrix \\
$\ACJ$ & $=  (1 - \ketbra t t) A$
	& ``constraint" part of the span program \\
$\bra\o$ & $= A_{OJ} = \bra t A$
	& ``output" row of the span program \\
$\Pi$ & $= \sum_{\text{true $j$}} \ketbra j j$
	& projection onto true grouped inputs \\
$\Pim$ & $= 1 - \Pi$
	& projection onto false grouped inputs \\
\noalign{\medskip}
$\abj$ & $= \sum_j \abj_{j} \ketbra j j$ & grouped input squared supports at $\lambda = 0$ \\
$\ri$ & $= \sum_i \ri_{i} \ketbra i i$
 	& input ratios $\ri_{i} = {a_{j(i)}} /{b_i}$ \\
$\rj$ & $= ({A_{IJ}}^\adjoint \ri^{-1} A_{IJ} - \lambda)^{-1}
	= \sum_j \rj_j \ketbra j j$
	& grouped input ratios \\
$\sj$ & $= -\frac1{\lambda} \Pi \rj^{-1} + \frac1{\lambda} \Pim \rj
	= \sum_j \sj_j \ketbra j j$
	& grouped input ratio multipliers \\
$\Uj$ & $= \sum_j \Uj_j  \ketbra j j$ & grouped input complexities, $\frac 1 {\abj_{j}}, \sj_{j} \lesssim \Uj_{j}$ \\
\noalign{\medskip}
$\cp$ & $= \ACJ \sj^{-1/2} \Pi$, $\cpp = \cp \cp^\adjoint$
	& true constraints scaled down by $\sqrt{\sj}$ \\
$\cm$ & $= \ACJ \sj^{1/2} \Pim$, $\cmm = \cm \cm^\adjoint$
	& false constraints scaled up by $\sqrt{\sj}$ \\
$\Delta$ & $= \cp \cp^+ = \ACJ \Pi (\ACJ \Pi)^+$
	& projection onto the range of true constraints \\
$\Deltam$ & $= 1 - \Delta$
	& complementary projection \\
$X$ & $= -\lambda^2 (1 + \frac 1 \lambda \ACJ \rj \ACJ^\adjoint)
	= \cpp - \lambda^2 (\cmm + 1)$
	& matrix to be inverted \\
$\XDeltaSchur$ & $= \DeltamXDeltam
	- \Deltam X (\DeltaXDelta)^\resinv X \Deltam$
	& Schur complement of $\DeltaXDelta$ in $X$ \\
$V$ & $= -\lambda^2 \cm^\adjoint \Deltam X^{-1} \Deltam \cm
	= -\lambda^2 \cm^\adjoint \XDeltaSchur^\resinv \cm$
	& a part of the inverse of $X$ \\
$\VV$ & $= \cm^\adjoint \Deltam (\dcmm + 1)^{-1} \Deltam \cm$
	& a useful $O(1)$ matrix, $V - \VV = O(\lambda^2 \norm{\Uj}^2)$ \\
\end{tabular}
\end{center}
\caption{Matrices used in the proof of \thmref{t:unifiedsolution}.}
\label{fig:notation}
\end{figure}

\subsection{Quantitative eigenvalue-zero spectral analysis of \texorpdfstring{$A_{G_P}$}{A\_\{G\_P\}}}
\label{s:lambda0case}

\thmref{t:weakzeroenergy} can be strengthened to put 
quantitative lower bounds on $\abi_O$, the achievable squared magnitude, in a unit-normalized eigenvalue-zero eigenvector, on the output node either $a_O$ if $f_P(x) = 1$ or $b_O$ if $f_P(x) = 0$:

\def\W {W}
\begin{theorem} \label{t:strongzeroenergy}
For an input $x \in \{0,1\}^n$, define a weighted graph $G_P(x)$ by deleting from $G_P$ the edges $(a_j, b_i)$ if the $i$th literal in $X_j$ is true.  
Also let 
\begin{equation}
\W = 
1 + \sum_{j \in J} \sum_{i \in I_j} \Big( \frac 1 {\abi_i} - 1 \Big) \left( \ketbra {b_i} {b_i} + \ketbra {a_{j}}{a_{j}} \right) \enspace .
\end{equation}
Consider all the eigenvalue-zero eigenvector equations of the weighted adjacency matrix $A_{G_P(x)}$, except for the constraint at $a_O$, i.e., Eqs.~\eqnref{e:zeroenergyequations} except~\eqnref{e:zeroenergyequationsaO}.  
By \thmref{t:weakzeroenergy}, these equations have a solution $\ket \psi$ with $\braket {a_O} \psi \neq 0$ if and only if $f_P(x) = 1$, and have a solution $\ket \psi$ with $\braket {b_O} \psi \neq 0$ if and only if $f_P(x) = 0$.  In fact, the solution $\ket \psi$ can be chosen so that the normalized square overlap
\begin{equation}
\abi_O \equiv \frac{ \abs{(\bra{a_O} + \bra{b_O}) \ket \psi}^2 } { \bra \psi \W \ket \psi } \geq \frac 1 {\wsizex P x + \mathrm{constant}} \enspace ,
\end{equation}
where the constant may depend on $P$ but is independent of the $\abi_I$, and $\wsizex P x$ is as defined in \defref{def:spwsize}, with $1/\abi_i \leq \Ui_{i}$ for all $i$.  
\end{theorem}

\begin{remark}
Note that \thmref{t:strongzeroenergy} implies the $\lambda = 0$ portion of \thmref{t:unifiedsolution}.  The weights in $\W$ mean that, e.g., setting $b_i = 1$ adds $\bra{b_i} \W \ket{b_i} = 1/\abi_i$ to the squared normalization factor.  
\end{remark}

\begin{proof}[Proof of \thmref{t:strongzeroenergy}]
Recall \figref{f:gadget}.  The vertex $a_j$ is a shared output node of all the inputs $i \in I_j$.  As in the proof of \thmref{t:weakzeroenergy}, Eq.~\eqnref{e:zeroenergyequationsbI} implies that $a_j$ can be nonzero only if grouped input $j$ is true, i.e., if all $i \in I_j$ evaluate to true.  

For $j \in J$, define $\abj_j$ by 
\begin{equation}
\abj_j = \left\{
\begin{array}{cc} 
\left( 1 + \sum_{i \in I_j} \big( \frac1{\abi_i} - 1\big ) \right)^{-1} & \text{if $j$ is true} \\
\sum_{\text{false $i \in I_j$}} \abi_i & \text{if $j$ is false}
\end{array}\right. 
\end{equation}
From \defref{def:subformulacomplexity}, $1/\abj_j \leq \Uj_{j}$.  Roughly speaking, for each $j$, the vertices $b_i$ for $i \in I_j$ can be treated as just a single input vertex with associated weight $\abj_j$ in $\W$.  Precisely, if $j$ is true, then $\bra{a_j} \W \ket{a_j} = 1/ \abj_j$.  And if $j$ is false, then the $b_{I_j}$ coefficients appear in Eq.~\eqnref{e:zeroenergyequationsaJ} only in the quantity $\bra{j} {A_{IJ}(x)^\adjoint} \ket{b_{I_j}} = \sum_{\text{false $i \in I_j$}} b_i$.  In order to minimize the weighted squared norm $\bra{b_{I_j}} \W \ket{b_{I_j}} = \sum_{i \in I_j} \abs{b_i}^2 / \abi_i$ for any fixed value of $\bra{j} {A_{IJ}(x)^\adjoint} \ket{b_{I_j}}$, each $b_i$ for $i$ false should be set proportional to $\abi_i$ (by Cauchy-Schwarz), so
\begin{equation}
\min_{\ket{b_{I_j}} : \bra{j} {A_{IJ}(x)^\adjoint} \ket{b_{I_j}} = 1} \bra{b_{I_j}} \W \ket{b_{I_j}}
 = \frac 1 {\abj_j} \enspace .
\end{equation}

Let $\abj = \sum_j \abj_j \ketbra{j}{j}$.  

\begin{description}
\item[Case $f_P(x) = 1$:]
When $f_P(x) = 1$, set $a_j = 0$ for all false grouped inputs $j$.  Set the other $a_j$ so as to maximize the magnitude of $-a_O =  A_{OJ} \ket{a_J} = \bra{\o} \Pi \ket{a_J}$, such that $\ACJ \ket{a_J} = \ACJ \Pi \ket{a_J} = 0$ and $\bra{a_J} \W \ket{a_J} = \bra{a_J} \Pi \abj^{-1} \ket{a_J} = 1$ (Eqs.~\eqnref{e:zeroenergyequationsbC} and~\eqnref{e:zeroenergyequationsbO} at $\lambda = 0$).  
Now, changing variables to $\ket w = \abj^{-1/2} \ket{a_J}$, 
\begin{equation}
\begin{split}
\abs{a_O}^2 
&= \max_{\substack{\ket w : \bra w \Pi \ket w = 1 \\ \ACJ \Pi \abj^{1/2} \ket w = 0}} \abs{\bra\o \Pi \abj^{1/2} \ket w}^2 \\
&= 1/\wsizexS P x {\abj^{-1/2}} \\
&\geq 1/\wsizex P x
 \enspace ,
\end{split}
\end{equation}
using Eq.~\eqnref{e:spcinterpretationtrue}, $1/\abj_j \leq \Uj_{j}$ and the monotonicity of $\wsizex P x$ (\remref{t:wsizexmonotone}).  
Finally, dividing by $(1+\abs{a_O}^2)$ so that the total norm is one, gives
\begin{equation}
\abi_O = \frac{\abs{a_O}^2}{1+\abs{a_O}^2} \geq \frac1{\wsizex{P}{x} + 1} \enspace .
\end{equation}

\item[Case $f_P(x) = 0$:]
When $f_P(x) = 0$, for each true grouped input $j$ set $b_i = 0$ for $i \in I_j$.  For each false $j$, set $b_i = 0$ for true $i \in I_j$ and set $b_i = f_j \abi_i / \sum_{\text{false $i' \in I_j$}} \abi_{i'}$.  Choose $\ket f$ to maximize $\abs{b_O}^2$ such that $\bra f \Pim \abj^{-1} \ket f = 1$ and, by Eq.~\eqnref{e:zeroenergyequationsaJ} at $\lambda = 0$, $\ket \o b_O + \ACJ^\adjoint \ket{b_C} + \Pim \ket f = 0$.  Equivalently, writing $\ket{w'} = \left(\begin{smallmatrix} b_O \\ b_C \end{smallmatrix}\right)$ so $b_O = \braket{t}{w'}$, we are constrained that $\Pim \ket f = -A^\adjoint \ket{w'}$, i.e., 
\begin{equation}\begin{split}
\abs{b_O}^2 
&= \max_{\substack{ \ket{w'} : \norm{\abj^{-1/2} A^\adjoint \ket{w'}} = 1 \\ \Pi A^\adjoint \ket{w'} = 0}} \abs{ \braket{t}{w'} }^2 \\
&= 1/\wsizexS P x {\abj^{-1/2}} \\
&\geq 1/{\wsizex P x}
\end{split}\end{equation}
by Eq.~\eqnref{e:spcinterpretationfalse} and the monotonicity of $\wsizex P x$.  
The constructed state has weighted squared norm $\bra \psi \W \ket \psi
= 1 + \norm{\ket{w'^*}}^2$, where $\ket{w'^*} = \arg \max_{\substack{\ket{w'}: \norm{\abj^{-1/2} A^\adjoint \ket{w'}} = 1\\ \Pi A^\adjoint \ket{w'} = 0}} \abs{\braket t {w'}}^2$.  Normalizing,  
\begin{equation}\begin{split}
\frac1{\abi_O} 
&= \frac{1 + \norm{\ket{w'^*}}^2}{\abs{b_O}^2} \\
&\leq \wsizex{P}{x} + \frac{\norm{\ket{w'^*}}^2}{\abs{b_O}^2} \enspace .
\end{split}\end{equation}
It remains to show that $\norm{\ket{w'^*} / b_O} = O(1)$.  Indeed, $\frac1{\braket t {w'^*}} \ket{w'^*} = {\arg \min}_{\substack{\ket{w'} : \braket{t}{w'} = 1 \\ \Pi A^\adjoint \ket{w'} = 0}} \norm{\abj^{-1/2} A^\adjoint \ket{w'}}^2 = \ket{w_{\abj^{-1/2}}'^*}$ has $O(1)$ norm by \lemref{t:spcinterpretation}.  
\qedhere
\end{description}
\end{proof}

\subsection{Small-eigenvalue spectral analysis of \texorpdfstring{$A_{G_P}$}{A\_\{G\_P\}}}

\begin{theorem} \label{t:strongsmallenergy}
For a span program $P$ and input $x$, given $\si_{I}$ with $0 < \si_{i} \leq \Ui_{i}$ for all $i \in I$, let $\si = \sum_i \si_{i} \ketbra i i$ and $\ri = \sum_i \ri_{i} \ketbra i i = -\frac1\lambda \Pi \si^{-1} + \lambda \Pim \si$.  Assume that $0 < \lambda \leq \eps / \si_{i}$ for a small enough constant $\eps > 0$ to be determined and for all $i \in I$.  
Then the equations 
\begin{subequations} \label{e:gadget}
\begin{align}
b_I &= \ri^{-1} A_{IJ} a_J \label{e:gadgetbI} \\
\lambda b_O &= a_O + A_{OJ} a_J \label{e:gadgetbO} \\
\lambda b_C &= A_{CJ} a_J \label{e:gadgetbC} \\ 
\lambda a_J &= {A_{OJ}}^\adjoint b_O + {A_{CJ}}^\adjoint b_C + {A_{IJ}}(x)^\adjoint b_I \label{e:gadgetaJ}
\end{align}
\end{subequations}
have a solution with $a_O, b_O \neq 0$.  Moreover, if $\ri_{O} = a_O / b_O$ and $\si_{O}$ is defined as $-1/(\lambda \ri_{O})$ or $\ri_{O} / \lambda$ if $f_P(x)$ is true or false, respectively, then 
\begin{equation}
0 < \si_{O} \lesssim \wsizexS P x {\sqrt{\Uj}}
 \enspace ,
\end{equation}
where the grouped input complexities $\Uj_{J}$ are defined in terms of $\Ui_{I}$ in \defref{def:groupedinputcomplexity}.  
\end{theorem}

\begin{proof}

Similarly to the argument in \appref{s:lambda0case}, it will be useful to define a ``grouped input ratio" $\rj_{j}$ so that, roughly speaking, the vertices $b_i$ for $i \in I_j$ can be treated as just a single input vertex.

\begin{definition}[Grouped input ratios]For $j \in J$, let $\rj_{j} = (- \lambda + \sum_{i \in I_j} \ri_{i}^{-1} )^{-1}$, and let $\rj = \sum_j \rj_{j} \ketbra j j = ({A_{IJ}}^\adjoint \ri^{-1} A_{IJ} - \lambda)^{-1}$.  Like an input ratio $\ri_{i}$, $\rj_{j}$ is large and negative if $j$ is true, and small and positive if $j$ is false.  Therefore let $\sj_{j} = -1 / (\lambda \rj_{j})$ if $j$ is true, and $\sj_{j} = \rj_{j} / \lambda$ if $j$ is false.  Let $\sj = \sum_j \sj_{j} \ketbra j j = -\frac1{\lambda} \Pi \rj^{-1} + \frac1{\lambda} \Pim \rj$, so $\rj = -\frac1{\lambda} \Pi \sj^{-1} + \lambda \Pim \sj$.  
\end{definition}

Before proceeding, we need to establish that $\rj$ and $\sj$ are well defined.  

\begin{lemma} \label{t:groupedinputs}
Assume that $0 < \lambda \leq \eps / \Ui_{i}$ for a small enough constant $\eps > 0$ and for all $i \in I$.  Then $\rj = ({A_{IJ}}^\adjoint \ri^{-1} A_{IJ} - \lambda)^{-1}$ exists, so $\sj$ exists as well.  Moreover, for each grouped input $j \in J$, $\sj_{j} \lesssim \Uj_{j}$.  
\end{lemma}

\begin{proof}
By definition, 
\[
\rj_{j} 
= ( -\lambda +  \sum_{i \in I_j} \ri_{i}^{-1} )^{-1}
= \Big( -\lambda -  \lambda \sum_{\text{true $i \in I_j$}} \si_{i} + \frac1{\lambda} \sum_{\text{false $i \in I_j$}} \si_{i}^{-1} \Big)^{-1} \enspace .
\]

If all inputs in $I_j$ are true, then 
\[
\sj_{j}
= -1 / (\lambda \rj_{j}) 
= 1 + \sum_{i \in I_j} \si_{i} 
 \enspace ,
\]
so $1 \leq \sj_{j} \leq 1 + \sum_{i \in I_j} \Ui_{i} \leq 1 + \Uj_{j}$.

Now assume at least one input in $I_j$ is false.  The true terms can be upper-bounded by $\lambda \sum_{\text{true $i \in I_j$}} \si_{i} \leq \lambda \sum_{i \in I_j} \Ui_{i} \leq \abs{I_j} \eps$.  On the other hand, if $i$ is false then $(\lambda \si_{i})^{-1} \geq (\lambda \Ui_{i})^{-1} \geq 1 / \eps$.  
Therefore, $\sj_{j} > 0$, and we also get $\sj_{j} \leq \Uj_{j} ( 1 + \eps' \lambda \max_{\text{false $i \in I_j$}} \Ui_{i} )$ for a constant $\eps'$.  
\end{proof}

Now we will solve for the output ratio $\ri_{O}$ using Eqs.~(\ref{e:gadget}b-d).  
Letting $\si_{O} = \ri_{O} / \lambda$ in case $f_P(x) = 0$, or $\si_{O} = - 1 / (\lambda \ri_{O})$ in case $f_P(x) = 1$, we aim to show that $0 < \si_{O} \lesssim \wsizexS P x {\sqrt{\sj}}$.  This will prove \thmref{t:strongsmallenergy} since, by \lemref{t:wsizefudgefactors}, $\wsizexS P x {\sqrt{\sj}} \lesssim \wsizexS P x {\sqrt{\Uj}} = \wsizex P x$.  Our proof will follow the sketch below \thmref{t:unifiedsolution} in \secref{sec:spc}.  We start by deriving an exact expression for $\ri_O$:

\begin{lemma}
The solution to Eq.~\eqnref{e:gadget} has $a_O = 0$ if $b_O = 0$, and otherwise,
\begin{equation} \label{e:gadgetgeneralratio}
\ri_{O} 
= \lambda + \bra\o \left( \rj - \frac1{\lambda} \rj \ACJ^\adjoint (1 + \tfrac1{\lambda} \ACJ \rj \ACJ^\adjoint)^{-1} \ACJ \rj \right) \ket\o
 \enspace ,
\end{equation}
provided that $\rj$ and $(1 + \tfrac1{\lambda} \ACJ \rj \ACJ^\adjoint)^{-1}$ exist.
\end{lemma}

\begin{proof}
Recall from \defref{def:additionalmatrixnotations} that $\ket\o = {A_{OJ}}^\adjoint$ and $A_{CJ}$ is $\ACJ$ with range restricted.  Substituting Eqs.~\eqnref{e:gadgetbI} and~\eqnref{e:gadgetbC} into~\eqnref{e:gadgetaJ}, and rearranging terms gives 
\[
\left( \lambda - {A_{IJ}}^\adjoint \ri^{-1} A_{IJ}  - \frac 1 {\lambda} \ACJ^\adjoint \ACJ \right) \ket{a_J} = \ket\o b_O \enspace .
\]
From Eq.~\eqnref{e:gadgetbO}, if $b_O \neq 0$, then $a_O / b_O 
= \lambda - \braket \o {a_J} / b_O$,
so 
\begin{align}
\ri_{O} 
&=
\lambda + \bra\o ( \rj^{-1} + \frac1 {\lambda} \ACJ^\adjoint \ACJ)^{-1} \ket\o \label{e:egadgetgeneralratioproof} \\
&= \lambda + \bra\o \left( \rj - \frac1{\lambda} \rj \ACJ^\adjoint (1 + \tfrac1{\lambda} \ACJ \rj \ACJ^\adjoint)^{-1} \ACJ \rj \right) \ket\o \nonumber
 \enspace ,
\end{align}
by the Woodbury matrix identity~\cite{GolubVanloan96matrixcomputations}, 
provided that $\rj$ and $(1 + \tfrac1{\lambda} \ACJ \rj \ACJ^\adjoint)^{-1}$ exist.
\end{proof}

\begin{remark}[Form of Eq.~\eqnref{e:gadgetgeneralratio}]
Note from Eq.~\eqnref{e:gadgetgeneralratio} that $\ri_{O}$ is a real number provided that all the input ratios $\ri_{I}$ are themselves reals.  
Also, note that $\ri_{O}$ depends on $\ACJ$ only through $\ACJ^\adjoint \ACJ$ (see too Eq.~\eqnref{e:egadgetgeneralratioproof} in the proof); in particular, left-multiplying $\ACJ$ by $U$ where $U$ is any linear isometry (i.e., satisfying $U^\adjoint U = 1$) has no effect.  Since the grouped input vectors $v_J$ can be arbitrary in \defref{t:spanprogramdef}, $\ACJ^\adjoint \ACJ$ is in general an arbitrary $\abs{J} \times \abs{J}$ positive semidefinite matrix.  
\end{remark}

Now the main step in simplifying Eq.~\eqnref{e:gadgetgeneralratio} is dividing the matrix we want to invert into a $2 \times 2$ block matrix and applying the following well-known claim:

\begin{claim} \label{t:twobytwoblockinverse}
Let $X$ be an operator, and let $\Delta$ and $\Deltam = 1-\Delta$ be a projection and its complement.  Assume that $\DeltaXDelta$ is invertible.  Let the ``Schur complement" of $X$ be $\XDeltaSchur \equiv \DeltamXDeltam - \Deltam X (\DeltaXDelta)^\resinv X \Deltam$.  If the Schur complement of $X$ is invertible on $\Deltam$, then $X$ is invertible, and $X^{-1}$ is given by:
\begin{align}
\Delta X^{-1} \Delta 
&= (\DeltaXDelta)^\resinv + (\DeltaXDelta)^\resinv X \XDeltaSchur^\resinv X (\DeltaXDelta)^\resinv &
\Delta X^{-1} \Deltam 
&= - (\DeltaXDelta)^\resinv X \XDeltaSchur^\resinv \nonumber \\
\Deltam X^{-1} \Delta 
&= - \XDeltaSchur^\resinv X (\DeltaXDelta)^\resinv &
\Deltam X^{-1} \Deltam 
&= \XDeltaSchur^\resinv 
\end{align}
\end{claim}
\begin{proof}
Multiply out the matrices.
\end{proof}

\begin{lemma} \label{t:Rtilde}
The inverse $(1 + \frac1{\lambda} \ACJ \rj \ACJ^\adjoint)^{-1}$ exists, provided $0 < \lambda \leq \eps / \Ui_{i}$ for a small enough positive constant $\eps$ and for all $i \in I$. 
\end{lemma}

\begin{proof}
Let $\cp = \ACJ \sj^{-1/2} \Pi$, $\cm = \ACJ \sj^{1/2} \Pim$, $\cpp = \cp \cp^\adjoint$ and $\cmm = \cm \cm^\adjoint$.  We aim to show that $X^{-1}$ exists, where  
\begin{equation}\begin{split}
X
&= -\lambda^2 \big(1 + \frac1{\lambda} \ACJ \rj \ACJ^\adjoint \big) \\
&= \cpp - \lambda^2 (\cmm + 1)
 \enspace .
\end{split}\end{equation}

Let $\Delta = \cp \cp^+$ be the projection onto the range of $\cp$, and let $\Deltam = 1- \Delta$.  Then
\begin{align}
\DeltaXDelta &= \cpp - \lambda^2 \Delta (\cmm + 1) \Delta 
& \Delta X \Deltam &= -\lambda^2 \Delta \cmm \Deltam \nonumber \\
\Deltam X \Delta &= -\lambda^2 \Deltam \cmm \Delta 
& \DeltamXDeltam &= -\lambda^2 \Deltam (\cmm + 1) \Deltam 
\end{align}
Now since $\cpp$ is invertible on $\Delta$ (i.e., $\Delta = \cpp \cpp^+$), so is $X$.  By the Neumann series, 
\begin{equation}
(\DeltaXDelta)^\resinv = \cpp^+ (1 + O(\lambda^2 \norm{\sj}^2)) \enspace ,
\end{equation}
where we have used that $\norm{\cmm} = O(\norm{\sj})$ and $\norm{\cpp^+} = O(\norm{\sj})$ (\claimref{t:hyperplanes}), and where we write $O(\lambda^2 \norm{\sj}^2)$ to mean some matrix with norm so-bounded.  In particular, $(\DeltaXDelta)^\resinv$ is positive definite on $\Delta$.  

Let $\XDeltaSchur$ be the Schur complement of $\DeltaXDelta$ in $X$, 
\begin{align}
\XDeltaSchur
&\equiv \DeltamXDeltam - \Deltam X (\DeltaXDelta)^\resinv X \Deltam \nonumber \\
&= - \lambda^2 \Deltam (\cmm + 1) \Deltam - \lambda^4 \Deltam \cmm (\DeltaXDelta)^\resinv \cmm \Deltam \enspace .
\end{align}
As $-\XDeltaSchur$ on $\Deltam$ is the sum of the positive definite matrix $\lambda^2 \Deltam$ and positive semidefinite matrices, $\XDeltaSchur$ is negative definite on $\Deltam$ and in particular is invertible on $\Deltam$.  

Since $\DeltaXDelta$ and $\XDeltaSchur$ are each invertible, on $\Delta$ and on $\Deltam$, respectively, $X^{-1}$ exists by \claimref{t:twobytwoblockinverse}, as claimed.  
\end{proof}

The following discussion will use the notation from the proof of \lemref{t:Rtilde}.  It will also be convenient to let $S = \sqrt{\sj}$.  We have from Eq.~\eqnref{e:gadgetgeneralratio} 
\begin{align} \label{e:rOintermsofX}
\ri_{O} &= \lambda + \bra{\o} \rj \ket{\o} + \lambda \bra{\o} \rj \ACJ^\adjoint X^{-1} \ACJ \rj \ket{\o} \\
&=
\lambda - \frac1{\lambda} \bra\o S^{-2} \Pi \ket\o + \lambda \bra\o S^2 \Pim \ket\o + \left(\begin{array}{l l}
\dfrac1{\lambda} \bra\o S^{-1} \cp^\adjoint X^{-1} \cp S^{-1} \ket\o &+ \lambda^3 \bra\o S \cm^\adjoint X^{-1} \cm S \ket\o \\
- \lambda \big( \bra\o S^{-1} \cp^\adjoint X^{-1} \cm S \ket\o &+ \bra\o S \cm^\adjoint X^{-1} \cp S^{-1} \ket\o \big) 
\end{array}\right) \enspace . \nonumber
\end{align}
Our goal now is to expand the above expression as a series in $\lambda$, evaluating the coefficients of $1/\lambda$ and of $\lambda$, and bounding higher-order terms.  In order to expand $X^{-1}$ as a series, we use the block decomposition of $X$ and \claimref{t:twobytwoblockinverse}.  

Let us start by evaluating two expressions, $(\DeltaXDelta)^\resinv$ and $V = -\lambda^2 \cm^\adjoint \XDeltaSchur^\resinv \cm$, that will reappear frequently in the following analysis.

\begin{claim} \label{t:Xpclaim}
$(\DeltaXDelta)^\resinv$ satisfies 
\begin{align}
(\DeltaXDelta)^\resinv
&= \cpp^+ + \lambda^2 \cpp^+ (\cmm + 1) \cpp^+ + \cpp^+ O(\lambda^4 \norm{\sj}^3) \cpp^+ \\
&= \cpp^+ ( 1 + O(\lambda^2 \norm{\sj}^2) ) = O(\norm{\sj}) \enspace . \nonumber
\end{align}
\end{claim}
\begin{proof}
We know that $\norm{(\DeltaXDelta)^\resinv} = O(\norm{\sj})$.  Note that for matrices $A$ and $B$, $(A + B)^{-1} = (1 - (A + B)^{-1} B) A^{-1} = A^{-1} (1 - B (A + B)^{-1}) = A^{-1} - A^{-1} B A^{-1} + A^{-1} B (A + B)^{-1} B A^{-1}$ provided $A$ and $A+B$ are invertible.  Applying this with $A = \cpp$ and $B = -\lambda^2 \Delta ( \cmm + 1) \Delta$ gives 
\begin{align*}
(\DeltaXDelta)^\resinv 
&=
\cpp^+ + \lambda^2 \cpp^+ (\cmm + 1) \cpp^+ + \lambda^4 \cpp^+ (\cmm + 1) (\DeltaXDelta)^\resinv (\cmm + 1) \cpp^+ \nonumber\\
&=
\cpp^+ + \lambda^2 \cpp^+ (\cmm + 1) \cpp^+ + \cpp^+ O(\lambda^4 \norm{\sj}^3) \cpp^+
 \enspace . \qedhere
\end{align*}
\end{proof}

\begin{claim} \label{t:Xmclaim}
Let $V = -\lambda^2 \cm^\adjoint \Deltam X^{-1} \Deltam \cm = -\lambda^2 \cm^\adjoint \XDeltaSchur^\resinv \cm$.  Then 
\begin{equation} \label{e:schurcomplement}
V = \VV + O(\lambda^2 \norm{\sj}^2) \qquad \text{where} \qquad
\VV = \cm^\adjoint \Deltam (\dcmm + 1)^{-1} \Deltam \cm \enspace .
\end{equation}
In particular, $\norm{\VV} < 1$ and $\norm{V} < 1 + O(\lambda^2 \norm{\sj}^2) = O(1)$.  
\end{claim}
\begin{proof}
We compute
\begin{align}
V &= 
-\lambda^2 \cm^\adjoint (\DeltamXDeltam - \Deltam X (\DeltaXDelta)^\resinv X \Deltam)^\resinv \cm
\nonumber\\
&= \cm^\adjoint \left[ \Deltam (\cmm + 1) \Deltam + \lambda^2 \Deltam \cmm (\DeltaXDelta)^\resinv \cmm \Deltam\right]^\resinv \cm \nonumber\\
&= M^\adjoint \left[ M \big( 1 + \lambda^2 \cm^\adjoint (\DeltaXDelta)^\resinv \cm \big) M^\adjoint + 1 \right]^\resinv M \enspace ,
\intertext{%
where $M = \Deltam \cm$.  Now $\lambda^2 \cm^\adjoint (\DeltaXDelta)^\resinv \cm = O(\lambda^2 \norm{\sj}^2)$ since $\cm = O(\norm{\sqrt{\sj}})$.  Therefore, $M \big( 1 + \lambda^2 \cm^\adjoint (\DeltaXDelta)^\resinv \cm \big) M^\adjoint + 1$ is the sum of positive-definite and positive-semidefinite matrices, hence is invertible.  
Again use $(A+B)^{-1} = (1 - (A+B)^{-1} B) A^{-1}$, now with $A = M M^\adjoint + 1 = \dcmm+1$ and $B = \lambda^2 M \cm^\adjoint (\DeltaXDelta)^\resinv \cm M^\adjoint$, to get}
V 
&= \VV - V (\lambda^2 \cm^\adjoint (\DeltaXDelta)^\resinv \cm) \VV \nonumber \\
&= \VV \cdot \left[ 1 + (\lambda^2 \cm^\adjoint (\DeltaXDelta)^\resinv \cm) \VV \right]^{-1} \enspace ,
\end{align}
provided the inverse right-multiplying $\VV$ exists.  Indeed, $M^\adjoint (M M^\adjoint + 1)^{-1} M < 1$ for an arbitrary matrix $M$ and in particular for $M = \Deltam \cm$ (if the singular-value decomposition of $M$ is $M = \sum_i m_i \ketbra{i}{i'}$, then $M^\adjoint (M M^\adjoint + 1)^{-1} M = \sum_i \frac{m_i^2}{m_i^2 + 1} \ketbra{i'}{i'}$).  Therefore the inverse does exist, and we obtain Eq.~\eqnref{e:schurcomplement}.
\end{proof}

Now, using \claimref{t:Xpclaim} and \claimref{t:Xmclaim}, we find 
\begin{align} \label{e:Xinversep}
\Delta X^{-1} \Delta 
&=
(\DeltaXDelta)^\resinv + (\DeltaXDelta)^\resinv X \XDeltaSchur^\resinv X (\DeltaXDelta)^\resinv \nonumber\\
&=
(\DeltaXDelta)^\resinv - \lambda^2 (\DeltaXDelta)^\resinv \cm V \cm^\adjoint (\DeltaXDelta)^\resinv \nonumber\\
&= 
\cpp^+ + \lambda^2 \cpp^+ (\cmm + 1) \cpp^+ + \cpp^+ O(\lambda^4 \norm{\sj}^3) \cpp^+ \nonumber\\
&\quad - \lambda^2 \cpp^+ (1 + O(\lambda^2 \norm{\sj}^2)) \cm \big( \VV + O(\lambda^2 \norm{\sj}^2) \big) \cm^\adjoint (1 + O(\lambda^2 \norm{\sj}^2)) \cpp^+ \nonumber\\
&=
\cpp^+ + \lambda^2 \cpp^+ \left(- \cm \VV \cm^\adjoint + \cmm + 1 + O(\lambda^2 \norm{\sj}^3) \right) \cpp^+
 \enspace ,
\intertext{and}
\label{e:Xinversepm}
\Delta X^{-1} \Deltam \cm
&=
\lambda^2 (\DeltaXDelta)^\resinv \cmm \XDeltaSchur^\resinv \cm \nonumber\\
&= 
- (\DeltaXDelta)^\resinv \cm V \nonumber\\
&= 
-\cpp^+ (1 + O(\lambda^2 \norm{\sj}^2)) \cm (\VV + O(\lambda^2 \norm{\sj}^2)) \nonumber\\
&=
-\cpp^+ \cm \VV + \cpp^+ \cdot O(\lambda^2 \norm{\sj}^{5/2})
 \enspace .
\end{align}
In particular, $\norm{\Delta X^{-1} \Delta} = O(\norm{\sj})$ and $\norm{\Delta X^{-1} \Deltam \cm} = O(\norm{\sj}^{3/2})$.  

Let us now substitute the expressions we have derived into Eq.~\eqnref{e:rOintermsofX} for $\ri_{O}$.  Consider each of the terms involving $X^{-1}$ separately.  First of all, 
\begin{align}
\frac1{\lambda} \bra\o S^{-1} \cp^\adjoint X^{-1} \cp S^{-1} \ket\o
&=
\frac1{\lambda} \bra\o S^{-1} \cp^\adjoint (\Delta X^{-1} \Delta) \cp S^{-1} \ket\o \nonumber \\
&=
\frac1{\lambda} \bra\o \Pi S^{-1} \cp^+ \cp S^{-1} \Pi \ket\o \label{e:rOXterm1}
\\ &\quad + \lambda \bra\o \Pi S^{-1} \cp^+ \big( 1 + \cmm - \cm \VV \cm^\adjoint \big) (\cp^\adjoint)^+ S^{-1} \Pi \ket\o \nonumber \\ &\quad 
+ O(\lambda^3 \norm{\sj}^3) 
\enspace , \nonumber
\intertext{where we have substituted Eq.~\eqnref{e:Xinversep} and applied $\norm{S^{-1} \cp^\adjoint \cpp^+} = \norm{S^{-1} \cp^+} = O(1)$ (\claimref{t:hyperplanes}).
\hfil \break \indent
Also, by \claimref{t:Xmclaim} and Eqs.~\eqnref{e:Xinversep}, \eqnref{e:Xinversepm}, 
}
\lambda^3 \bra\o S \cm^\adjoint X^{-1} \cm S \ket\o
&= 
\lambda^3 \bra\o S \cm^\adjoint (\Delta + \Deltam) X^{-1} (\Delta + \Deltam) \cm S \ket\o \nonumber \\
&= 
-\lambda \bra\o S V S \ket\o + O(\lambda^3 \norm{\sj}^3) \label{e:rOXterm2} \\
&=
-\lambda \bra\o \Pim S \VV S \Pim \ket\o + O(\lambda^3 \norm{\sj}^3) \enspace . \nonumber
\intertext{Lastly, by Eqs.~\eqnref{e:Xinversep}, \eqnref{e:Xinversepm},}
\lambda \bra\o S^{-1} \cp^\adjoint X^{-1} \cm S \ket\o
&=
\lambda \bra\o S^{-1} \cp^\adjoint \Delta X^{-1} (\Delta + \Deltam) \cm S \ket\o \nonumber\\
&=
\lambda \bra\o S^{-1} \cp^\adjoint \cpp^+ (1 + O(\lambda^2 \norm{\sj}^2)) \cm S \ket\o
\nonumber\\&\quad- \lambda \bra\o S^{-1} \cp^\adjoint \cpp^+ \left[ \cm \VV + O(\lambda^2 \norm{\sj}^{5/2}) \right] S \ket\o \nonumber\\
&= 
\lambda \bra\o \Pi S^{-1} \cp^+ ( \cm - \cm \VV ) S \Pim \ket\o + O(\lambda^3 \norm{\sj}^3) \enspace . \label{e:rOXterm3}
\end{align}

Substituting Eqs.~\eqnref{e:rOXterm1}, \eqnref{e:rOXterm2}, \eqnref{e:rOXterm3} into the expression for $\ri_{O}$ gives
\begin{align}
\ri_{O} 
&=
-\frac1{\lambda} \norm{ (1 - \cp^+ \cp)S^{-1} \Pi \ket\o}^2 + \lambda \bra{v} \big( 1 - \VV \big) \ket{v} + O(\lambda + \lambda^3 \norm{\sj}^3)
 \enspace ,
\intertext{where $\ket{v} = (S \Pim - \cm^\adjoint (\cp^\adjoint)^+ S^{-1})\ket\o$.  
From the singular-value decomposition, one infers that $(1 - M^\adjoint (M M^\adjoint + 1)^{-1} M) - (1 - M^+ M) = M^+ (1 - (M M^\adjoint + 1)^{-1}) (M^\adjoint)^+$ for any matrix $M$, and in particular for $M = \Deltam \cm$.  Moreover, $\bra{v} \dcm^+ (1 - (\dcmm + 1)^{-1}) (\dcm^\adjoint)^+ \ket{v} = O(1)$, since $\norm{S \dcm^+} = O(1)$ and $\norm{\bra{v} S^{-1}} = O(1)$.  Therefore the above equation simplifies to 
}
\ri_{O} \label{e:finalrOingeneral}
&=
-\frac1{\lambda} \norm{ (1 - \cp^+ \cp)S^{-1} \Pi \ket\o}^2 + \lambda \big(\norm{(1 - \dcm^+ \dcm) \ket{v}}^2 + O(1) \big) + O(\lambda^3 \norm{\sj}^3)
 \enspace .
\end{align}

This is as far as we can simplify $\ri_{O}$ in general.  When $f_P(x) = 1$, the first term is $-1 / (\lambda \, \wsizexS P x S)$, as desired, using the last expression of Eq.~\eqnref{e:spcinterpretationtrue} for $\wsizexS P x S$. 
Assume then that $f_P(x) = 0$, i.e., $\Pi \ket\o \in \Range(\Pi \ACJ^\adjoint)$.  In this case, the first term in Eq.~\eqnref{e:finalrOingeneral} is zero, and the second term can be simplified slightly further.  Using $(\cp^\adjoint)^+ S^{-1} \ket\o = (\Pi S^{-1} \ACJ^\adjoint)^+ S^{-1} \Pi \ket\o = (\Pi \ACJ^\adjoint)^+ \ket\o$ and $\Pi \ACJ^\adjoint (\Pi \ACJ^\adjoint)^+ \ket\o = \Pi \ket\o$,
\begin{equation}\begin{split}
\ket{v} 
&= (S \Pim - \cm^\adjoint (\Pi \ACJ^\adjoint)^+) \ket\o \\
&= (S \Pim - \cm^\adjoint (\Pi \ACJ^\adjoint)^+) \ket\o + S \Pi (1 - \ACJ^\adjoint (\Pi \ACJ^\adjoint)^+) \ket\o \\
&= S (1 - \ACJ^\adjoint (\Pi \ACJ^\adjoint)^+) \ket\o
 \enspace .
\end{split}\end{equation}
Moreover, since $\Deltam \ACJ S \Pi = 0$, $\Deltam \cm = \Deltam \ACJ S \Pim = \Deltam \ACJ S$.  Therefore, $\norm{(1 - \dcm^+ \dcm) \ket{v}}^2 = \wsizexS P x S$, as desired, using the last expression of Eq.~\eqnref{e:spcinterpretationfalse} for $\wsizexS P x S$.  
This concludes the proof of \thmref{t:strongsmallenergy}.
\end{proof}

\thmref{t:strongsmallenergy} completes the $\lambda \neq 0$ portion of \thmref{t:unifiedsolution}, finishing its proof.
\hspace{\stretch{1}}  \qed

\section{Quantum algorithm} \label{s:algorithm}

The approach outlined in \secref{s:qalg} is slightly indirect.  To motivate it, we begin by briefly considering in \secref{s:continuoustime} a more direct algorithm $\algorithm'$, that runs phase estimation directly on $\exp(i A_{\tilde G_P(x)})$.  $\algorithm'$ is analogous to the algorithm described by Cleve et al.~\cite{ccjy:and-or} soon after the original NAND formula evaluation paper~\cite{fgg:and-or}.
Algorithm $\algorithm'$ is nearly optimal, but not quite.  
The operator $\exp(i A_{\tilde G_P(x)})$ is a continuous-time quantum walk, and 
the overhead can be thought of as coming from simulating continuous-time quantum dynamics with a discrete computational model, in particular with discrete oracle queries.
To avoid this overhead, the proof of \thmref{t:result} in \secref{s:discretetime} works with a discrete-time quantum walk.

The approach in \secref{s:continuoustime} is optional motivation, and the reader may choose to skip directly to \secref{s:discretetime}.

\subsection{Intuition: Continuous-time quantum walk algorithm} \label{s:continuoustime}

\thmref{t:formulaspectrum} immediately suggests the basic form of a quantum algorithm for evaluating $\varphi(x)$:
\begin{center}
\fbox{
\begin{minipage}[l]{5.5in}
\noindent {\bf Algorithm $\algorithm'$:} 
Input $x \in \{0,1\}^N$, Output true/false.
\begin{enumerate}
\item Prepare an initial state on the output node, $\ket{a_O}$.  
\item
Run phase estimation, with precision $\delta_p \leq \frac1{\norm{A_{\tilde G_P(0^N)}}} \frac\eps {\ADV(\varphi)}$ and small enough constant error rate $\delta_e$, on the unitary $V = \exp ( i A_{\tilde G_P(x)} / \norm{A_{\tilde G_P(0^N)}} )$.  
\item Output true if and only if the phase estimation output is $\lambda = 0$.
\end{enumerate} 
\end{minipage} }
\end{center}
The idea of the second step is to ``measure the Hamiltonian $A_{\tilde G_P(x)}$."
In this step, we have normalized $A_{\tilde G_P(x)}$ by $\norm{A_{\tilde G_P(0^N)}} \geq \norm{A_{\tilde G_P(x)}}$ instead of by $\norm{A_{\tilde G_P(x)}}$, in order to minimize dependence on the input $x$.  
This norm is $O(1)$ since the graph $\tilde G_P(0^N)$ has vertex degrees and edge weights all $O(1)$.

Algorithm $\algorithm'$ evaluates $\varphi(x)$ correctly, with a constant gap between its completeness and soundness:
\begin{itemize}
\item
\thmref{t:formulaspectrum} implies that if the formula evaluates to true, then $A_{\tilde G_P(x)}$ has an eigenvalue-zero eigenstate with squared support $\abs{a_O}^2 = \Omega(1)$ on $a_O$.  Therefore, the phase estimation outcome is $\lambda = 0$ with probability at least $\abs{a_O}^2 - \delta_e = \Omega(1)$ (the completeness parameter).  
\item
On the other hand, if the formula evaluates to false, then \thmref{t:formulaspectrum} implies that $A_{\tilde G_P(x)} / \norm{A_{\tilde G_P(0^N)}}$ has no eigenvalue-$\lambda$ eigenstates supported on $a_O$ with $\abs{\lambda} \leq \delta_p$.  Therefore, the measured outcome will be $\lambda = 0$ only if there is an error in the phase estimation.  By choosing $\delta_e$ a small enough constant, the soundness error $\delta_e$ will be bounded away from the completeness parameter.
\end{itemize}

The efficiency of $\algorithm'$ also seems promising.  Phase estimation of $V$ with precision $\delta_p$ and error rate $\delta_e$ requires $O(1/(\delta_p \delta_e))$ calls to $V$~\cite{cemm:qalg}.  Therefore, the second step requires only $O(\ADV(\varphi))$ calls to $V$.  
However, we still need to explain how to implement $V$.  This is important because $A_{\tilde G_P(x)}$ depends on the input $x$.  Therefore, implementing $V$ requires querying the $x$.  If each call to $V$ requires many queries to the input oracle $O_x$ of Eq.~\eqnref{e:ox}, then the overall query efficiency of $\algorithm'$ will be poor.  

Note now that only the input edges of $G_P(x)$ depend on the input $x$.  Therefore, $A_{\tilde G_P(x)}$ can be split up into two terms: $(\text{input edges}) + (\text{all other edges})$.  The first term can be exponentiated with only two queries to the input oracle $O_x$, while exponentiating the second term requires no input queries.  The two terms do not commute, but the exponential of their sum can still be computed to sufficient precision by using a Lie product decomposition.  These are more quantitative versions of identities like $e^{A+B} = \lim_{n \rightarrow \infty} (e^{A/n} e^{B/n})^n$.  For more details, see~\cite{ccjy:and-or}.  

Unfortunately, implementing the exponential of $A_{\tilde G_P(x)}$ will require $\omega(1)$ input queries.  By using higher-order Lie product formulas, the overhead can be reduced to $\exp(O(\sqrt{\log \abs{x}}))$, which is $N^{o(1)}$.  However, this is still a super-constant overhead, so it appears that this approach cannot yield an optimal formula evaluation algorithm---the best we can hope for is $O(\ADV(\varphi)) \cdot N^{o(1)}$ queries.

\subsection{Proof of \texorpdfstring{\thmref{t:result}}{Theorem \ref*{t:result}}: Discrete-time quantum walk algorithm} \label{s:discretetime}

Therefore, we turn to the approach used in the NAND formula evaluation algorithm of~\cite{ChildsReichardtSpalekZhang07andor}.
Instead of running phase estimation on the exponential of $A_{\tilde G_P(x)}$, we construct a discrete-time, or ``coined," quantum walk $U_x = \tilde{O}_x U_{0^N}$, where $\tilde{O}_x$ is the adjusted oracle of Eq. \eqnref{e:phasefliporacledef}, that has spectrum and eigenvectors corresponding in a precise way to those of $A_{\tilde G_P(x)}$.  Then we run phase estimation on $U_x$.  Each call to $U_x$ requires exactly one oracle query, so there is no query overhead.

\subsubsection{Construction of the coined quantum walk \texorpdfstring{$U_x$}{U\_x}}

The first step in constructing $U_x$ is to decompose $\tilde A_{G_P(0^N)}$ into $(\mathrm{constant}) \cdot \Delta^\adjoint \circ \Delta$, where $\Delta$ is a square matrix with row norms one, and $\circ$ denotes the entrywise matrix product.
We follow~\cite{ChildsReichardtSpalekZhang07andor}.  
One minor technical difference, though, is that for us, $A_{\tilde G_P(x)}$ is a Hermitian matrix with possibly complex entries.  In~\cite{ChildsReichardtSpalekZhang07andor}, the analogous weighted adjacency matrix, for the NAND formula $\varphi$, is a real symmetric matrix.  Therefore, we need to slightly modify the construction of $\Delta$ to obtain the correct phases for the entries of $\tilde A_{G_P(0^N)}$.

\begin{definition}
For notational convenience, let $\H = A_{\tilde G_P(0^N)}$ be the weighted adjacency matrix for $\tilde G_P(0^N)$.  (Recall from \thmref{t:formulaspectrum} that $\tilde G_P(0^N)$ is the same as $G_P$ except with the edge weight on the output edge $(a_O, b_O)$ reduced.)  $\H = \sum_{v,w} \H_{v,w} \ketbra{v}{w}$ is a Hermitian matrix.  

$\tilde G_P(0^N)$ is a bipartite graph, so we may color each vertex red or black, such that every edge is between one red vertex and one black vertex.
\end{definition}

\begin{claim} \label{t:Hamiltoniancoindecomposition}
Let $\H' = \sum_{v,w} \abs{\H_{v,w}} \ketbra{v}{w}$ be the entrywise absolute value of $\H$.  $\H'$ is a real symmetric matrix.  Let $\norm{\H'}$ be the largest-magnitude eigenvalue of $\H'$.  
Let $\ket{\delta}$ be the principle eigenvector of $\H'$, with $\braket{v}{\delta} = \delta_v > 0$ for every $v$, and let 
\begin{equation} \label{e:Hamiltoniancoindecomposition}
\Delta = \frac{1}{\sqrt{\norm{\H'}}} \sum_{\substack{\text{black $v$}\\\text{red $w$}}} \left( \big(\sqrt{\H_{v,w}}\big)^* \sqrt{\frac {\delta_w}{\delta_v}} \ketbra{v}{w} + \sqrt{\H_{v,w}} \sqrt{\frac {\delta_v}{\delta_w}} \ketbra{w}{v} \right)
 \enspace .
\end{equation}
Then $\Delta$ has all row norms one, and $\H = \norm{\H'} \cdot \Delta^\adjoint \circ \Delta$. 
\end{claim}

\begin{proof}
Since $\H'$ has nonnegative entries, the principal eigenvector $\ket \delta$ is also nonnegative.  Since $\tilde G_P(0^N)$ is a connected graph, $\delta_v > 0$ for every $v$.  Hence $\Delta$ is well defined up to choice of sign of the square root, which doesn't matter.  

By construction, for all $v$ and $w$, $\Delta_{v,w}^* \Delta_{w,v} = \H_{v,w} / \norm{\H'}$, i.e., $\norm{\H'} \cdot \Delta^\adjoint \circ \Delta = \H$.  Furthermore, the squared norm of the $v$-th row of $\Delta$ is $\sum_w \abs{\Delta_{v,w}}^2 = \frac1{\norm{\H'}} \frac 1 {\delta_v} \sum_w \H'_{v,w} \delta_w = \frac1{\norm{\H'}} \frac {(\H' \delta)_v} {\delta_v} = 1$.  
\end{proof}

\begin{remark}
In defining $\Delta$, we have evenly divided the complex phases of entries of $\H$ between red-black and black-red terms.  However, any division of the phases would have worked.  For example, \claimref{t:Hamiltoniancoindecomposition} would also hold with Eq.~\eqnref{e:Hamiltoniancoindecomposition} replaced by 
$\Delta = \frac{1}{\sqrt{\norm{\H'}}} \sum_{\substack{\text{black $v$}\\\text{red $w$}}} \Big( \frac{\H_{v,w}^*}{\sqrt{\abs{\H_{v,w}}}} \sqrt{\frac {\delta_w}{\delta_v}} \ketbra{v}{w} + \sqrt{\abs{\H_{v,w}} \frac {\delta_v}{\delta_w}} \ketbra{w}{v} \Big)$.  
\end{remark}

We can now apply Szegedy's correspondence theorem~\cite{Szegedy04walkfocs} to relate the spectrum of $\tilde A_{G_P}$ to that of a discrete-time coined quantum walk unitary.

\def \tensor {\otimes}
\def \cH {{\mathcal{H}}}

\begin{theorem}[\cite{Szegedy04walkfocs}]
	\label{t:szegedization}
	Let $\{ \ket v : v \in V \}$ be an orthonormal basis for $\cH_V$.  For each $v \in V$, let $\ket{\tilde{v}} = \ket{v} \tensor \sum_{w \in V} \delta_{vw} \ket{w} \in \cH_V \tensor \cH_V$, where $\braket{\tilde v}{\tilde v} = \sum_w \abs{\delta_{vw}}^2 = 1$.  Let $T = \sum_v \ketbra{\tilde v}{v}$ and $\Pi = T T^\adjoint = \sum_v \ketbra{\tilde v}{\tilde v}$ be the projection onto the span of the $\ket{\tilde v}$s.  Let $S = \sum_{v,w} \ketbra{v,w}{w,v}$, a swap. 
	Let $U = (2\Pi -1) S$, a swap followed by reflection about the span of the $\ket{\tilde v}s$.
	Let $M = T^\adjoint S T  
= \sum_{v,w} \delta_{vw}^* \delta_{wv} \ketbra{v}{w}$.   
	
	Then the spectral decomposition of $U$ corresponds to that of $M$ as follows: 
	Take $\{\ket{\lambda_\alpha}\}$ a complete set of orthonormal eigenvectors of the Hermitian matrix $M$ with respective eigenvalues $\lambda_\alpha$.  
	Let $R_\alpha = \Span\{T \ket{\lambda_\alpha}, ST\ket{\lambda_\alpha}\}$.  Then $R_a \perp R_{\alpha'}$ for $\alpha \neq \alpha'$; let $R = \oplus_\alpha R_\alpha$.  
	$U$ fixes the spaces $R_a$ and is $-S$ on $R^\perp$.  The eigenvalues and eigenvectors of $U$ within $R_a$ are given by $\beta_{\alpha,\pm} = \lambda_\alpha \pm i \sqrt{1 - \lambda_\alpha^2}$ and $(1 + \beta_{\alpha,\pm}S)T \ket{\lambda_\alpha}$, respectively.
\end{theorem}

\noindent A proof of \thmref{t:szegedization} in the above form is given in~\cite{ChildsReichardtSpalekZhang07andor}, and see~\cite{MagniezNayakRolandSantha07search}.

\begin{remark}[Coined quantum walks]
The operator $U = (2 \Pi-1)S$ in \thmref{t:szegedization} is known as a ``coined quantum walk."  $S$ is known as the ``step operator," and the reflection $(2\Pi-1)$ is the ``coin-flip operator."  On the space $R$, $(2\Pi-1)$ decomposes as $\sum_{v} \ketbra{v}{v} \tensor (\text{reflection about $\sum_{w} \delta_{vw} \ket{w}$})$.

In a classical random walk on a graph, a coin is flipped between each step to decide which adjacent vertex to step to next.  In a coined quantum walk, on the other hand, the coin is maintained as part of the coherent quantum state, and is reflected between steps (also known as ``diffusion").  
\end{remark}

\begin{remark}
\thmref{t:szegedization} can be viewed as giving a correspondence between coined quantum walks and classical random walks; in the special case that each $\delta_{vw} = \delta_{wv} \geq 0$, $M$ is a classical random walk transition matrix.  For general $\delta_{vw}$, \thmref{t:szegedization} can be viewed a correspondence between coined quantum walks and continuous-time quantum walks.  We use the theorem in the latter sense.
\end{remark}

\begin{lemma} \label{t:szegedizationapplied}
For $\Delta$ defined by Eq.~\eqnref{e:Hamiltoniancoindecomposition} and with $\delta_{vw} = \Delta_{v,w}$, let $U_{0^N}$ be the coined quantum walk operator $U_{0^N} = i U = i (2\Pi - 1)S$ in the notation of \thmref{t:szegedization}.  $U_{0^N}$ acts on $\cH_V \otimes \cH_V$, where $V$ is the vertex set of $G_P$.  
For $x \in \{0,1\}^N$, let $U_x = \tilde{O}_x U_{0^N}$, where $\tilde{O}_x$ applies a phase $(-1)^{x_i}$ to input vertex $b_i$ and otherwise does nothing (see Eq. \eqnref{e:phasefliporacledef}).  
Then,
\begin{itemize}
\item
If $\varphi(x) = 1$, there exist eigenvalue 1 and eigenvalue $-1$ normalized eigenstates of $U_x$ each with $\Omega(1)$ support on $\ket{a_O, b_O}$.  
\item
If $\varphi(x) = 0$, then $U_x$ does not have any eigenstates supported on $\ket{a_O, b_O}$ with eigenvalues $\pm e^{i \lambda}$ for $\abs{\lambda} \leq \arcsin\!\big(\frac1{\norm{\H'} } \, \frac \epsilon {\ADV(\varphi)} \big) = \Omega(1/\ADV(\varphi))$, where $\epsilon > 0$ is the constant of \thmref{t:formulaspectrum}.
\end{itemize}
\end{lemma}

\begin{proof}
Note that for an input vertex $b_i$ on a span program input edge $(a_i, b_i)$, the $b_i$th row of $\Delta$ is $\bra{b_i} \Delta = \bra{a_i}$. 
Define $\Delta(x)$ as follows:  If $x_i = 1$, then let $\bra{b_i} \Delta(x) = \bra{b_i}$ (i.e., in the classical walk formulation, make $b_i$ a probability sink), and let the other rows of $\Delta(x)$ be the same as those of $\Delta = \Delta(0^N)$.  

In the notation of \thmref{t:szegedization} with each $\delta_{vw}$ set to the $(v,w)$ entry of $\Delta(x)$, the vectors $\ket{\tilde v}$ do not depend on $x$ if $v \notin \{b_i\}_{i \in I}$, whereas 
\[
\ket{\tilde b_i} = \left\{ \begin{array}{cc} \ket{b_i, a_i} & \text{if $x_i = 0$} \\ \ket{b_i, b_i} & \text{if $x_i = 1$} \end{array} \right.
\]
Therefore, in $M = \Delta(x)^\adjoint \circ \Delta(x)$, entries $(a_i, b_i)$ and $(b_i, a_i)$ are zeroed out when $x_i = 1$, while other entries are unchanged: so $M = \frac{1}{\norm{\H'}} A_{\tilde G_P(x)}$.  Also, on $R$, $i U = i (2\Pi -1)S$ is the same as $U_x$.  So \thmref{t:szegedization} implies that the spectrum of $U_x = \tilde{O}_x U_{0^N}$ corresponds exactly to that of $A_{\tilde G_P(x)} / \norm{\H'}$.  If the eigenvalues of $A_{\tilde G_P(x)} / \norm{\H'}$ are $\{ \lambda_\alpha \}$, then the eigenvalues of $U_x$ are given by $i \beta_{\alpha,\pm} = \pm \sqrt{1 - \lambda_\alpha^2} + i \lambda_\alpha$, i.e., $e^{i \arcsin \lambda_\alpha}$ and $-e^{- i \arcsin \lambda_\alpha}$.  

In case $\varphi(x) = 1$, \thmref{t:formulaspectrum} promises that $A_{\tilde G_P(x)}$ has an eigenvalue-zero eigenstate with $\Omega(1)$ support on $a_O$.  Denote this eigenstate by $\ket{\lambda_\alpha = 0}$.  By \thmref{t:szegedization}, $(1 \pm i S) T \ket{\lambda_\alpha = 0}$ are eigenstates of $U_x$ with eigenvalues $\pm 1$.  
Since $T \ket{a_O} = \ket{a_O, b_O}$, the eigenvectors $(1 \pm i S) T \ket{\lambda_\alpha = 0}$ each have $\Omega(1)$ support on $\ket{a_O, b_O}$.   
Moreover, this remains true even after renormalizing: $T$ is an isometry, while the swap $S$ is unitary, so $\norm{ (1 \pm i S) T \ket{\lambda_\alpha = 0} } \leq 2$. 

The claim also follows for the case $\varphi(x) = 0$ by Theorems~\ref{t:formulaspectrum} and~\ref{t:szegedization}.  
Every eigenstate of $i U$ with support on $\ket{a_O, b_O} = T \ket{a_O}$ must be of the form $(1 + \beta_{\alpha,\pm}S)T \ket{\lambda_\alpha} = (1 + (\lambda_\alpha \pm i \sqrt{1-\lambda_\alpha^2}) S)T \ket{\lambda_\alpha}$.  The terms which can overlap $T \ket{a_O}$ are either $\braket{a_O}{\lambda_\alpha}$ (via $T$) or $\braket{b_O}{\lambda_\alpha}$ (via $S T$).  But by \thmref{t:formulaspectrum}, both coefficients must be zero.  
Note that $\norm{\H'} = O(1)$ since the graph $\tilde G_P(0^N)$ has vertex degrees and edge weights all $O(1)$.  
Therefore, the spectral gap from zero of $A_{\tilde G_P(x)} / \norm{\H'}$ is only a constant factor worse than that of $A_{\tilde G_P(x)}$.
\end{proof}

\subsubsection{Algorithm \texorpdfstring{$\algorithm$}{A}, correctness, and query and time complexity}

\begin{center}
\fbox{
\begin{minipage}[l]{5.5in}
\noindent {\bf Algorithm $\algorithm$:} 
Input $x \in \{0,1\}^N$, Output true/false.
\begin{enumerate}
\item
Prepare an initial state on the output edge $\ket{a_O, b_O}$.
\item
Run phase estimation on $U_x = \tilde{O}_x U_{0^N}$, with precision $\delta_p \leq \arcsin\!\big(\frac1{\norm{\H'} } \, \frac \epsilon {\ADV(\varphi)} \big)$ and small enough constant error rate $\delta_e$.  
\item
Output true if the measured phase is $0$ or $\pi$.  Otherwise output false.
\end{enumerate} 
\end{minipage} }
\end{center}

\noindent {\bf{Correctness:}}
\lemref{t:szegedizationapplied} implies that $\algorithm$ is both complete and sound:
\begin{itemize}
\item
If $\varphi(x) = 1$, then $U(x)$ has eigenvalue-($\pm 1$) eigenstates each with $\Omega(1)$ squared support on $\ket{a_O,b_O}$.  The completeness parameter is at least this squared support minus the phase estimation error rate $\delta_e$.  For small enough constant $\delta_e$, the completeness is $\Omega(1)$.   
\item
If $\varphi(x) = 0$, then since the precision parameter $\delta_p$ is smaller than the promised gap away from $\pm 1$ in \lemref{t:szegedizationapplied}, phase estimation will output $0$ or $\pi$ only if there is an error.  By choosing the error rate $\delta_e$ a small enough constant, the soundness error $\delta_e$ will be bounded away from the completeness parameter.  
\end{itemize}
Therefore, algorithm $\algorithm$ is correct.  The constant gap between its completeness and soundness parameters can be amplified as usual.  

\smallskip\noindent {\bf{Query and time complexity:}} 
Phase estimation of $U_x$ with precision $\delta_p$ and error rate $\delta_e$ requires $O(1/(\delta_p \delta_e))$ calls to $U_x = \tilde{O}_x U_{0^N}$~\cite{cemm:qalg}.  Therefore, $\algorithm$ makes $O(\ADV(\varphi))$ queries to the input oracle $O_x$.

The time-efficiency claim of \thmref{t:result} is slightly more complicated.  Here, we need to allow a preprocessing phase in which the algorithm can compute $A_{\tilde G_P(0^N)}$ and in particular (approximations to) the coin diffusion operators in $U_{0^N}$.  This preprocessing depending on $\varphi$, but not $x$, takes $\poly(N)$ time.
The algorithm then needs coherent access to the precomputed information in order to apply efficiently the coin diffusion operators.  For further details, see~\cite{ChildsReichardtSpalekZhang07andor}.  

\smallskip
This completes the proof of \thmref{t:result}.  \hspace{\stretch{1}}$\square$

\subimport*{functiontable/}{functiontable}

\end{document}

%% file: functiontable/functiontable.tex
\section{Table of functions on up to four bits}

\renewcommand{\H}[1]{\hyperref{}{#1}{}{\ensuremath{\# #1}}}

Which four-bit gates can be added to the gate set $\gateset$ without affecting the correctness of \thmref{t:wsizegates} or \thmref{t:result}?  As summarized in \secref{s:fourbitgates}, we have made partial progress toward answering this question.  In this appendix, we present a table that lists all the four-bit gates, up to equivalences, and says what we know for each gate.  

\begin{figure}
\begin{center}
\def\cs{\hspace{6pt}}
\begin{tabular}{c | *{15}{c@{\cs}} c}
\hline \hline
x & 0000 & 0001 & 0010 & 0011 & 0100 & 0101 & 0110 & 0111 & 1000 & 1001 & 1010 & 1011 & 1100 & 1101 & 1110 & 1111 \\
$x_1 \wedge \overline{x_3}$ & 0 & 0 & 0 & 0 & 0 & 0 & 0 & 0 & 1 & 1 & 0 & 0 & 1 & 1 & 0 & 0 \\
$x_1 \wedge x_2$ & 0 & 0 & 0 & 0 & 0 & 0 & 0 & 0 & 0 & 0 & 0 & 0 & 1 & 1 & 1 & 1 \\
\hline \hline
\end{tabular}
\end{center}
\caption{Function numbering scheme.  We number each four-bit function by considering its sixteen-bit truth table as the binary representation of an integer.  For example, $f(x_1, x_2, x_3, x_4) = x_1 \wedge \overline{x_3}$ is numbered $204 = 4+8+64+128$, while $x_1 \wedge x_2$ is numbered $15 = 1+2+4+8$.  These two functions are equivalent, so in the table below we only list number $15$, the smaller of the two.} \label{f:functionnumberingscheme}
\end{figure}

\begin{definition}
Two $k$-bit boolean functions $f_1, f_2 : \{0, 1\}^k \rightarrow \{0, 1\}$ are equivalent if there exists a string $y \in \{0, 1\}^k$ and a $k$-element permutation $\sigma \in S_k$ such that either for all $x \in \{0, 1\}^k$ $f_1(x) = f_2(\sigma(x) \oplus y)$, or for all $x \in \{0, 1\}^k$ $f_1(x) = \neg f_2(\sigma(x) \oplus y)$.  Here for a bit string $x = x_1 x_2 \ldots x_k \in \{0, 1\}^k$, $\sigma(x) \in \{0, 1\}^k$ is defined as $x_{\sigma^{-1}(1)} x_{\sigma^{-1}(2)} \ldots x_{\sigma^{-1}(k)}$.
\end{definition}

That is, two functions are
equivalent if they differ by permuting the inputs, and complementing a subset of the input bits and output bit.  There are 222 inequivalent four-bit functions.  We number each function for reference by considering its sixteen-bit truth table as the binary representation of an integer as in \figref{f:functionnumberingscheme}.  

The table begins with the fourteen inequivalent functions of \figref{f:threebitgates} that depend on at most three input bits.  The 208 functions that depend on all four input bits are first sorted according to their polynomial degree~\cite{Ambainis06polynomial} and then by their non-negative-weight and general adversary bounds, $\ADV$ and $\ADV^\pm$.  If these two bounds are equal, then we state only the first of the two.  The numerical adversary bounds were taken from \cite{HoyerLeeSpalek07negativeadvurl}.  

For each function, we give its $\{ \text{AND, OR, NOT} \}$ formula size.  Note that by \lemref{t:parityorwsize}, a function of formula size $k$ has a trivial span program with witness size $\sqrt{k}$.  The table includes comments on some of the functions.  These comments might include, for example, the smallest-size $\{ \text{AND, OR, NOT} \}$ formula for an equivalent function.  
For some of the functions, e.g., $\H{7128}$, we have listed in the Status column the witness size for the best span program we have found that computes the function.  
Those functions for which we know an optimal span program, with witness size matching the adversary bound, are marked with a checkmark ($\checkmark$), together with a brief justification.  For example, a reference to \lemref{t:parityorwsize} means that the optimal span program follows from composing two simpler span programs based on an OR or PARITY gate.  The lower bound for such functions follows from a simple fact on composition of adversary lower bounds: 

\begin{lemma} \label{t:parityoradv}
For two functions $f_1, f_2$ on disjoint inputs, $\ADV(f_1 \oplus f_2) = \ADV(f_1) + \ADV(f_2)$, and $\ADV(f_1 \vee f_2) = \sqrt{\ADV(f_1)^2 + \ADV(f_2)^2}$.  
\end{lemma}

\noindent
For a function $f$ marked ``opt.~NAND," the optimal span program comes from optimizing the edge weights of the tree corresponding to a minimal-size NAND formula computing $f$.  

Listed below the table are the optimal adversary matrices (\defref{t:adversarydef}) and optimal span programs for all the four-bit functions for which we know an optimal span program, except when one bound or the other follows from other stated results, e.g., from \lemref{t:parityoradv} or \lemref{t:parityorwsize}.  In a Mathematica file included as an electronic supplement to this article, we have included code to verify the adversary bound and span program witness size calculations.

For details on any of the span programs referenced with witness size that does not match the adversary lower bound, e.g., for function $\H{7128}$, please contact one of the authors.  

\def\C{$\checkmark$}
\setlength\LTleft{0pt} 
\setlength\LTright{0pt}
\renewcommand{\AA}[1]{#1&--}
\newcommand{\N}[1]{\raisebox{15pt}[0pt][0pt]{\hyperdef{#1}{}{}}#1}	

\vspace{.5in}	
Functions depending on up to three input bits:
\input{f_3bits.tex}

Functions depending on four input bits, with polynomial degree either two or three:
\input{f_4bits_degree3.tex}

Functions depending on four input bits, with polynomial degree four:
\input{f_4bits_degree4.tex}

\smallskip

We now list optimal adversary matrices and optimal span programs for the four-bit functions for which we know an optimal span program.  For the adversary matrices, we include only the rows and columns with nonzero entries.  

Several of the four-bit functions, for example $\H{287}$, $\MAJ(x_1, x_2, x_3 \wedge x_4)$, or $\H{1647}$, $\MAJ(x_1, x_2, x_3 \oplus x_4)$, are compositions of a three-bit function with a two-bit function.  Therefore, to consolidate cases and increase the generality of our results, it will sometimes be convenient to give span programs for three-bit functions with unbalanced input complexities.  \lemref{t:maj3unbalanced}, for example, implies optimal span programs for both functions $\H{287}$ and $\H{1647}$.  For the adversary lower bound, it is also convenient to define the general adversary bound with costs~\cite{HoyerLeeSpalek07negativeadv, HoyerLeeSpalek05compose}:

\begin{definition}[Adversary bound with costs] \label{t:adversaryboundwithcostsdef}
Let $f: \{0,1\}^k \rightarrow \{0,1\}$, and let $\alpha \in {\bf R}_+^n$ be a vector of positive reals.  Define 
\begin{equation} \label{e:adversaryboundwithcostsdef}
\ADV_\alpha^\pm(f) = \max_{\substack{\Gamma \neq 0}} \frac { \norm{\Gamma} }
  { \max_i \frac1{\alpha_i} \norm{ \Gamma \circ D_i } }
 \enspace ,
\end{equation}
where the matrices $D_i$ are defined as in \defref{t:adversarydef}.  
The maximum is over all $2^k \times 2^k$ nonzero, symmetric matrices $\Gamma$ satisfying $\bra{x} \Gamma \ket{y} = 0$ if $f(x) = f(y)$.  
\end{definition}

This weighted version of the adversary bound composes nicely, as shown by the following generalization of \thmref{t:weakadversarycomposition}, still a special case of \cite[Theorem~13]{HoyerLeeSpalek07negativeadv}: \comment{Check final version theorem no.}

\begin{theorem}[{\cite{HoyerLeeSpalek07negativeadv}}] \label{t:adversarycomposition}
Let $f = g \circ (h_1, \ldots, h_k)$ and let $\alpha = (\ADV^\pm(h_1), \ldots, \ADV^\pm(h_k))$.  Then $\ADV^\pm(f) \geq \ADV_\alpha^\pm(f)$.  
\end{theorem}

\noindent
\cite[Theorem~12]{HoyerLeeSpalek07negativeadv} shows that the weighted version of the \emph{nonnegative} adversary bound composes exactly, i.e., the $\geq$ sign can be replaced with equality, but equality is not known to hold for the general adversary bound.

\thmref{t:adversarycomposition} lets us compute the adversary bound for the four-bit function $\H{287}$, $\MAJ(x_1, x_2, x_3 \wedge x_4)$, for example, by considering $2^3 \times 2^3$ adversary matrices for the three-bit majority function with costs $\alpha = (1, 1, \sqrt{2})$.  

\def\operatorname#1{\mathop{\mathrm{#1}}\nolimits}  

\subsection{Function \texorpdfstring{\H{831}}{$\#831$}, {$\MAJ$}, with partly unbalanced inputs:}


If the first two inputs have equal costs $1$ and the third input has cost $\beta$, then the optimal adversary matrix for function $\H{831}$, $\MAJ$, comes from 
\[
\Gamma = \;
\setcounter{sprows}{3}
\setcounter{spcols}{3}
\setlength{\spheight}{6pt * \value{sprows} + 4pt}
\setlength{\spraise}{-16pt}
\begin{array}{r @{} r @{} *{\value{spcols}}{@{\ }c} @{} l}
	&\spleft[4pt]{}&011&101&110&\spright{}\\
100	&&0	&\beta&1	\\
010	&&\beta&0&1	\\
001	&&1	&1&0	
\end{array}
\]
The adversary bound is $\ADV_{(1, 1, \beta)}^\pm(\MAJ) = \tfrac12 \big(\sqrt{8 + \beta^2} + \beta \big)$.  
An optimal span program for this function with equal first two input complexities is given in \lemref{t:maj3unbalanced}.  Note that this also settles the complexities of functions $\H{287}$ and $\H{1647}$.  

\subsection{Function \texorpdfstring{\H{975}}{$\#975$} with partly unbalanced inputs:}

Function $\H{975}$ is $(x_3 \wedge x_2) \vee (\overline{x_3} \wedge x_1)$.  If the first two inputs have equal costs $1$ and the third input has cost $\beta$, then the optimal adversary matrix comes from 
\[
\Gamma = \;
\setcounter{sprows}{2}
\setcounter{spcols}{2}
\setlength{\spheight}{6pt * \value{sprows} + 4pt}
\setlength{\spraise}{-16pt}
\begin{array}{r @{} r @{} *{\value{spcols}}{@{\ }c} @{} l}
	&\spleft[4pt]{}&010&101&\spright{}\\
100	&&1	&\beta	\\
011	&&\beta&1	
\end{array}
\]
The adversary bound is $\ADV_{(1,1,\beta)}^\pm(\H{975}) = \beta + 1$.  A span program with matching witness size is 
\[
\setcounter{sprows}{3}
\setlength{\spheight}{6pt * \value{sprows} + 4pt}
\setlength{\spraise}{6pt}
\begin{array}{r@{}c@{}l r@{}*{3}{c@{\ }}c@{}l}
& & & X_J = ( & \{x_1\} & \{x_3\} & \{\overline{x_3}\} & \{x_2\} & ) \\
\noalign{\smallskip}
\spleft{t=} & 1 & \spright{,} & \spleft[2pt]{v_J=}
             & 0 & 1 & 1 & 0 & \spright{\enspace .} \\
& 0 & & & 1 & 1 & 0 & 0 & \\
& 0 & & & 0 & 0 & 1 & 1 & 
\end{array}
\]
This function can be evaluated as an ``if-then-else" statement, by evaluating the third input $x_3$ at cost $\beta$ and then one of $x_1$ or $x_2$ at cost 1.  A span program with witness size $\beta + 1$ is thus not surprising.  
Note that this also settles the complexities of functions $\H{495}$ and $\H{1695}$.  

\comment{For $\MAJ$, all three input bits were symmetrical, so $\ADV_{(1, 1, \beta)}^\pm(\MAJ) = \ADV_{(1, \beta, 1)}^\pm(\MAJ) = \ADV_{(\beta, 1, 1)}^\pm(\MAJ)$.  That will also hold for functions $\H{960}$ and $\H{828}$ below.  However, for this function, $\H{975}$, the third input bit is not symmetrical to the first two, and we do not have an expression for $\ADV_{(1, \beta, 1)}^\pm(\H{975})$.  It may be that the adversary and nonnegative adversary bounds differ in this case, for $\beta \neq 1$.  TODO: Using sedumi, compute adversary and nonnegative adversary bounds for this case.}

\subsection{Function \texorpdfstring{\H{960}}{$\#960$}, $\EQUAL_3$, with partly unbalanced inputs:}

For function $\H{960}$, $\EQUAL_3$, the adversary bound when the first two inputs have equal costs $1$ and the third input has cost $\beta$ is 
\[
\ADV_{(1,1,\beta)}^\pm(\EQUAL_3) = \left\{\begin{array}{c l} 
\beta + \sqrt{2-\beta^2} & \text{if $0 < \beta \leq \sqrt{2/5}$} \\
\sqrt{\frac{3}{2}(2+\beta^2)} & \text{if $\sqrt{2/5} \leq \beta \leq 2$} \\
\beta + 1 & \text{if $\beta \geq 2$}
\end{array}\right . \enspace .
\]
Indeed, optimal adversary matrices are
\[
\setcounter{sprows}{2}
\setcounter{spcols}{2}
\setlength{\spheight}{6pt * \value{sprows} + 4pt}
\setlength{\spraise}{-16pt}
\begin{array}{r @{} r @{} *{\value{spcols}}{@{\ }c} @{} l}
	&\spleft[4pt]{}&000&111&\spright{}\\
001	&&\beta	&1		&\\
110	&&1		&\beta	&
\end{array}
\enspace ,
\qquad\qquad
\setcounter{sprows}{6}
\setcounter{spcols}{2}
\setlength{\spheight}{6pt * \value{sprows} + 4pt}
\setlength{\spraise}{-16pt}
\begin{array}{r @{} r @{} *{\value{spcols}}{@{\ }c} @{} l}
	&\spleft[4pt]{}&000&111&\spright{}\\
100	&&2\alpha	&\alpha	&\\
010	&&2\alpha	&\alpha	&\\
001	&&2		&1		&\\
011	&&\alpha	&2\alpha	&\\
101	&&\alpha	&2\alpha	&\\
110	&&1		&2		&
\end{array}
\enspace
\qquad\text{and}\qquad
\setcounter{sprows}{4}
\setcounter{spcols}{2}
\setlength{\spheight}{6pt * \value{sprows} + 4pt}
\setlength{\spraise}{-16pt}
\begin{array}{r @{} r @{} *{\value{spcols}}{@{\ }c} @{} l}
	&\spleft[4pt]{}&000&111&\spright{}\\
100	&&\sqrt{2-\beta^2}	&\beta	&\\
010	&&\sqrt{2-\beta^2}	&\beta	&\\
011	&&\beta			&\sqrt{2-\beta^2}	&\\
101	&&\beta			&\sqrt{2-\beta^2}	&
\end{array}
\]
for the ranges $0 < \beta \leq \sqrt{2/5}$, $\sqrt{2/5} \leq \beta \leq 2$ and $\beta \geq 2$, respectively, where $\alpha = \sqrt{\frac{4 - \beta^2}{5\beta^2 - 2}}$.

The optimal span program is 
\[
\setcounter{sprows}{3}
\setlength{\spheight}{6pt * \value{sprows} + 4pt}
\setlength{\spraise}{6pt}
\begin{array}{r@{}c@{}l r@{}*{3}{c@{\ }}c@{}l}
& & & X_J = ( & \{x_3\} & \{x_1, x_2\} & \{\overline{x_1}, \overline{x_2}\} & \{\overline{x_3}\} & ) \\
\noalign{\smallskip}
\spleft{t=} & 1 & \spright{,} & \spleft[2pt]{v_J=}
             & 0 & 1 & 1 & 0 & \spright{\enspace ,} \\
& 0 & & & w & 1 & 0 & 0 & \\
& 0 & & & 0 & 0 & 1 & w & 
\end{array}
\]
where $w = \bigg( \frac{\beta + \sqrt{2-\beta^2}}{2(1-\beta^2)} \bigg)^{1/2}$ for $\beta \leq \sqrt{2/5}$, $w = 1/\sqrt{\beta}$ for $\sqrt{2/5} \leq \beta \leq 2$, and $w = 1/\sqrt{2}$ for $\beta \geq 2$.  

Note that this also settles the complexities of functions $\H{424}$, $\EQUAL_3(x_1, x_2, x_3 \wedge x_4)$, and $\H{1680}$, $\EQUAL_3(x_1, x_2, x_3 \oplus x_4)$.  

\subsection{Function \texorpdfstring{$\H{963}$}{$\#963$} with partly unbalanced inputs:}   

See \claimref{t:gwsize} for the balanced complexity case for function $\H{963}$, $(x_1 \wedge x_2 \wedge x_3) \vee (\overline{x_1} \wedge \overline{x_2})$.  If the first two, symmetrical inputs have equal costs $1$ and the third input has cost $\beta$, the the optimal adversary matrix comes from
\[
\Gamma = \;
\setcounter{sprows}{3}
\setcounter{spcols}{3}
\setlength{\spheight}{6pt * \value{sprows} + 4pt}
\setlength{\spraise}{-16pt}
\begin{array}{r @{} r @{} *{\value{spcols}}{@{\ }c} @{} l}
	&\spleft[4pt]{}&000&001&111&\spright{}\\
011	&&0	&1	&1	\\
101	&&0	&1	&1	\\
110	&&1	&0	&\beta	
\end{array}
\]
The adversary bound is $\ADV_{(1,1,\beta)}^\pm(\H{963}) = \frac1{\sqrt 2} \big( 5 + \beta^2 + \sqrt{9 + 2 \beta^2 + \beta^4} \big)^{1/2}$.  A span program with matching witness size is 
\[
\setcounter{sprows}{2}
\setcounter{spcols}{3}
\setlength{\spheight}{6pt * \value{sprows} + 4pt}
\setlength{\spraise}{6pt}
\begin{array}{r@{}c@{}l r@{}*{2}{c@{\ }}c@{}l}
& & & X_J = ( & \{x_1\} & \{x_2, x_3\} & \{\overline{x_2}, \overline{x_3}\} & ) \\
\noalign{\smallskip}
\spleft{t=} & 1 & \spright{,} & \spleft[2pt]{v_J=} & 0 & w_1 & 1 & \spright{\enspace ,} \\
& 0 & & & 1 & w_2 & 0 & 
\end{array}
\]
where $w_1 = \frac12 \big( 1 + \beta^2 + \sqrt{9 + 2\beta^2 + \beta^4} \big)^{1/2}$ and $w_2 = \frac1{\sqrt{2} \beta} \big( -3 + \beta^2 + \sqrt{9 + 2\beta^2 + \beta^4} \big)^{1/2}$.  Note that this also settles the complexities of functions $\H{2032}$, $\H{426}$ and $\H{1776}$. 

\subsection{Function \texorpdfstring{\H{828}}{$\#828$}, $\EXACT_\mathrm{2\, of\, 3}$, with partly unbalanced inputs:}

Function $\H{828}$ is $\text{EXACT}_\text{2 of 3}(x_1, x_2, x_3) = \MAJ(x_1, x_2, x_3) \wedge (\overline{x_1} \vee \overline{x_2} \vee \overline{x_3})$.  If the first two inputs have equal costs $1$ and the third input has cost $\beta$, then the optimal adversary matrix comes from 
\[
\Gamma = \;
\setcounter{sprows}{4}
\setcounter{spcols}{3}
\setlength{\spheight}{6pt * \value{sprows} + 4pt}
\setlength{\spraise}{-16pt}
\begin{array}{r @{} r @{} *{\value{spcols}}{@{\ }c} @{} l}
	&\spleft[4pt]{}&110&101&011&\spright{}\\
001	&&0		&1		&1		&\\
010	&&1		&0		&\beta	&\\
100	&&1		&\beta	&0		&\\
111	&&\beta	&1		&1		&
\end{array}
\]
The adversary bound is $\ADV_{(1, 1, \beta)}^\pm(\H{828}) = \big( 3 + \beta^2 + \sqrt{1+8 \beta^2} \big)^{1/2}$.  
A span program with matching witness size is
\[
\setcounter{sprows}{3}
\setcounter{spcols}{5}
\setlength{\spheight}{6pt * \value{sprows} + 4pt}
\setlength{\spraise}{6pt}
\begin{array}{r@{}c@{}l r@{}*{4}{c@{\ }}c@{}l}
& & & X_J = ( & \{\overline{x_1}, \overline{x_2}\} & \{\overline{x_3}\} & \{x_1\} & \{x_2\} & \{x_3\} & ) \\
\noalign{\smallskip}
\spleft{t=} & 1 & \spright{,} & \spleft[2pt]{v_J=} 
		&1		&0	&w_2	&w_2	&w_1/\sqrt{2}	& \spright{\enspace ,} \\
& 0 & & 	&0		&0	&i		&-i		&w_1		& \\
& 0 & & 	&w_1	&1	&0		&0		&0			&
\end{array}
\]
where $w_1^2 = \big( \sqrt{ 8 \beta^2 + 1 } - 1 \big) / (2 \beta)$ and $w_2^2 = 1 / \big( \sqrt{8 \beta^2 + 1} + 3 \big)$.  
Note that this also settles the complexities of functions $\H{1912}$, $\H{286}$ and $\H{1641}$.

\subsection{Function \texorpdfstring{\H{393}}{$\#393$}:}

The optimal adversary matrix and a matching span program for function $\H{393}$ are
\[
\setcounter{sprows}{4}
\setlength{\spheight}{6pt * \value{sprows} + 4pt}
\setlength{\spraise}{6pt}
\begin{array}{r @{} r @{} *{4}{@{\ }c} @{} l}
	&&0010&1101&0111&1011&\\
0001	&\spleft[4pt]{}&1&3	&0	&0	&\spright{}\\
1110	&&3	&1	&0	&0	&\\
0011	&&2\sqrt 2&0	&3	&3	&\\
1111	&&0	&2\sqrt 2	&3	&3	&
\end{array}
\qquad\qquad
\setcounter{sprows}{3}
\setlength{\spheight}{6pt * \value{sprows} + 4pt}
\setlength{\spraise}{6pt}
\begin{array}{r@{}c@{}l r@{}*{3}{c@{\ }}c@{}l}
& & & X_J = ( & \{x_1, x_2\} & \{x_3\} & \{\overline{x_1}, \overline{x_2}\} & \{x_4\} & ) \\
\noalign{\smallskip}
\spleft{t=} & 1 & \spright{,} & \spleft[2pt]{v_J=}
  & 1 & 0 & 1 & 0 & \spright{\enspace .} \\
& 0 & & & 1 & \sqrt{3/2} & 0 & 0 & \\
& 0 & & & 0 & 0 & 1 & \sqrt{3/2} & 
\end{array}
\]

\subsection{Function \texorpdfstring{\H{989}}{$\#989$}:}

The optimal adversary matrix and a matching span program for function $\H{989}$ are
\[
\setcounter{sprows}{3}
\setlength{\spheight}{6pt * \value{sprows} + 4pt}
\setlength{\spraise}{6pt}
\begin{array}{r @{} r @{} *{3}{@{\ }c} @{} l}
	&&0010&0111&1100&\\
0011	&\spleft[4pt]{}&1&1	&0	&\spright{}\\
0110	&&1	&1	&1	&\\
1101	&&0	&1	&1	&
\end{array}
\qquad\qquad
\setcounter{sprows}{3}
\setlength{\spheight}{6pt * \value{sprows} + 4pt}
\setlength{\spraise}{6pt}
\begin{array}{r@{}c@{}l r@{}*{5}{c@{\ }}c@{}l}
& & & X_J = ( & \{\overline{x_1}, x_2\} & \emptyset & \{\overline{x_2}\} & \{\overline{x_3}\} & \{ x_4 \} & \{ \overline{x_4} \} ) \\
\noalign{\smallskip}
\spleft{t=} & 1 & \spright{,} & \spleft[2pt]{v_J=} 
		&0		&1	&1		&1		&2^{1/4}	&1	& \spright{\enspace .} \\
& 0 & & 	&0 		&1	&2^{-1/4}	&2^{-1/4}	&1		&1	& \\
& 0 & & 	&2^{1/4}	&1	&0		&0		&0		&0	&
\end{array}
\]
\comment{Note that the entries in the first row should all be divided by $2^{1/4}-1$ in order to balance out the true and false span program complexities.  Unlike most span programs in this appendix, this one does not come from optimizing a NAND formula expansion.  Although this adversary matrix corresponds to the function computed by the span program, that function is only equivalent to $\H{989}$.}

\subsection{Function \texorpdfstring{\H{1968}}{$\#1968$}:}

The optimal adversary matrix and a matching span program for function $\H{1968}$ are
\[
\setcounter{sprows}{6}
\setlength{\spheight}{6pt * \value{sprows} + 4pt}
\setlength{\spraise}{-17pt}
\begin{array}{r @{} r @{} *{4}{@{\ }c} @{} l}
	&\spleft[4pt]{}&0000&1110&0011&1101&\spright{}\\
0010	&&3\sqrt 7&2\sqrt 7&3\sqrt 7	&0		\\
1100	&&2\sqrt 7&3\sqrt 7&0		&3\sqrt 7	\\
1001	&&5	&0	&2	&7	\\
0101&&5	&0	&2	&7	\\
0111&&0	&5	&7	&2	\\
1011&&0	&5	&7	&2	
\end{array}
\qquad
\setcounter{sprows}{3}
\setlength{\spheight}{22pt} 
\setlength{\spraise}{6pt}
\begin{array}{r@{}c@{}l r@{}*{5}{c@{\ }}c@{}l}
& & & X_J = ( & \{x_1, x_2\} & \{x_3\} & \{x_4\} & \{\overline{x_1}, \overline{x_2}\} & \{ \overline{x_3} \} & \{ x_4 \} ) \\
\noalign{\smallskip}
\spleft{t=} & 1 & \spright{,} & \spleft[2pt]{v_J=} & 1 & 0 & 0 & 1 & 0 & 0 & \spright{\enspace .} \\
& 0 & & & 1 & \sqrt{3}/2 & \sqrt{3}/2 & 0 & 0 & 0 & \\
& 0 & & & 0 & 0 & 0 & 1 & \sqrt{3}/2 & \sqrt{3}/2 &
\end{array}
\]

\subsection{Function \texorpdfstring{\H{1910}}{$\#1910$}:}

The optimal adversary matrix and a matching span program for function $\H{1910}$ 
are
\[
\setcounter{sprows}{4}
\setcounter{spcols}{7}
\setlength{\spheight}{6pt * \value{sprows} + 4pt}
\setlength{\spraise}{-16pt}
\begin{array}{r @{} r @{} *{\value{spcols}}{@{\ }c} @{} l}
	&\spleft[4pt]{}&0001&0010&0100&1000&0011&1100&1111&\spright{}\\
0111	&&0	&0	&1	&0	&1	&0	&1\\
1011	&&0	&0	&0	&1	&1	&0	&1\\
1101	&&1	&0	&0	&0	&0	&1	&1\\
1110&&0	&1	&0	&0	&0	&1	&1
\end{array}
\qquad
\setcounter{sprows}{1}
\setlength{\spheight}{6pt * \value{sprows} + 4pt}
\setlength{\spraise}{6pt}
\begin{array}{r@{}c@{}l r@{}*{2}{c@{\ }}c@{}l}
& & & X_J = ( & \{x_1, x_2, x_3, x_4\} & \{ \overline{x_1}, \overline{x_2}\} & \{\overline{x_3}, \overline{x_4}\} ) \\
\noalign{\smallskip}
\spleft{t=} & 1 & \spright{,} & \spleft[2pt]{v_J=} & \sqrt2 & 1 & 1 & \spright{\enspace .}
\end{array}
\]

\subsection{Function \texorpdfstring{\H{317}}{$\#317$}:}

The optimal adversary matrix for function $\H{317}$, $(x_1 \wedge x_2 \wedge x_3) \vee \MAJ(\overline{x_2}, \overline{x_3}, x_4)$, comes from 
\[
\Gamma = \;
\setcounter{sprows}{4}
\setcounter{spcols}{3}
\setlength{\spheight}{6pt * \value{sprows} + 4pt}
\setlength{\spraise}{-16pt}
\begin{array}{r @{} r @{} *{\value{spcols}}{@{\ }c} @{} l}
	&\spleft[4pt]{}&0111&1010&1100&\spright{}\\
1000	&&0	&1	&1\\
0011	&&1	&1	&0\\
0101	&&1	&0	&1\\
1110&&1	&1	&1
\end{array}
\]
A matching span program follows from \lemref{t:parityorwsize} applied to the span programs for AND$_3$ and $\MAJ$.  

\subsection{Function \texorpdfstring{\H{5790}}{$\#5790$}:}

The optimal adversary matrix and a matching span program for function $\H{5790}$ are 
\[
\setcounter{sprows}{4}
\setcounter{spcols}{4}
\setlength{\spheight}{6pt * \value{sprows} + 4pt}
\setlength{\spraise}{-16pt}
\begin{array}{r @{} r @{} *{\value{spcols}}{@{\ }c} @{} l}
	&\spleft[4pt]{}&0011&0100&1000&1111&\spright{}\\
0000	&&1	&1	&1	&0\\
0111	&&1	&1	&0	&1\\
1011	&&1	&0	&1	&1\\
1100&&0	&1	&1	&1
\end{array}
\qquad\quad
\setcounter{sprows}{3}
\setlength{\spheight}{6pt * \value{sprows} + 4pt}
\setlength{\spraise}{6pt}
\begin{array}{r@{}c@{}l r@{}*{5}{c@{\ }}c@{}l}
& & & X_J = ( & \{x_1\} & \{ x_2, \overline{x_3} \} & \{ \overline{x_2}, x_3 \} & \{ \overline{x_1} \} & \{ x_2, x_4 \} & \{ \overline{x_2}, \overline{x_3} \} & ) \\
\noalign{\smallskip}
\spleft{t=} & 1 & \spright{,} & \spleft[2pt]{v_J=} 		
		&1	&0	&0	&1	&0	&0	&\spright{\enspace .} \\
& 0 & & 	&1	&1	&1	&0	&0	&0	\\
& 0 & & 	&0	&0	&0	&1	&1	&1
\end{array}
\]

%

\subsection{Function \texorpdfstring{\H{385}}{$\#385$}:}

The optimal adversary matrix and a matching span program for function $\H{385}$ are
\[
\setcounter{sprows}{4}
\setcounter{spcols}{3}
\setlength{\spheight}{6pt * \value{sprows} + 4pt}
\setlength{\spraise}{-16pt}
\begin{array}{r @{} r @{} *{\value{spcols}}{@{\ }c} @{} l}
	&\spleft[4pt]{}&0111&1000&1111&\spright{}\\
0000	&&\sqrt2&\sqrt2&0	\\
1001	&&0	&\sqrt2&1	\\
1010	&&0	&\sqrt2&1	\\
1100&&0	&\sqrt2&1	
\end{array}
\qquad\qquad
\setcounter{sprows}{2}
\setlength{\spheight}{6pt * \value{sprows} + 4pt}
\setlength{\spraise}{6pt}
\begin{array}{r@{}c@{}l r@{}*{2}{c@{\ }}c@{}l}
& & & X_J = ( & \{x_2, x_3, x_4\} & \{x_1\} & \{\overline{x_2}, \overline{x_3}, \overline{x_4}\} & ) \\
\noalign{\smallskip}
\spleft{t=} & 1 & \spright{,} & \spleft[2pt]{v_J=} & 1 & 0 & \frac{(61+7\sqrt{73})^{1/4}}{2^{3/4}\sqrt{3}} & \spright{\enspace .} \\
& 0 & & & 0 & \frac{(49+5\sqrt{73})^{1/4}}{2^{3/4}\sqrt{3}} & 1
\end{array}
\]

\subsection{Function \texorpdfstring{\H{279}}{$\#279$}, $\text{Threshold}_\mathrm{2\,of\,4}$:}

The optimal adversary matrix for function $\H{279}$, $\text{Threshold}_\text{2 of 4}$, is 
\[
\Gamma = \;
\setcounter{sprows}{6}
\setcounter{spcols}{4}
\setlength{\spheight}{6pt * \value{sprows} + 4pt}
\setlength{\spraise}{-16pt}
\begin{array}{r @{} r @{} *{\value{spcols}}{@{\ }c} @{} l}
	&\spleft[4pt]{}&0111&1011&1101&1110&\spright{}\\
0011	&&1	&1	&0	&0\\
0101	&&1	&0	&1	&0\\
0110	&&1	&0	&0	&1\\
1100&&0	&0	&1	&1\\
1010&&0	&1	&0	&1\\
1001&&0	&1	&1	&0
\end{array}
\]
A matching span program was given in \exampleref{t:threshold2of4example}.  

\subsection{Function \texorpdfstring{\H{1639}}{$\#1639$}:}

The optimal adversary matrix and a matching span program for function $\H{1639}$, equivalent to the formula $(x_1 \wedge x_2) \vee \MAJ(\overline{x_1} \wedge \overline{x_2}, x_3, x_4)$, are

\[
\setcounter{sprows}{7}
\setcounter{spcols}{5}
\setlength{\spheight}{6pt * \value{sprows} + 4pt}
\setlength{\spraise}{-16pt}
\begin{array}{r @{} r @{} *{\value{spcols}}{@{\ }c} @{} l}
	&\spleft[4pt]{}&0000&0101&1001&0110&1010&\spright{}\\
0001	&&1	&1	&1	&0	&0	\\
0010	&&1	&0	&0	&1	&1	\\
0111	&&0	&1	&0	&1	&0	\\
1011&&0	&0	&1	&0	&1	\\
1100&&1	&0	&0	&0	&0	\\
1101&&0	&1	&1	&0	&0	\\
1110&&0	&0	&0	&1	&1	
\end{array}
\qquad\qquad
\setcounter{sprows}{2}
\setlength{\spheight}{6pt * \value{sprows} + 4pt}
\setlength{\spraise}{6pt}
\begin{array}{r@{}c@{}l r@{}*{3}{c@{\ }}c@{}l}
& & & X_J = ( & \{x_1, x_2\} & \{\overline{x_1}, \overline{x_2}\} & \{x_3\} & \{x_4\} & ) \\
\noalign{\smallskip}
\spleft{t=} & 1 & \spright{,} & \spleft[2pt]{v_J=} & 1 & \sqrt{5}/2 & 1/2 & 1/2 & \spright{\enspace .} \\
& 0 & & & 0 & 1 & i & -i
\end{array}
\]
This span program is based on the one in \lemref{t:maj3unbalanced}.  

\subsection{Function \texorpdfstring{\H{6014}}{$\#6014$}, $\EXACT_\mathrm{2\, or\, 3\, of \, 4}$:}

The optimal adversary matrix and a matching span program for function $\H{6014}$ are
\begin{gather*}
\Gamma = \;
\setcounter{sprows}{6}
\setcounter{spcols}{5}
\setlength{\spheight}{6pt * \value{sprows} + 4pt}
\setlength{\spraise}{-16pt}
\begin{array}{r @{} r @{} *{\value{spcols}}{@{\ }c} @{} l}
	&\spleft[4pt]{}&1000&0100&0010&0001&1111&\spright{}\\
0011	&&0	&0	&1	&1	&3^{-1/2}\\
0101	&&0	&1	&0	&1	&3^{-1/2}\\
0110	&&0	&1	&1	&0	&3^{-1/2}\\
1100&&1	&1	&0	&0	&3^{-1/2}\\
1010&&1	&0	&1	&0	&3^{-1/2}\\
1001&&1	&0	&0	&1	&3^{-1/2}
\end{array}
\\
\\
\setcounter{sprows}{4}
\setlength{\spheight}{6pt * \value{sprows} + 4pt}
\setlength{\spraise}{6pt}
\begin{array}{r@{}c@{}l r@{}*{11}{c@{\ }}c@{}l}
& & & X_J = ( & \{x_1\} & \{x_1\} & \{x_2\} & \{x_2\} & \{x_3\} & \{x_3\} & \{x_4\} & \{x_4\} & \{\overline{x_1}\} & \{\overline{x_2}\} & \{\overline{x_3}\} & \{\overline{x_4}\} & ) \\
\noalign{\smallskip}
\spleft{t=}	& 1 & \spright{,} & \spleft[2pt]{v_J=} 
				&0	&0	&0	&0	&0	&0	&0	&0	&2^{-1/4}	&2^{-1/4}	&2^{-1/4}	&2^{-1/4}	&\spright{\enspace .} \\
		& 0 & &	&1	&1	&1	&1	&1	&1	&1	&1	&\sqrt{3}	&\sqrt{3}	&\sqrt{3}	&\sqrt{3}	&\\
		& 0 & &	&1	&1	&1	&-1	&i	&-i	&i	&i	&0		&0		&0		&0		&\\
		& 0 & & 	&i	&-i	&i	&i	&1	&1	&1	&-1	&0		&0		&0		&0		&
\end{array}
\end{gather*}
This span program is based on the expansion $\EXACT_\text{2 or 3 of 4} = \text{Threshold}_\text{2 of 4}(x_1, x_2, x_3, x_4) \wedge (\overline{x_1} \vee \overline{x_2} \vee \overline{x_3} \vee \overline{x_4})$, and the span program for $\text{Threshold}_\text{2 of 4}$ $\H{279}$ from \exampleref{t:threshold2of4example}.  

\subsection{Function \texorpdfstring{\H{278}}{$\#278$}, $\EXACT_\mathrm{1\, of\, 4}$:}

The optimal adversary matrix for function $\H{278}$, $\EXACT_\text{1 of 4} = \text{Threshold}_\text{3 of 4}(\overline{x_1},\overline{x_2},\overline{x_3},\overline{x_4}) \wedge (x_1 \vee x_2 \vee x_3 \vee x_4)$, is 
\[
\Gamma = \;
\setcounter{sprows}{7}
\setcounter{spcols}{4}
\setlength{\spheight}{6pt * \value{sprows} + 4pt}
\setlength{\spraise}{-16pt}
\begin{array}{r @{} r @{} *{\value{spcols}}{@{\ }c} @{} l}
	&\spleft[4pt]{}&0111&1011&1101&1110&\spright{}\\
0011	&&1	&1	&0	&0	\\
0101	&&1	&0	&1	&0	\\
0110	&&1	&0	&0	&1	\\
1100&&0	&0	&1	&1	\\
1010&&0	&1	&0	&1	\\
1001&&0	&1	&1	&0	\\
1111&&1	&1	&1	&1	
\end{array}
\]
A span program with matching witness size $\sqrt{10}$ follows from \lemref{t:parityorwsize} applied to the span programs for OR$_4$ and $\text{Threshold}_\text{3 of 4}$ $\H{279}$.  

\subsection{Function \texorpdfstring{\H{5736}}{$\#5736$}, $\EXACT_\mathrm{2\, of\, 4}$:}

The optimal adversary matrix for function $\H{278}$, $\EXACT_\text{1 of 4} = \text{Thr.}_\text{2 of 4}(x_1,\ldots,x_4) \wedge \text{Thr.}_\text{2 of 4}(\overline{x_1},\ldots,\overline{x_4})$, is 
\[
\Gamma = \;
\setcounter{sprows}{6}
\setcounter{spcols}{8}
\setlength{\spheight}{6pt * \value{sprows} + 4pt}
\setlength{\spraise}{-16pt}
\begin{array}{r @{} r @{} *{\value{spcols}}{@{\ }c} @{} l}
	&\spleft[4pt]{}&1000&0100&0010&0001&0111&1011&1101&1110&\spright{}\\
0011	&&0	&0	&1	&1	&1	&1	&0	&0\\
0101	&&0	&1	&0	&1	&1	&0	&1	&0\\
0110	&&0	&1	&1	&0	&1	&0	&0	&1\\
1100&&1	&1	&0	&0	&0	&0	&1	&1\\
1010&&1	&0	&1	&0	&0	&1	&0	&1\\
1001&&1	&0	&0	&1	&0	&1	&1	&0
\end{array}
\]
A span program with matching witness size $\sqrt{12}$ follows from \lemref{t:parityorwsize} applied to the span program for $\text{Threshold}_\text{2 of 4}$ $\H{279}$.

\comments{
\item BUG: Hyperlinks are linking to the entry \emph{after} the correct entry, e.g., \H{15555}.  FIXED. -BR
\item Note that capitalization of Threshold$_\text{2 of 4}$ is inconsistent with EXACT$_\text{2 of 3}$.  It also is italicized and bolded differently (compare to $\EXACT$ or $\MAJ$ or $\EQUAL$).
\item For which functions should we include the minimum-size formula, e.g., $\H{5737}$?
\item Settle on format for $\ADV = \ADV^\pm$ entries in the table.
\item In the table for function $\H{963}$, we have permuted the bits to correspond to the stated formula.
\item I have not tried to balance out the true and false cases in these span programs.  Our code does it automatically.
}

%% file: functiontable/f_3bits.tex
\begin{longtable}{c c c c r@{\hspace{2pt}}l X}
\hline \hline
$\#$ & Size & $\ADV$ & $\ADV^\pm$ & \multicolumn{2}{c}{Status} & Comments \\
\hline
\endfirsthead
$\#$ & Size & $\ADV$ & $\ADV^\pm$ & \multicolumn{2}{c}{Status} & Comments \\
\hline
\endhead
\N{0}	&0	&$(1/2)\sqrt{13+\sqrt{73}}$&2.84354	&\C&status			&constant 0\kill		
\N{0}		&0	&\AA{0}					&\C&trivial					&constant 0\\
\N{255}	&1	&\AA{1}					&\C&trivial					&$x_1$\\
\N{15}	&2	&\AA{$\sqrt 2$}				&\C&\lemref{t:parityorwsize}	&$x_1 \vee x_2$\\
\N{3}	&3	&\AA{$\sqrt 3$}					&\C&\lemref{t:parityorwsize}	&$x_1 \vee x_2 \vee x_3$\\
\N{63}	&3	&\AA{$\sqrt 3$}				&\C&\lemref{t:parityorwsize}	&$x_1 \vee (x_2 \wedge x_3)$\\
\N{4080}	&4	&\AA{2}					&\C&\lemref{t:parityorwsize}	&$x_1 \oplus x_2$\\
\N{975}	&4	&\AA{2}					&\C&\lemref{t:parityorwsize}	&$(x_3 \wedge x_2) \vee (\overline{x_3} \wedge x_1)$\\
\N{831}	&5	&\AA{2}					&\C&\claimref{t:majwsize}		&$\MAJ$\\
\N{960}	&6	&\AA{$3/\sqrt{2}$}			&\C&\claimref{t:equalwsize}	&$\EQUAL_3$\\
\N{963}	&5	&\AA{$\sqrt{3+\sqrt 3}$}		&\C&\claimref{t:gwsize}		&$(x_1 \wedge x_2 \wedge x_3) \vee (\overline{x_1} \wedge \overline{x_2})$\\
\N{60}	&5	&\AA{$\sqrt 5$}				&\C&\lemref{t:parityorwsize}	&$x_1 \vee (x_2 \oplus x_3)$\\
\N{1020}	&6	&\AA{$1+\sqrt 2$}			&\C&	\lemref{t:parityorwsize}	&$x_1 \oplus (x_2 \wedge x_3)$\\
\N{828}	&8	&\AA{$\sqrt 7$}				&\C&	\lemref{t:parityorwsize}	&$\EXACT_\text{2 of 3}$\\
\N{15555}&10	&\AA{3}					&\C&\lemref{t:parityorwsize}	&$x_1 \oplus x_2 \oplus x_3$\\
\hline \hline
\end{longtable}

%% file: functiontable/f_4bits_degree3.tex
\begin{longtable}{c c c c r@{\hspace{2pt}}l X}
\hline \hline
$\#$ & Size & $\ADV$ & $\ADV^\pm$ & \multicolumn{2}{c}{Status} & Comments \\
\hline
\endfirsthead
$\#$ & Size & $\ADV$ & $\ADV^\pm$ & \multicolumn{2}{c}{Status} & Comments \\
\hline
\endhead
\N{15555}&10	&$(1/2)\sqrt{13+\sqrt{73}}$&2.84354	&\C&status			&constant 0\kill		
{\N{7128}}&10	&2.50000		&2.51353		&&2.77394			&sorted input bits~\cite{Ambainis06polynomial}, $(x_1 \wedge ((x_2 \wedge x_3) \vee (\overline{x_3} \wedge \overline{x_4}))) \vee (\overline{x_1} \wedge ((\overline{x_2} \wedge \overline{x_3}) \vee (x_3 \wedge x_4)))$\\
{\N{863}}	&5	&2.00000		&2.07136		&&2.22833			&monotone two adjacent 1s, $(x_1 \wedge x_2) \vee (x_4 \wedge (x_1 \vee x_3))$\\
\N{427}	&5	&2.18398		&2.20814		&&2.22833			&$\H{975}(x_1 \wedge x_2, x_3, x_4)$\\
\N{27}	&5	&\AA{$\sqrt 5$}				&\C&\lemref{t:parityorwsize}&$x_1 \wedge \H{975}(x_2, x_3, x_4)$\\
\N{393}	&6	&\AA{$4/\sqrt 3$}			&\C&	opt.~NAND	&$(x_1 \wedge x_2 \wedge x_3) \vee (\overline{x_1} \wedge \overline{x_2} \wedge x_4)$\\
\N{383}	&7	&2.30278		&2.34406		&&$\sqrt{4+\sqrt 3}$		&$(x_1 \wedge x_2 \wedge x_3) \vee ((x_1 \vee x_2 \vee x_3) \wedge x_4)$\\
\N{126}	&7	&\AA{$\sqrt{11/2}$}			&\C&\lemref{t:parityorwsize}&$x_1 \wedge \neg \EQUAL_3(x_2, x_3, x_4)$\\
\N{24}	&7	&\AA{$\sqrt{11/2}$}			&\C&\lemref{t:parityorwsize}&$x_1 \wedge \EQUAL_3(x_2, x_3, x_4)$\\
\N{303}	&6	&\AA{2.35829}				&&$1+\sqrt 2$			&$((x_1 \vee x_2) \wedge x_3) \vee ((\overline{x_1} \vee x_2) \wedge x_4)$, span program size 5\\
\N{495}	&6	&\AA{$1+\sqrt 2$}			&\C&	opt.~NAND		&$\H{975}(x_1, x_2, x_3 \wedge x_4)$\\
\N{989}	&6	&\AA{$1+\sqrt 2$}			&\C&opt.~gadget		&$(x_1 \wedge x_2 \wedge x_3) \vee (\overline{x_1} \wedge (\overline{x_2} \vee x_4))$\\
\N{965}	&7	&2.41531		&2.42653		&&2.59234			&$(x_1 \wedge x_2 \wedge x_3) \vee (x_4 \wedge \overline{x_3}) \vee (\overline{x_4} \wedge \overline{x_2})$\\
\N{987}	&7	&2.43128		&2.44711		&&					&\\
\N{424}	&8	&\AA{$\sqrt 6$}				&\C&	opt.~NAND		&$\EQUAL_3(x_1, x_2, x_3 \wedge x_4)$\\
\N{2034}	&7	&2.47323		&2.48653		&&					&\\
\N{429}	&7	&2.48837		&2.50826		&&					&\\
\N{490}	&8	&2.52207		&2.52326		&&					&\\
\N{281}	&8	&2.51464		&2.54019		&&					&\\
\N{2022}	&8	&2.55719		&2.58189		&&					&\\
\N{1968}	&8	&\AA{$2\sqrt{5/3}$}			&\C&	opt.~NAND		&$(x_1 \wedge x_2 \wedge (x_3 \vee x_4)) \vee (\overline{x_1} \wedge \overline{x_2} \wedge (\overline{x_3} \vee x_4))$\\
\N{1782}	&7	&2.56155		&2.59040		&&2.61804			&$\H{975}(x_1 \oplus x_2, x_3, x_4)$\\
\N{984}	&8	&2.58539		&2.60282		&&					&\\
\N{1973}	&8	&2.61704		&2.63510		&&					&\\
\N{1910}	&8	&\AA{$\sqrt 7$}				&\C&	opt.~NAND		&$(x_1 \wedge x_2 \wedge x_3 \wedge x_4) \vee (\overline{x_1} \wedge \overline{x_2}) \vee (\overline{x_3} \wedge \overline{x_4})$\\
\N{317}	&8	&\AA{$\sqrt 7$}				&\C&\lemref{t:parityorwsize}&$(x_1 \wedge x_2 \wedge x_3) \vee \MAJ(\overline{x_2}, \overline{x_3}, x_4)$\\
{\N{858}}	&8	&2.64575		&2.64658		&&$\sqrt{5+2\sqrt 2}$	&$(x_1 \wedge x_2 \wedge x_3) \vee (\overline{x_3} \wedge \overline{x_4}) \vee (\overline{x_1} \wedge (\overline{x_2} \vee \overline{x_3}))$\\
\N{300}	&9	&2.69932		&2.70595		&&					&\\
\N{6030}	&10	&\AA{2.70928}				&&3					&$(x_1 \wedge (x_2 \vee x_3) \wedge (\overline{x_2} \vee x_4)) \vee (\overline{x_1} \wedge (\overline{x_2} \vee x_3) \wedge (x_2 \vee \overline{x_4}))$, classical three-query alg.\\
\N{894}	&10	&2.70808		&2.71982		&&					&\\
\N{1714}	&10	&2.75018		&2.75944		&&					&\\
\N{1980}	&9	&2.75779		&2.77469		&&					&\\
\N{1719}	&9	&2.76916		&2.78319		&&					&\\
\N{366}	&10	&2.76569		&2.80341		&&					&\\
\N{6042}	&10	&2.84104		&2.84923		&&					&\\
\N{1716}	&10	&2.86854		&2.88186		&&					&\\
\N{1680}	&12	&\AA{3}					&\C&opt.~NAND		&$\EQUAL_3(x_1, x_2, x_3 \oplus x_4)$\\
\N{1695}	&10	&\AA{3}					&\C&opt.~NAND		&$\H{975}(x_1, x_2, x_3 \oplus x_4)$, tight classical alg.\\
\N{5790}	&10	&\AA{3}					&\C&	opt.~NAND		&$(x_1 \wedge (x_2 \vee x_3) \wedge (\overline{x_2} \vee x_4)) \vee (\overline{x_1} \wedge (x_2 \vee \overline{x_3}) \wedge (\overline{x_2} \vee x_3))$, tight classical alg.\\
\N{7140}	&10	&\AA{3}					&\C&\lemref{t:parityorwsize}&$x_1 \oplus \H{975}(x_2, x_3, x_4)$, tight classical alg.\\
\N{1683}	&11	&3.00283		&3.01009		&&					&\\
\N{5766}	&12	&3.02533		&3.03595		&&					&\\
\N{6627}	&12	&3.01470		&3.04933		&&					&\\
\N{5783}	&12	&3.03917		&3.07294		&&					&\\
\N{6375}	&14	&\AA{$1+\sqrt{9/2}$}		&\C&\lemref{t:parityorwsize}&$x_1 \oplus \EQUAL_3(x_2, x_3, x_4)$\\
\hline \hline
\end{longtable}

%% file: functiontable/f_4bits_degree4.tex
\begin{longtable}{c c c c r@{\hspace{2pt}}l X}
\hline \hline
$\#$ & Size & $\ADV$ & $\ADV^\pm$ & \multicolumn{2}{c}{Status} & Comments \\
\hline
\endfirsthead
$\#$ & Size & $\ADV$ & $\ADV^\pm$ & \multicolumn{2}{c}{Status} & Comments \\
\hline
\endhead
\N{15555}&10	&$(1/2)\sqrt{13+\sqrt{73}}$&2.84354	&\C&status			&constant 0\kill		
\N{1}		&4	&\AA{2}					&\C&\lemref{t:parityorwsize}&$x_1 \wedge x_2 \wedge x_3 \wedge x_4$\\
\N{127}	&4	&\AA{2}					&\C&	\lemref{t:parityorwsize}&$x_1 \vee (x_2 \wedge x_3 \wedge x_4)$\\
\N{31}	&4	&\AA{2}					&\C&	\lemref{t:parityorwsize}&$x_1 \wedge (x_2 \vee (x_3 \wedge x_4))$\\
\N{7}		&4	&\AA{2}					&\C&\lemref{t:parityorwsize}&$x_1 \wedge x_2 \wedge (x_3 \vee x_4)$\\
\N{855}	&4	&\AA{2}					&\C&\lemref{t:parityorwsize}&$(x_1 \vee x_3) \wedge (x_2 \vee x_4)$\\ 
{\N{319}}	&6	&2.17533		&2.20453		&&2.27731			&$(x_1 \wedge (x_2 \vee x_3)) \vee (x_2 \wedge x_3 \wedge x_4)$\\
\N{431}	&5	&2.20635		&2.20721		&&2.22835			&$((x_1 \vee x_2) \wedge x_3) \vee (\overline{x_1} \wedge x_4)$\\
\N{23}	&6	&\AA{$\sqrt 5$}				&\C&\lemref{t:parityorwsize}&$x_1 \wedge \MAJ(x_2, x_3, x_4)$\\
\N{399}	&6	&2.22595		&2.25274		&&					&\\
\N{287}	&6	&\AA{$\sqrt{3+\sqrt 5}$}		&\C&	\lemref{t:maj3unbalanced}&$\MAJ(x_1, x_2, x_3 \wedge x_4)$\\
\N{384}	&8	&\AA{$4/\sqrt 3$}			&\C&	\claimref{t:equalwsize}&$\EQUAL_4$\\
\N{447}	&7	&2.30278		&2.31707		&&					&\\
\N{385}	&7	&\AA{$(1/2)\sqrt{13+\sqrt{73}}$}&\C&opt.~NAND		&$(x_1 \vee (x_2 \wedge x_3)) \wedge (\overline{x_1} \vee (\overline{x_2} \wedge \overline{x_3} \wedge x_4))$\\
\N{395}	&6	&2.29062		&2.32499		&&					&\\
\N{2032}	&6	&\AA{$\sqrt{(7+\sqrt{17})/2}$}	&\C&	opt.~NAND		&$\H{963}(x_1, x_2, x_3 \vee x_4)$\\
\N{426}	&6	&\AA{$\sqrt{(7+\sqrt{17})/2}$}	&\C&	opt.~NAND		&$\H{963}(x_1, x_2, x_3 \wedge x_4)$\\
\N{967}	&6	&2.36698		&2.37004		&&					&\\
\N{25}	&6	&\AA{$\sqrt{4+\sqrt 3}$}		&\C&\lemref{t:parityorwsize}&$x_1 \wedge \H{963}(x_2, x_3, x_4)$\\
\N{61}	&6	&\AA{$\sqrt{4+\sqrt 3}$}		&\C&	\lemref{t:parityorwsize}&$x_1 \vee \H{963}(x_2, x_3, x_4)$\\
\N{981}	&7	&2.39489		&2.40490		&&					&\\
\N{283}	&7	&2.40091		&2.42364		&&					&\\
\N{983}	&6	&2.42266		&2.42424		&&					&\\
\N{409}	&7	&2.39120		&2.42693		&&					&\\
\N{961}	&7	&2.43531		&2.43869		&&					&\\
\N{111}	&6	&\AA{$\sqrt 6$}				&\C&\lemref{t:parityorwsize}&$x_1 \vee \H{60}(x_2, x_3, x_4)$\\
\N{1638}	&6	&\AA{$\sqrt 6$}				&\C&\lemref{t:parityorwsize}&$(x_1 \wedge x_2) \vee (x_3 \oplus x_4)$\\
\N{279}	&8	&\AA{$\sqrt 6$}				&\C&	\exampleref{t:threshold2of4example}&$\text{Threshold}_\text{2 of 4}$\\
\N{6}		&6	&\AA{$\sqrt 6$}				&\C&\lemref{t:parityorwsize}&$x_1 \wedge x_2 \wedge (x_3 \oplus x_4)$\\
\N{387}	&7	&2.48662		&2.50251		&&					&\\
\N{859}	&7	&2.49857		&2.50575		&&					&\\
\N{988}	&7	&2.53791		&2.53835		&&					&\\
\N{411}	&8	&2.52192		&2.54009		&&					&\\
\N{1654}	&7	&2.51758		&2.54347		&&					&\\
\N{445}	&9	&2.53019		&2.55017		&&					&\\
\N{425}	&7	&\AA{2.55654}				&&2.55654			&$\H{963}(x_1 \wedge x_2, x_3, x_4)$, exact exp.~for $\ADV$?\\
\N{494}	&7	&\AA{2.55654}				&&2.55654			&$\H{963}(x_1 \vee x_2, x_3, x_4)$, exact exp.~for $\ADV$?\\
\N{2018}	&8	&2.54502		&2.55711		&&					&\\
\N{2016}	&8	&2.55128		&2.55874		&&					&\\
\N{2033}	&8	&\AA{2.56155}				&&2.59163			&$((x_1 \vee x_2) \wedge x_3 \wedge x_4) \vee (x_2 \wedge x_3) \vee (\overline{x_3} \wedge \overline{x_4})$\\
\N{491}	&7	&2.56388		&2.57901		&&2.64302			&$(x_1 \wedge x_2 \wedge x_3) \vee (x_4 \wedge (\overline{x_1} \vee (\overline{x_2} \wedge \overline{x_3})))$\\
\N{879}	&8	&2.57641		&2.58289		&&					&\\
\N{415}	&8	&2.57096		&2.58436		&&					&\\
\N{428}	&8	&2.59386		&2.60618		&&					&\\
\N{430}	&7	&2.60716		&2.60975		&&					&\\
\N{30}	&7	&\AA{$\sqrt{4+2\sqrt 2}$}		&\C&\lemref{t:parityorwsize}&$x_1 \wedge \H{1020}(x_2, x_3, x_4)$\\
\N{2019}	&8	&2.61439		&2.62037		&&					&\\
\N{488}	&10	&\AA{2.62705}				&&					&\\
\N{1639}	&8	&\AA{$\sqrt 7$}				&\C&opt.~gadget		&$(x_1 \wedge x_2) \vee \MAJ(\overline{x_1} \wedge \overline{x_2}, x_3, x_4)$\\  
\N{1969}	&9	&2.64898		&2.65285		&&					&\\
\N{1778}	&8	&2.66716		&2.66873		&&					&\\
\N{985}	&8	&2.65949		&2.67406		&&					&\\
\N{980}	&9	&2.67345		&2.67735		&&					&\\
\N{1650}	&8	&2.67869		&2.68369		&&					&\\
\N{391}	&9	&2.68828		&2.70027		&&					&\\
\N{878}	&9	&2.68845		&2.70057		&&					&\\
\N{1918}	&10	&\AA{2.70131}				&&					&$(\overline{x_1} \wedge x_2 \wedge (x_3 \vee x_4)) \vee (x_1 \wedge \neg \EQUAL_4(x_2, x_3, x_4))$\\
\N{386}	&9	&2.67082		&2.70300		&&					&\\
\N{966}	&8	&2.70246		&2.70500		&&					&\\
\N{367}	&8	&2.70387		&2.70585		&&					&\\
\N{1662}	&9	&2.69544		&2.70815		&&					&\\
\N{1776}	&8	&\AA{$\sqrt{(9+\sqrt{33})/2}$}	&\C&opt.~NAND		&$\H{963}(x_1, x_2, x_3 \oplus x_4)$\\ 
\N{280}	&10	&2.71411		&2.71639		&&					&\\
\N{408}	&9	&2.69951		&2.71811		&&					&\\
\N{301}	&8	&2.71328		&2.71856		&&					&\\
\N{1647}	&8	&\AA{$1+\sqrt 3$}			&\C&	\lemref{t:maj3unbalanced}&$\MAJ(x_1, x_2, x_3 \oplus x_4)$\\
\N{2040}	&8	&\AA{$1+\sqrt 3$}			&\C&\lemref{t:parityorwsize}&$x_1 \oplus (x_2 \wedge (x_3 \vee x_4))$\\
\N{510}	&8	&\AA{$1+\sqrt 3$}			&\C&\lemref{t:parityorwsize}&$x_1 \oplus (x_2 \wedge x_3 \wedge x_4)$\\
\N{1634}	&8	&2.74013		&2.74148		&&					&\\
\N{856}	&9	&2.73510		&2.74171		&&					&\\
\N{892}	&9	&2.72608		&2.74205		&&					&\\
\N{316}	&9	&2.74228		&2.74790		&&					&\\
\N{862}	&9	&2.74628		&2.75157		&&					&\\
\N{829}	&9	&2.75490		&2.76384		&&					&\\
\N{990}	&8	&2.76066		&2.76706		&&					&\\
\N{1651}	&8	&2.75952		&2.76736		&&					&\\
\N{1715}	&9	&2.76490		&2.77389		&&					&\\
\N{489}	&9	&\AA{2.78575}				&&2.86182			&$(x_1 \wedge x_2 \wedge x_3) \vee (x_4 \wedge (\overline{x_2} \vee \overline{x_3}) \wedge (\overline{x_1} \vee (\overline{x_2} \wedge \overline{x_3})))$, $x_2$ and $x_3$ symmetrical\\
\N{6040}	&10	&2.77499		&2.79246		&&					&\\
\N{1914}	&9	&\AA{2.79485}				&&2.85539			&$(x_1 \wedge x_2 \wedge (x_3 \vee x_4)) \vee (x_3 \wedge x_4) \vee (\overline{x_1} \wedge \overline{x_3} \wedge \overline{x_4})$, $x_3$ and $x_4$ symmetrical\\
\N{282}	&9	&\AA{2.80369}				&&2.92535			&$(\overline{x_1} \wedge \overline{x_3} \wedge \overline{x_4}) \vee (x_3 \wedge x_4) \vee ((x_1 \vee x_2) \wedge (x_3 \vee x_4))$, $x_2$ and $x_3$ symmetrical\\
\N{1972}	&9	&2.79678		&2.80499		&&					&\\
\N{893}	&9	&2.79694		&2.80607		&&					&\\
\N{444}	&10	&2.80787		&2.81297		&&					&\\
\N{874}	&10	&2.81477		&2.81815		&&					&\\
\N{107}	&9	&\AA{$2 \sqrt 2$}			&\C&\lemref{t:parityorwsize}&$x_1 \wedge \H{828}(x_2, x_3, x_4)$\\
\N{1632}	&8	&\AA{$2 \sqrt 2$}			&\C&\lemref{t:parityorwsize}&$(x_1 \oplus x_2) \wedge (x_3 \oplus x_4)$\\
\N{22}	&9	&\AA{$2 \sqrt 2$}			&\C&	\lemref{t:parityorwsize}&$x_1 \wedge \H{828}(x_2, x_3, x_4)$\\
\N{6014}	&12	&\AA{$2 \sqrt 2$}			&\C&opt.~gadget&$\EXACT_\text{2 or 3 of 4} = \text{Thr.}_\text{2 of 4}(x_1,x_2,x_3,x_4) \wedge (\overline{x_1} \vee \overline{x_2} \vee \overline{x_3} \vee \overline{x_4})$\\
\N{854}	&8	&\AA{$2 \sqrt 2$}			&\C&\lemref{t:parityorwsize}&$(x_1 \wedge x_2) \oplus (x_3 \wedge x_4)$\\
\N{6060}	&10	&2.82034		&2.83150		&&					&\\
\N{875}	&9	&2.83428		&2.83570		&&					&\\
\N{2017}	&10	&2.83417		&2.83655		&&					&\\
\N{1712}	&10	&2.84346		&2.84354		&&					&\\
\N{857}	&9	&2.82909		&2.85105		&&					&\\
\N{1718}	&9	&2.84413		&2.85900		&&					&\\
\N{382}	&11	&\AA{2.87999}				&&					&\\
\N{362}	&11	&2.88004		&2.88205		&&					&\\
\N{5758}	&11	&\AA{2.89586}				&&					&\\
\N{876}	&10	&2.89815		&2.90403		&&					&\\
\N{318}	&10	&2.90163		&2.90404		&&					&\\
\N{407}	&11	&2.89638		&2.90417		&&					&\\
\N{410}	&10	&2.90592		&2.91505		&&					&\\
\N{446}	&10	&2.91254		&2.91560		&&					&\\
\N{1974}	&10	&2.91560		&2.91835		&&					&\\
\N{363}	&10	&2.92040		&2.92652		&&					&\\
\N{1658}	&10	&2.92940		&2.93141		&&					&\\
\N{5774}	&11	&2.92267		&2.93717		&&					&\\
\N{1717}	&10	&2.93360		&2.94668		&&					&\\
\N{1777}	&10	&\AA{2.96176}				&&					&\\
\N{982}	&10	&2.96425		&2.96696		&&					&\\
\N{1635}	&10	&2.98641		&2.98673		&&					&\\
\N{6120}	&12	&\AA{3}					&\C&\lemref{t:parityorwsize}&$x_1 \oplus \MAJ(x_2, x_3, x_4)$\\
\N{1713}	&11	&2.99622		&3.00474		&&					&\\
\N{1912}	&10	&\AA{$\sqrt{5+\sqrt{17}}$}	&\C&opt.~NAND		&$\H{828}(x_1 \vee x_2, x_3, x_4)$\\
\N{286}	&10	&\AA{$\sqrt{5+\sqrt{17}}$}	&\C&	opt.~NAND		&$\H{828}(x_1 \wedge x_2, x_3, x_4)$\\
\N{5786}	&11	&3.01265		&3.02051		&&					&\\
\N{1687}	&11	&3.01018		&3.02207		&&					&\\
\N{1681}	&12	&3.02473		&3.02629		&&					&\\
\N{360}	&12	&3.04017		&3.04042		&&					&\\
\N{1659}	&11	&3.04139		&3.04288		&&					&\\
\N{6625}	&12	&3.02185		&3.04627		&&					&\\
\N{5820}	&11	&3.03542		&3.04710		&&					&\\
\N{1725}	&12	&3.04048		&3.05100		&&					&\\
\N{877}	&11	&3.04498		&3.05270		&&					&\\
\N{390}	&12	&3.03755		&3.05354		&&					&\\
\N{872}	&12	&3.06480		&3.06823		&&					&\\
\N{5782}	&11	&3.06111		&3.07244		&&					&\\
\N{5784}	&12	&3.07314		&3.07400		&&					&\\
\N{5742}	&12	&3.09004		&3.09058		&&					&\\
\N{1686}	&11	&\AA{3.11060}				&&					&$\H{963}(x_1 \oplus x_2, x_3, x_4)$\\
\N{5804}	&12	&3.11134		&3.11717		&&					&\\
\N{2025}	&12	&3.12714		&3.12849		&&					&\\
\N{6038}	&12	&3.12842		&3.12999		&&					&\\
\N{414}	&12	&3.12805		&3.13416		&&					&\\
\N{5767}	&13	&3.13610		&3.13705		&&					&\\
\N{361}	&11	&\AA{3.14864}				&&					&\\
\N{5787}	&12	&3.14761		&3.15425		&&					&\\
\N{105}	&11	&\AA{$\sqrt{10}$}			&\C&\lemref{t:parityorwsize}&$x_1 \wedge (x_2 \oplus x_3 \oplus x_4)$\\
\N{278}	&12	&\AA{$\sqrt{10}$}			&\C&	\lemref{t:parityorwsize}&$\EXACT_\text{1 of 4} = \text{Thr.}_\text{3 of 4}(\overline{x_1},\overline{x_2},\overline{x_3},\overline{x_4}) \wedge (x_1 \vee x_2 \vee x_3 \vee x_4)$\\
\N{1656}	&12	&3.16284		&3.16420		&&					&\\
\N{6630}	&12	&\AA{$1+\sqrt{3+\sqrt 3}$}	&\C&\lemref{t:parityorwsize}&$x_1 \oplus \H{963}(x_2, x_3, x_4)$\\
\N{1721}	&12	&3.19509		&3.19570		&&					&\\
\N{1913}	&12	&\AA{3.19640}				&&					&\\
\N{5763}	&13	&3.21393		&3.21633		&&					&\\
\N{5771}	&13	&3.21305		&3.21888		&&					&\\
\N{1643}	&12	&\AA{3.21930}				&&					&\\
\N{1633}	&12	&\AA{3.23607}			&&3.33513			&$((x_1 \oplus x_2) \wedge (x_3 \oplus x_4)) \vee (x_1 \wedge x_2 \wedge x_3 \wedge x_4)$\\ 
\N{1785}	&12	&\AA{$1+\sqrt 5$}			&\C&	\lemref{t:parityorwsize}&$x_1 \oplus \H{60}(x_2, x_3, x_4)$\\
\N{873}	&12	&3.26876		&3.27185		&&					&\\
\N{5785}	&12	&3.27183		&3.27189		&&					&\\
\N{5738}	&14	&\AA{3.28207}				&&					&\\
\N{5805}	&14	&3.34542		&3.34781		&&					&\\
\N{5761}	&15	&\AA{3.36028}				&&					&$\EXACT_\text{1 or 4 of 4}$\\
\N{406}	&14	&3.36637		&3.37384		&&					&\\
\N{1657}	&14	&3.39009		&3.39051		&&					&\\
\N{5769}	&14	&\AA{3.39400}				&&					&\\
\N{7905}	&12	&\AA{$2+\sqrt 2$}			&\C&\lemref{t:parityorwsize}&$x_1 \oplus x_2 \oplus (x_3 \wedge x_4)$\\
\N{5736}	&16	&\AA{$2\sqrt 3$}			&\C&	\lemref{t:parityorwsize}&$\EXACT_\text{2 of 4} = \text{Thr.}_\text{2 of 4}(x_1,x_2,x_3,x_4) \wedge \text{Thr.}_\text{2 of 4}(\overline{x_1}, \overline{x_2}, \overline{x_3}, \overline{x_4})$\\
\N{5801}	&14	&3.51041		&3.51129		&&					&\\
\N{5739}	&15	&\AA{3.51414}				&&					&\\
\N{1641}	&14	&\AA{$\sqrt{7+\sqrt{33}}$}	&\C&	opt.~NAND		&$\H{828}(x_1 \oplus x_2, x_3,x_4)$\\
\N{5865}	&16	&\AA{$1+\sqrt 7$}			&\C&	\lemref{t:parityorwsize}&$x_1 \oplus \H{828}(x_2,x_3,x_4)$\\
\N{5737}	&16	&\AA{3.78478}				&&					&$\EXACT_\text{2 or 4 of 4}$
\rem{$((\overline{x_4} \wedge x_1) \vee (\overline{x_3} \wedge \overline{x_2}) \vee (((x_3 \wedge x_2) \vee x_4) \wedge \overline{x_1})) \wedge ((\overline{x_3} \wedge x_2) \vee (\overline{x_4} \wedge \overline{x_1}) \vee (((x_4 \wedge x_1) \vee x_3) \wedge \overline{x_2}))$}\\
\N{27030}	&16	&\AA{4}					&\C&\lemref{t:parityorwsize}&PARITY$_4$, $x_1 \oplus x_2 \oplus x_3 \oplus x_4$\\
\hline \hline
\end{longtable}